\documentclass[hidelinks,10pt]{article}

\usepackage{authblk}
\usepackage[utf8]{inputenc} 
\usepackage[T1]{fontenc}    
\usepackage{url}            
\usepackage{booktabs}       
\usepackage{amsfonts}       
\usepackage{nicefrac}       
\usepackage{microtype}      
\usepackage{lipsum}		
\usepackage{graphicx}
\usepackage{doi}
\usepackage{amsmath, amsthm, amssymb}
\usepackage{enumitem}
\usepackage{mathrsfs}

\usepackage{multirow}
\usepackage{footnote}
\usepackage{float}
\usepackage{comment}
\usepackage{xcolor}
\usepackage{amssymb}
\usepackage{rotating}
\usepackage{natbib}
\usepackage{hyperref}
\hypersetup{colorlinks, citecolor={blue!70!black}, linkcolor={blue!70!black}, urlcolor={blue!80!black}}
\usepackage{titlesec}

\usepackage{listings}
\usepackage[verbose=true,letterpaper]{geometry}
\AtBeginDocument{
  \newgeometry{
    textheight=9in,
    textwidth=6.5in,
    top=1in,
    headheight=14pt,
    headsep=25pt,
    footskip=28pt
  }
}

\usepackage{setspace}

\newcommand{\bE}{\mathbb E}
\newcommand{\bV}{\mathbb V}
\newcommand{\bI}{\mathbb I}
\newcommand{\Cov}{\mathrm{Cov}}

\newcommand{\corr}{\mathrm{corr}}

\newcommand{\OW}{\mbox{\tiny{OW}}}
\newcommand{\IPW}{\mbox{\tiny{IPW}}}

\newcommand{\ATT}{\mbox{\tiny{ATT}}}

\newcommand{\RCT}{\mbox{\tiny{RCT}}}

\newcommand{\SC}{\mbox{\tiny{Schoen}}}

\newcommand{\HL}{\mbox{\tiny{HL}}}

\usepackage{lineno}
\usepackage{adjustbox, multirow}

\theoremstyle{plain}
\newtheorem{theorem}{Theorem}
\newtheorem{lemma}{Lemma}

\newtheorem{proposition}{Proposition}
\newtheorem{corollary}{Corollary}

\theoremstyle{definition}

\theoremstyle{remark}

\makeatletter
\renewenvironment{proof}[1][]{%
  \par\pushQED{\qed}%
  \normalfont \topsep6\p@\@plus6\p@\relax
  \trivlist
  \item[\hskip\labelsep \itshape Proof\ifx&#1&\else~#1\fi\@addpunct{.}]\ignorespaces
}{%
  \popQED\endtrivlist\@endpefalse
}
\makeatother

\DeclareMathOperator\logit{\mathrm logit}
\allowdisplaybreaks

\title{Sample size and power calculations for causal inference with time-to-event outcomes}

\author[1]{Chengxin Yang}
\author[2]{Bo Liu}
\author[2]{Fan Li\thanks{Corresponding author (Fan Li, email: fl35@duke.edu). The authors gratefully acknowledge Jay Lusk, Fan Li (Yale), Steve Cole, Jing Ning, Mei-Cheng Wang, Ying Qing Chen for helpful comments.}}
\affil[1]{Department of Biostatistics and Bioinformatics, Duke University}
\affil[2]{Department of Statistical Science, Duke University}

\begin{document}

\def\spacingset#1{
  \renewcommand{\baselinestretch}{#1}\small\normalsize
  \setlength{\abovedisplayskip}{4pt}
  \setlength{\belowdisplayskip}{6pt}
  \setlength{\abovedisplayshortskip}{2pt}
  \setlength{\belowdisplayshortskip}{4pt}
}

\spacingset{1.5}

\renewenvironment{abstract}
  {\vspace{0pt}%
   {\centering\textbf{\large Abstract}\par}\vspace{8pt}%
   \normalfont\normalsize}
  {\par}

\date{}
\maketitle

\begin{abstract}

Power calculations for time-to-event outcomes in randomized trials need only a few summary inputs under proportional hazards. Power calculations in observational studies require additional adjustment for confounding, and most existing approaches require individual-level data or simulation from assumed data-generating models. This paper develops the first analytical power analysis formula for the hazard ratio estimand from a Cox proportional hazards model in observational studies. We derive the asymptotic robust variance of the propensity score weighted Cox partial-likelihood estimator and separate it into two parts: a component determined by canonical design inputs, and a remainder requiring survival and censoring curves and individual covariate data, which are unavailable at the design stage. Dropping the remainder gives the formula. Power determination for observational studies then needs only four inputs: the treatment proportion, targeted hazard ratio, event rates in the target population, and an overlap coefficient summarizing similarity between comparison groups. Treating randomized trials as a special limiting case, our formula recovers the classical exponential-model formula without its parametric assumptions, and provides bias correction to Schoenfeld's formula. We provide an online calculator and an R package \texttt{PSpower} to implement the method.
\end{abstract}

\vspace{36pt}


\newpage
\spacingset{1.53}

\section{Introduction}\label{sec:intro}

Power and sample size calculations are required in the design of most causal studies to ensure the planned sample size offers sufficient statistical power to detect a {targeted} treatment effect. The distinctive feature of power analysis is the common absence of individual-level data; investigators must work with, at best, a few summary quantities from prior studies. Power analysis thus typically relies on quantifying the variance of a chosen test statistic of the target parameter as a function of a few summary inputs. Such methods have been well established in randomized trials \citep{chow2017sample}. Observational studies have been increasingly used for causal inference, but corresponding power analysis methods remain limited. 
The central challenge to observational studies is confounding; adjusting for confounders inflates variances of effect estimates, and thus increases the sample size required for the same power.

This paper concerns power analysis for causal inference with time-to-event outcomes, ubiquitous in medical research.  The Cox proportional hazards model \citep{cox1972regression} is the most widely used method for analyzing time-to-event outcomes, with the hazard ratio as the primary estimand. {Similar to the case of non-censored outcomes, classical power formulas for time-to-event outcomes have been restricted to randomized trials. \cite{schoenfeld1983sample} derived an analytical formula under proportional hazards for the log-rank test that compares two survival curves, but the formula was not derived for the hazard ratio estimand and is inaccurate under unequal allocation \citep{Yung2020sample} and large effects \citep{lee2026asymptotic}. \cite{Rubinstein1981planning} and \cite{lachin1986evaluation} developed formulas for the Wald test based on estimated hazards, but assumed an exponential model with constant hazards. More importantly, directly applying these formulas to observational studies would ignore the variance inflation due to confounding adjustment and thus lead to underpowered studies.}

{Existing approaches for power analysis with observational data include two broad families.} The first quantifies the variance inflation relative to a randomized trial \citep{hsieh2000sample, scosyrev2019power, austin2021informing, shook2022power}, but generally relies on individual-level data that are rarely available in the design stage. The second simulates individual data \citep[e.g.,][]{mukherjee2003estimating, qin2024sample}, but usually relies on \emph{ad hoc} choices of data-generating models that are sensitive and subject to cherry-picking. For non-censored outcomes, \cite{liu2025sample} proposed a third approach: an analytical formula that needs only two inputs beyond the canonical ones in randomized trials; their method holds under mild modeling assumptions on the propensity score and outcome but places no distributional assumptions on the covariates. This approach is theoretically justified, easy to implement with a few interpretable inputs and avoids \emph{ad hoc} simulations. However, to our knowledge, no such formula exists for observational studies with time-to-event outcomes, which requires an entirely new analysis of the variance of an estimator of the Cox hazard ratio.

This paper develops the first analytical power analysis formula for the Wald test of the hazard ratio from a Cox model in observational studies with only a few summary inputs, with randomized trials as a limiting case. We derive the asymptotic robust variance of the propensity score weighted Cox partial-likelihood estimator, and separate it into a component determined by canonical summary inputs at the design stage and a remainder. The remainder depends on the survival and censoring curves over follow-up and on how the outcome varies with the covariates, which is unknown before data acquisition; dropping it gives a new analytical formula (Section~\ref{sec:sample-size}).
For observational studies, besides the canonical inputs, one additional input suffices: an overlap coefficient that summarizes the similarity between the propensity score distributions of the two groups (Section~\ref{sec:v-obs}). As the overlap between arms improves toward a randomized trial, our formula converges to the classical exponential-model formula without assuming any parametric survival distribution, and quantifies the error that Schoenfeld's formula can produce (Section~\ref{sec:v-rct}). We implement the proposed method in an R package \texttt{PSpower}, available at \url{https://CRAN.R-project.org/package=PSpower}, and an interactive online calculator, available at \url{https://www2.stat.duke.edu/~fl35/pspower.html}.

\section{{Setup and the weighted Cox analysis}}\label{sec:setup}

\subsection{{Notation and hypotheses}}\label{subsec:notation}

Consider a study with a right-censored time-to-event outcome of units indexed by $i=1, \ldots, n$, {with data analyzed up to a fixed follow-up time $t^\dagger<\infty$, measured from each unit's entry}. We adopt the potential outcome framework for causal inference. Let $Z_i$ be the treatment indicator, with $z=1$ for the treated and $0$ for the control, $T_i(z)$ the potential time of event of interest under treatment $z$, and $X_i$ baseline covariates. Denote $C_i$ the censoring time, {which covers censoring during follow-up and administrative censoring at a common calendar closing date}. Following \cite{lin1989robust}, each unit $i$ is an i.i.d. draw from the study population with distribution $P$ over $(T_i(0), T_i(1), C_i, Z_i, X_i)$. Let $r =\Pr(Z_i=1)$ be the \emph{treatment proportion}. We denote $T^*_i=T_i(Z_i)$ the event time corresponding to the observed treatment, $T_i = \min\{T^*_i, C_i, t^\dagger \}$ the observed time, $\delta_i = \mathbb{I}(T^*_i \!=\! T_i)$ the event indicator, and the at-risk process $Y_i(t)=\mathbb{I}(T_i\geq t)$.

Let {$F_z(t|x)=\Pr\{T_i(z) \!\geq\! t | X_i\!=\!x \}$ be the conditional potential survival function of treatment $z$ given $x$; averaging it over the distribution of $X_i$ in the study population gives $F_z(t)=\Pr\{T_i(z) \geq t\}$, with hazard function $\lambda_z(t)=-\{\log F_z(t)\}'$. Similarly, let $G_z(t | x)=\Pr(C_i \geq t\mid Z_i\!=\!z, X_i\!=\!x)$ and $d_z(x)=\Pr(\delta_i\!=\!1 \mid Z_i \!=\! z, X_i \!=\! x)$ be the conditional censoring distribution and event rate of arm $z$ given $x$; averaging them over the same distribution gives $G_z(t)$ and the \emph{standardized event rate} $d_z$, which gives the combined event rate $d=rd_1+(1-r)d_0$. In contrast, averaging $d_z(x)$ over the conditional distribution of $X_i$ given $Z_i=z$ gives the \emph{observed event rate} of arm $z$, $\bar{d}_z = \Pr(\delta_i=1 | Z_i=z)$, which generally differs from $d_z$ in observational studies; in randomized trials, the covariates of each arm follow the distribution of the full study sample, so $d_z=\bar{d}_z$.}

{We test the null hypothesis $H_0$ against the proportional hazards alternative $H_1$:
\begin{equation}\label{eq:hypotheses}
    H_0:\ F_1(t)=F_0(t) \quad\text{versus}\quad H_1:\ \lambda_1(t)=\lambda_0(t)\exp(\tau_0),
\end{equation}
for $t \in [0,t^\dagger]$. Under $H_1$, the study is designed to detect the log hazard ratio $\tau_0$, the estimand of a marginal Cox model, i.e., the Cox model with the treatment $Z$ as the only predictor, if $Z$ were randomized in the study population. This alternative places no parametric assumptions on the baseline hazard $\lambda_0$ or the censoring patterns. Some alternatives allow the hazard ratio to vary over follow-up, in which case study design generally requires specifying its trajectory, the control survival curve, and the censoring distributions, often through parametric models \citep[e.g.,][]{Barthel2006evaluation, Luo2019design, Yung2020sample}.}

{\subsection{Inverse probability weights and the Cox-Wald test}\label{subsec:ipw-test}}

{In an observational study, the treated and control arms differ in their distributions of confounding covariates, so any comparison of $F_1$ and $F_0$ must adjust for confounding. A popular approach is via the inverse probability weighted marginal Cox regression, which reweights both arms to the covariate distribution of the study population \citep{austin2017performance}.} Let $e(x)=\Pr(Z=1|X=x)$ be the propensity score \citep{rosenbaum1983central}. In randomized trials, $e_i=e(X_i) \equiv r$; in observational studies, $e_i$ is unknown and varies across $i$. Inverse probability weights assign $w_i=Z_iw_1(e_i)+(1-Z_i)w_0(e_i)$ to unit $i$, with $w_1(e)=r/e$ and $w_0(e)=(1-r)/(1-e)$, both normalized so that $w_i\equiv1$ in a randomized trial.

We maintain three standard identification assumptions for $\tau_0$: (A1) SUTVA, $T^*_i = Z_i T_i(1) + (1-Z_i)T_i(0)$; (A2) unconfoundedness, $\{T_i(0), T_i(1)\} \perp Z_i \mid X_i$; (A3) overlap, $0 < e(X_i) < 1$ for all $i$. In addition, we assume a simple censoring mechanism:
\begin{flalign*}
    & \text{(A4) Arm-specific independent censoring, }  {\{T_i(z), X_i\} \perp C_i \mid Z_i=z \text{ for } z=0,1.}
\end{flalign*}
{Common censoring is $G_1=G_0$. A4 is reasonable under administrative censoring, the most common type of censoring in practice. A4 can be relaxed to conditionally independent censoring, $T_i(z) \!\perp\! C_i \mid Z_i\!=\!z,X_i$, but it would require inverse probability of censoring weights and thus individual covariates \citep[e.g.,][]{Cheng2022survivals, luo2025doubly},  which are usually not available in power analysis.}

A consistent estimator $\widehat{\tau}_n$ to $\tau_0$ in the weighted Cox analysis is the solution to the weighted Cox partial-likelihood estimating equation \citep{cox1975partial, Binder1992CoxSurvey}:
\begin{equation}\label{eq:estimating-equation}
    U_n(\tau) = \sum_{i=1}^n w_i\delta_i \left\{Z_i - \frac{S_1^*(\tau,T_i)}{S_0^*(\tau,T_i)}\right\} = 0, \quad S_k^*(\tau,t)=n^{-1}\sum_{j=1}^n w_jY_j(t)\exp(\tau Z_j)Z_j^k,
\end{equation}
where $S_k^*(\tau,t)$, $k=0,1$, are averages of the weighted risk set at time $t$ under a log hazard ratio $\tau$, where each unit $j$ at risk is weighted by $w_j\exp(\tau Z_j)$; let $s_k(\tau,t)=\mathbb E\{S_k^*(\tau,t)\}$. 
{The treated arm's share of the weighted risk set, $\pi(\tau,t)=s_1/s_0$, is the probability that an event at $t$ is from the treated arm.} {Our power analysis plans for the \emph{Cox-Wald test} based on the statistic $n^{1/2}\widehat{\tau}_n/\widehat{V}^{1/2}$, with $\widehat{V}$ being the robust variance estimate for $\tau_0$ \citep{Binder1992CoxSurvey, Lin2000CoxSurvey}.} Denote $V$ the asymptotic variance of $n^{1/2}(\widehat{\tau}_n\!-\tau_0)$. Standard calculations show that, for a one-sided level-$\alpha$ test, {which rejects $H_0$ when $\mathrm{sgn}(\tau_0)\,n^{1/2}\widehat{\tau}_n/\widehat{V}^{1/2}>z_{1-\alpha}$,} the minimal sample size of \emph{units} to reach power $\beta$ is
\begin{equation}\label{eq:sample-size-general}
    N = (z_{1 - \alpha} + z_\beta)^2 {V}/ {\tau_0^2},
\end{equation}
where $z_c$ is the $c$-quantile of the standard normal distribution. The power calculation therefore depends on {correct evaluation of $V$ under $H_1$.}

{\section{Power analysis for the weighted Cox analysis}\label{sec:sample-size}}

{\subsection{Variance under the proportional hazards alternative}\label{subsec:Cox-variance}}

Because power analysis is conducted before data {acquisition}, our central task is to approximate $V$ using only a few expected summary quantities. This requires an analytical expression of $V$. {Let $\mu_z(t)=\bE[w_i\delta_i\mathbb{I}(Z_i\!=\!z)\mathbb{I}(T_i\leq t)]$ be the event rate by $t$ contributed by arm $z$ and $\mu(t)=\mu_0(t)+\mu_1(t)$ the combined event rate, with $\mu(t^\dagger)=d$ at the end of follow-up.} Throughout, we maintain the standard regularity conditions (stated in Appendix A.1), and derive an analytical expression of $V$ in Theorem \ref{thm:asymptotic-variance}.

\begin{theorem}\label{thm:asymptotic-variance}
For $\widehat{\tau}_n$ from equation \eqref{eq:estimating-equation}, $\sqrt{n}(\widehat{\tau}_n - \tau_0)$ is asymptotically normal with mean $0$ and variance {$V=A^{-2}B$, where $A=\int_0^{t^\dagger}\pi(\tau_0,t)\{1-\pi(\tau_0,t)\}\,d\mu(t)$ and $B=\bE[\eta_i^2]$, with $\eta_i=w_i\zeta_i$ and $\zeta_i=\delta_i\{Z_i-\pi(\tau_0,T_i)\}-\exp(\tau_0Z_i)\int_0^{T_i}\{Z_i-\pi(\tau_0,t)\}\,s_0^{-1}(\tau_0,t)\,d\mu(t)$.}
\end{theorem}
Theorem~\ref{thm:asymptotic-variance} holds without any model or identification assumption, with $\tau_0$ then being the root of the limit of $n^{-1}U_n(\tau)$, and extends to any nonnegative weights $w_i$ that are determined by $X_i$ and $Z_i$. It also gives the analytical limit of the robust $\widehat{V}$, and setting $w_i \equiv 1$ recovers that of \cite{lin1989robust}. The weights are estimated at the analysis stage, whereas Theorem~\ref{thm:asymptotic-variance} treats them as fixed, because accounting for their estimation in power analysis would require the joint distribution of the covariates and the outcome, and thus individual-level data rarely available at the design stage. Moreover, for inverse probability weights, treating weights as fixed overestimates the variance of $\widehat{\tau}_n$ and yields a conservative sample size \citep{Henmi2004paradox, shu2021variance}, which is desirable to ensure adequate power.  {Under the same regularity conditions, Theorem~\ref{thm:decomposition} presents our main result, separating $V$ under the alternative $H_1$.}

{\begin{theorem}\label{thm:decomposition}
Under $H_1$ and assumptions A1--A4, with inverse probability weights $w_i$,
\begin{equation}\label{eq:V-decomposition}
    V=\frac{\kappa_1(1+\Delta_1)}{rd_1}+\frac{\kappa_0(1+\Delta_0)}{(1-r)d_0}+\epsilon,
\end{equation}
where $\kappa_z = 1+\bV[w_i | Z_i\!=\!z]$ and $\epsilon=\{rC_1+(1\!-\!r)C_0\}/A^2$, with $C_z=\Cov\{w_z(e_i),Q_z(X_i)\}$ and $Q_z(x)=\bE[\zeta_i^2\mid Z_i\!=\!z, X_i\!=\!x]$ for $\zeta_i$ from Theorem \ref{thm:asymptotic-variance}, and
\begin{equation}\label{eq:Delta}
    \Delta_1=\frac{\bV\{\pi(\tau_0,\widetilde{T}_1)\}}{[1-\bE\{\pi(\tau_0,\widetilde{T}_1)\}]^2},\qquad \Delta_0=\frac{\bV\{\pi(\tau_0,\widetilde{T}_0)\}}{[\bE\{\pi(\tau_0,\widetilde{T}_0)\}]^2},
\end{equation}
are the squared coefficients of variation of $1-\pi(\tau_0,\widetilde{T}_1)$ and $\pi(\tau_0,\widetilde{T}_0)$, respectively, where $\widetilde{T}_z$ follows  
the weighted distribution of observed event times in arm $z$,
$\Pr(\widetilde{T}_z\leq t) = \mu_z(t)/\mu_z(t^\dagger)$.
\end{theorem}}

{We call $\epsilon$ the \emph{confounding residual}; it can be positive or negative in observational studies. We call $\kappa_z$ the \emph{design effect} of arm $z$, in the spirit of \cite{kish1965survey}; it measures the inflation of arm $z$'s contribution to $V$ in an inverse probability weighted study relative to that in the target randomized trial with the same treatment proportion. In a randomized trial, propensities are known and fixed, and thus $\epsilon=0$ and $\kappa_1=\kappa_0=1$.}

\subsection{{The working variance and sample size formula}} \label{subsec:working-variance}

{Theorem~\ref{thm:decomposition} separates conventional planning inputs from terms requiring the survival and censoring trajectories over follow-up. The targeted effect $\tau_0$ and the treatment proportion $r$ are inherent to every study design, the event rates $d_1$ and $d_0$ are conventional for time-to-event outcomes, and the design effects $\kappa_1$ and $\kappa_0$ depend only on propensity scores. In contrast, $\Delta_z$ depends on $\pi(\tau_0,t)$ over the follow-up, and hence on the survival and censoring curves, while $\epsilon$ additionally depends on how the outcome varies with the covariates. They are challenging to specify at the design stage: the survival and censoring curves are what the study is designed to investigate, and their relation to covariates needs individual-level data yet to be collected. Therefore, following the convention in the literature (elaborated below), we set $\Delta_z=0$, equivalently treating $\pi(\tau_0,t)$ as constant, and drop $\epsilon$, as if $w_z(e_i)$ is uncorrelated with $Q_z(X_i)$; this gives the \emph{working variance}
\begin{equation}\label{eq:working-variance}
    \widetilde{V}=\frac{\kappa_1}{rd_1}+\frac{\kappa_0}{(1-r)d_0}.
\end{equation}
Replacing $V$ by $\widetilde{V}$ in equation \eqref{eq:sample-size-general} gives the required sample size $N=(z_{1-\alpha} + z_\beta)^2\widetilde{V}/\tau_0^2$. 
}

{Classical formulas for randomized trials likewise avoid specifying these trajectories, often by treating $\pi(\tau_0,t)$ as a constant. 
\cite{schoenfeld1983sample} holds $\pi(\tau_0,t)$ at $r$, the constant value of $\pi(0,t)$ under $H_0$ with common censoring. \cite{freedman1982tables} holds $\pi(\tau_0,t)$ at $\pi(\tau_0,0)$ by fixing the at-risk ratio of the two arms at its initial value. The exponential-model formulas instead assume constant hazards \citep{Rubinstein1981planning, lachin1986evaluation}. Our working variance $\widetilde{V}$ holds $\pi(\tau_0,t)$ at $q=rd_1/d$, the treated arm's share of inverse probability weighted events over follow-up, which is also the average of $\pi(\tau_0,t)$ over event times weighted by $w_i$. The following proposition bounds the error from treating $\pi(\tau_0,t)$ as constant.}

{\begin{proposition} \label{prop:working-error}
Under the conditions and with $\epsilon$ of Theorem~\ref{thm:decomposition}, $V$ can be expressed in $\widetilde{V}$ as $V=\widetilde{V}(1+\Delta)+\epsilon$, where $\Delta\geq0$ and $\Delta=0$ if and only if $\pi(\tau_0,t)$ is constant over the event times. Moreover, $\Delta \leq (\pi_u-\pi_l)^2/[4\min\{\pi_l\pi_u,(1-\pi_l)(1-\pi_u)\}]$, where $\pi_l$ and $\pi_u$ are the infimum and supremum of $\pi(\tau_0,t)$ over $[0,t^\dagger]$, respectively, assumed to lie in $(0,1)$.
\end{proposition}

Proposition~\ref{prop:working-error} shows that, apart from $\epsilon$, the working variance $\widetilde{V}$ understates $V$ by a factor of $1+\Delta\geq1$, where $\Delta$ depends on the variation of $\pi(\tau_0,t)$ and thus on $r$, $\tau_0$, and the survival and censoring curves. In randomized trials with treatment proportions of $1/3$, $1/2$ or $2/3$, complete follow-up to $t^\dagger$ ($G_1=G_0\equiv1$) and control standardized event rates $d_0$ from $0.2$ to $0.8$, $\widetilde{V}$ understates $V$ by at most $0.2\%$, $0.8\%$, $1.8\%$ and $3.1\%$ for hazard ratios of $0.8$, $0.6$, $0.4$ and $0.2$, respectively (Appendix B.1).}

\section{Power calculation in observational studies}\label{sec:v-obs}

\subsection{Sample size formula with inverse probability weights}\label{subsec:v-ipw}

In observational studies, $e_i$ is unknown and varies across units. {The working variance $\widetilde{V}$ therefore requires the design effects $\kappa_z$, and hence} the propensity score distribution, which is rarely available before data {acquisition}. We address this by adopting the approach of \cite{liu2025sample}, which characterizes the distribution of $e_i$ through two {\emph{user-specified}} summary inputs: (i) the treatment proportion $r$, and (ii) a measure of covariate similarity between the two arms, {for which they use the Bhattacharyya coefficient of the propensity score densities, $\int_0^1 \sqrt{f_1(u) f_0(u)}\, \mathrm{d} u$}, where $f_z(u)$ is the density of {$e_i$ in arm $z$}. Under a logistic propensity score model, \cite{liu2025sample} showed that $e_i$ is well approximated by a Beta distribution $\mathrm{Beta}(a, b)$: the linear combination of covariates is approximately normal by appealing to the Lyapunov central limit theorem, so $e_i$ follows a logit-normal distribution, which is in turn closely approximated by a Beta distribution. The two inputs uniquely determine $(a,b)$. {However, the Bhattacharyya coefficient is concentrated in $(0.8,1)$ despite its theoretical range of $[0,1]$, which complicates its interpretation and specification at the design stage.

We propose to replace the similarity measure by the \emph{overlap coefficient} \citep{inman1989overlapping}, $\phi:=\int_0^1\min\{f_1(u),f_0(u)\}\,\mathrm{d}u$, {the mass shared by the propensity score densities $f_1$ and $f_0$ of the two arms. Equivalently, $\phi$ is the largest probability, over all joint distributions of a pair of randomly drawn treated and control units, that they have the same propensity score. $0 \leq \phi\leq 1$ with  $\phi=0$ meaning no overlap between two arms and $\phi=1$ meaning complete overlap, as the case in randomized trials. {The overlap coefficient $\phi$ can be estimated empirically from prior data (details in Appendix B.2). When no such data are available, we recommend the following rule of thumb: setting $\phi < 0.50$, $[0.50, 0.65)$, $[0.65, 0.75)$, $\geq 0.75$ for very poor, poor, moderate and good overlap, respectively; these approximately match the threshold of \cite{liu2025sample} for the Bhattacharyya coefficient.}

For a study with treatment proportion $r$ and $e(x)\sim \mbox{Beta}(a,b)$, we have two facts (the proof is relegated to Appendix A.3): (i) $r$ is the expectation of propensity score, $\bE[e(x)]=r$, and (ii) $\phi$ has a closed-form expression in terms of $(a,b)$, giving two identities}:
\begin{equation} \label{eq:ab-rphi}
      {r = \frac{a}{a + b}, \qquad \phi=1-\frac{\Gamma(a+b+1)}{\Gamma(a+1)\Gamma(b+1)}\,r^a(1-r)^b.}
\end{equation}
In Appendix A.3, we show that the mapping from $(r,\phi) \in (0, 1)^2$ to $(a,b) \in (0, \infty)^2$ based on equation \eqref{eq:ab-rphi} is bijective, with $\phi$ strictly increasing in $a+b$ for any fixed $r$. Therefore, given $(r,\phi)$, we can uniquely solve the Beta parameters $(a,b)$ and thus the propensity score distributions from equation \eqref{eq:ab-rphi} using bisection.}

{Theorem~\ref{thm:asymptotic-variance} requires regularity conditions to ensure a finite $V$. For inverse probability weights with $e_i\sim\mathrm{Beta}(a,b)$, this requires $a,b>1$, which is equivalent to $\phi>1-r^{\,r/(1-r)}$ for $r\geq1/2$, and the threshold is symmetric about $r=1/2$ (proved in Appendix A.3). The minimum $\phi$ increases from $0.50$ toward $0.63$ as $r$ moves from $1/2$ to $1$, so the condition always holds under moderate or better overlap. For $a,b>1$, the design effects $\kappa_z$ are finite and determined by $a$ and $b$, and hence by $r$ and $\phi$ via equation \eqref{eq:ab-rphi}, unlike those requiring individual weights to compute \citep{zhou2020propensity}. Expressing $\kappa_z$ of equation \eqref{eq:working-variance} in $a$ and $b$ gives the working variance for observational studies (Appendix Corollary 2):}
{
\begin{equation}\label{eq:working-V-ipw}
    \widetilde{V}_{\IPW} = (a+b-1) \left[\frac{1}{(a-1)d_1} + \frac{1}{(b-1)d_0}\right].
\end{equation}
}%
{As $\phi\to1$, both $\kappa_1$ and $\kappa_0$ tend to $1$, so $\widetilde{V}_{\IPW}$ tends to the working variance of a randomized trial. In summary, computing $N$ with $\widetilde{V}_{\IPW}$ only requires four inputs: (i) treatment proportion $r$, (ii) the overlap coefficient $\phi$, (iii) the targeted effect $\tau_0$, and (iv) the standardized event rates $d_1$ and $d_0$.}

{The working variance drops the confounding residual $\epsilon$, which, by Theorem~\ref{thm:decomposition}, arises from the covariances $C_z$ between $w_z(e_i)$ and $Q_z(X_i)$.
Appendix Proposition~3 bounds $\epsilon$ relative to $V\!-\epsilon$ via two key quantities: the coefficient of variation of $w_z(e_i)$ and the $R^2$ value of regressing $Q_z(X_i)$ on the propensity score $e_i$. The weights become less dispersed as overlap improves, so $\epsilon$ vanishes as $\phi\to1$. Furthermore, because $Q_z(x)$ depends on covariates $X_i$ only through the distribution of $T_i(z)$ given $X_i=x$, the $R^2$ is small unless the propensity score alone largely explains how the outcome varies with the covariates. Hence, we expect $\epsilon$ to remain small relative to $V$, especially under good overlap. Previous work has dropped similar terms. For example, the design effect of \cite{kish1965survey} depends only on the weights but not on their association with the outcome, and \cite{shook2022power} similarly drop an association term in their power analysis with non-censored outcomes.}

{The standardized event rate $d_z$ is the event rate that arm $z$ would show if its covariates followed the distribution of the study population, as is the case in a randomized trial. At the design stage, prior observational studies with comparable cohorts can inform $d_1$ and $d_0$ through reported inverse probability weighted Kaplan--Meier or counterfactual survival curves and the accompanying weighted numbers at risk at selected follow-up times \citep{Xie2005AdjustKM, Cheng2022survivals}. Randomized trials that recruit from a similar population, such as premarketing trials, can also inform these rates. If these sources are not available, the observed event rates $\bar{d}_z$ from a preliminary database query are a cruder substitute, as they average over the covariates of arm $z$ rather than those of the study population. Appendix Proposition 4 gives the worst-case bound $|d_z/\bar d_z-1|\leq(\kappa_z-1)^{1/2}c_z$, with $c_z$ the coefficient of variation of $\bE(\delta_i\mid e_i,Z_i=z)$ over arm $z$, so the observed event rates can serve under good overlap or when the event rate of arm $z$ varies only mildly with the propensity score.}

\subsection{Sample size formula with general balancing weights}\label{subsec:v-wate}

Inverse probability weights {target the treatment effect in the overall population (ATE),} represented by the study sample, and are sensitive to extreme propensities when the two arms overlap poorly. Moreover, real applications often concern a different \emph{target population} in which the treatment effect is evaluated. For example, the treated population represented by units who receive the treatment, or the overlap population, consisting of units in clinical equipoise. These populations need alternative weights, and we extend the results under inverse probability weights to the general propensity score \emph{balancing weights} \citep{li2018balancing}, where $w_i=Z_iw_1(e_i)+(1-Z_i)w_0(e_i)$ with
{
\begin{equation}\label{eq:balancing-weights}
    w_1(e)\propto h(e)/e, 
    \qquad
    w_0(e)\propto h(e)/(1-e),
\end{equation}
which weight both arms to the target population whose covariate distribution is tilted from that of $P$ by $h(e)$. As with inverse probability weights, $w_1$ and $w_0$ are normalized so that $w_i$ has mean one within each arm; then $w_i \equiv 1$ in a randomized trial for every weighting scheme.} Besides inverse probability weights ($h\equiv1$), balancing weights include two common special cases: the ATT weights ($h=e$) and the overlap weights ($h=e(1-e)$), which target the estimands in the treated (ATT) and overlap (ATO) populations, respectively.

{A Cox analysis with a chosen balancing weight $w$ tests $H_0$ against the proportional hazards alternative $H_1$ in the corresponding target population:
\begin{equation}\label{eq:hypotheses-weighted}
    H_0:\ F_1^w(t)=F_0^w(t) \quad\text{versus}\quad H_1:\ \lambda_1^w(t)=\lambda_0^w(t)\exp(\tau_0^w),
\end{equation}
where the potential survival function $F_z^w(t)$ averages $F_z(t|x)$ over the distribution of $X_i$ in the target population, with hazard function $\lambda_z^w(t)=-\{\log F_z^w(t)\}^\prime$. Denote $\widehat{\tau}_n^w$ the solution to equation \eqref{eq:estimating-equation} with $w_i$ from the chosen weighting scheme, and $V_w$ its asymptotic variance from Theorem~\ref{thm:asymptotic-variance}, with robust estimate $\widehat{V}_w$. A one-sided level-$\alpha$ Cox-Wald test based on $\widehat{\tau}_n^w$ and $\widehat{V}_w$ then requires the sample size $N = (z_{1 - \alpha} + z_\beta)^2 V_w/ {(\tau_0^w)^2}$ to reach power $\beta$.}

{
Under $H_1$ in \eqref{eq:hypotheses-weighted}, Theorem~\ref{thm:decomposition} also extends to general balancing weights (Appendix~A.2):
\begin{equation}\label{eq:Vw-decomposition}
    V_w=\frac{\kappa_1(1+\Delta_1)}{rd_1^w}+\frac{\kappa_0(1+\Delta_0)}{(1-r)d_0^w}+\epsilon^w,
\end{equation}
where $\kappa_z$, $\Delta_1$, $\Delta_0$ and $\epsilon^w$ are calculated as in Theorem \ref{thm:decomposition} with the chosen weights, except that the covariance $C_z$ in $\epsilon^w$ and the standardized event rate $d_z^w$ are taken over the target population. Dropping $\Delta_z$ and $\epsilon^w$ gives the working variance $\widetilde{V}_{w}=\kappa_1/(rd_1^w)+\kappa_0/\{(1-r)d_0^w\}$.
Under $e_i\sim\mathrm{Beta}(a,b)$ with $(a,b)$ solved from $(r, \phi)$ by equation \eqref{eq:ab-rphi}, overlap and treated weights have design effects (Appendix Corollary~2):
\begin{equation}\label{eq:kappa-balancing}
\begin{aligned}
    &\text{overlap weights:}\quad && \kappa_1=\frac{(b+1)(a+b+1)}{b(a+b+2)},\quad && \kappa_0=\frac{(a+1)(a+b+1)}{a(a+b+2)},\\
    &\text{ATT weights:} && \kappa_1=1, && \kappa_0=\frac{b(a+1)}{a(b-1)}.
\end{aligned}
\end{equation}
For ATT weights, finite $V_w$ requires only $b>1$ because $w_1\equiv1$; for overlap weights, $V_w$ is finite for all $a,b>0$. Computing $N$ with $\widetilde{V}_w$ thus requires the treatment proportion $r$ and the overlap coefficient $\phi$ as for inverse probability weights, whereas the targeted effect $\tau_0^w$ and the standardized event rates $d_1^w$ and $d_0^w$ refer to the target population.}

{The standardized event rate $d_z^w$ is the event rate that arm $z$ would show if its covariates followed the distribution of the target population. It can be informed by weighted Kaplan--Meier or counterfactual survival curves from prior studies with the same weighting scheme. For ATT weights, $d_1^w=\bar{d}_1$ and only $d_0^w$ needs to be anticipated. For overlap weights, prior randomized trials can also inform $d_z^w$, as they enroll patients in clinical equipoise. Without these sources, the observed event rates $\bar{d}_z$ are again a cruder substitute, as the bound of Section~\ref{subsec:v-ipw} extends to $|d_z^w/\bar d_z-1|\leq(\kappa_z-1)^{1/2}c_z$ with $\kappa_z$ of the chosen weights (Appendix Proposition~4).}

\section{Connection to classical power formulas in randomized trials}\label{sec:v-rct}

Our primary focus is on observational studies. Nonetheless, our main result, Theorem \ref{thm:decomposition}, includes randomized trial as a special case and thus connects to classical power formulas. Specifically, in a randomized trial, $e_i \equiv r$ for all $i$ and all balancing weights reduce to $w_{i} \equiv 1$. Denote $V_{\RCT}$ as the resulting $V$. {Then $\kappa_1=\kappa_0=1$ and $\epsilon=0$ exactly, and $\widetilde{V}$ simplifies to
\begin{equation}\label{eq:V-ipw-rct}
    \widetilde{V}_{\RCT}=\frac{1}{rd_1}+\frac{1}{(1-r)d_0},
\end{equation}
which, substituted into equation \eqref{eq:sample-size-general}, gives the required number of units $N$. Formula~\eqref{eq:V-ipw-rct} coincides with the classical exponential-model formula for the log ratio of two estimated hazards \citep{Rubinstein1981planning, lachin1986evaluation}, and thus extends it to the Cox analysis without requiring exponential survival. Multiplying $N$ by the combined event rate $d$ gives the required number of events
\begin{equation}\label{eq:event-share}
    Nd=\frac{(z_{1-\alpha}+z_\beta)^2}{q(1-q)\,\tau_0^2},
\end{equation}
which is Schoenfeld's formula with $r$ replaced by $q$, the expected treated arm's share of weighted events. Equivalently, Schoenfeld computes $N$ with $\widetilde{V}_{\SC}=1/\{r(1-r)d\}$, which is a local limit of $\widetilde{V}_{\RCT}$ at $d_1=d_0=d$, as is the case under the null $\tau_0=0$ with common censoring. The following result compares them under balanced allocation.}
\begin{proposition} \label{prop:Schoenfeld}
For a balanced randomized trial ($r=1/2$), {$\widetilde{V}_{\RCT}/\widetilde{V}_{\SC}=\cosh^2(\vartheta/2)\geq1$ with $\vartheta=\log(d_1/d_0)$, and equality holds if and only if $d_1=d_0$; under $H_1$ with common censoring, $d_1=d_0$ if and only if $\tau_0=0$.}
\end{proposition}

{In a balanced trial with proportional hazards, Propositions~\ref{prop:working-error} and \ref{prop:Schoenfeld} give $\widetilde{V}_{\SC} \leq \widetilde{V}_{\RCT} \leq V_{\RCT}$ for any censoring. Hence, for the Cox-Wald test, Schoenfeld's formula understates $V_{\RCT}$ at least as much as $\widetilde{V}_{\RCT}$ does. Under unequal allocation, however, $\widetilde{V}_{\SC}$ can fall on either side of $\widetilde{V}_{\RCT}$, so Schoenfeld's can underpower or over-enroll a trial, echoing its limitations in the log-rank test \citep{Yung2020sample}.} In practice, researchers often desire the number of units rather than events, which depends on censoring. {Censoring that differs by arm generally makes $d_1 \neq d_0$, and formula~\eqref{eq:V-ipw-rct} accommodates this by taking them as distinct inputs, which in a randomized trial are the observed event rates that earlier-phase trials routinely report. Under complete follow-up to $t^\dagger$, $d_1=1-(1-d_0)^{\exp(\tau_0)}$, so $d_0$ alone suffices.}

\section{Simulations}\label{sec:simu}

\subsection{Synthetic Data}\label{sec:simu-synthetic}
\subsubsection{Simulation design}\label{sec:simu-design}

{We adopt the simulation design of \cite{Cheng2022survivals} and adapt it following \cite{Evans2024simulating} to generate outcomes from a marginal causal model.} We simulate a superpopulation of $M=1{,}000{,}000$ units, where covariates $X_1$, $X_2$, and $X_3$ are standard normal with pairwise correlation $0.5$, and $X_4$, $X_5$, and $X_6$ are independent from $\mathrm{Bernoulli}(0.5)$; {denote $X=(X_1, \dots, X_6)$}. The true propensity score follows the logistic model {$\logit e(X)=\omega_0+\omega_1 X_e$, where $X_e=X\theta_e$ with $\theta_e\propto(0.2,0.3,-0.3,-0.2,-0.3,0.2)^\top$ scaled so that $\bV[X_e]=1$}, and $\omega_1$ controls the covariate overlap between the two arms. {We also replace the logit by the (reflected) Aranda--Ordaz link \citep{Francisco1981binary}, $e(X)=[1+\lambda\exp\{-(\omega_0+\omega_1 X_e) \}]^{-1/\lambda}$ with $\lambda=5$ ($\lambda=1$ recovers the logit), which, at the same $r$ and $\phi$, gives the propensity scores a much heavier tail near one.} We draw $Z \sim \mathrm{Bernoulli}(e(X))$.

We generate the potential event time $T(z)$ {from the marginal Weibull proportional hazards model in the target population of each weighting scheme $w$, $F_0^w(t)=\exp(-t^{1.2})$ and $F_1^w(t)=F_0^w(t)^{\exp(\tau_0^w)}$. Specifically, $T(z)=(F_z^w)^{-1}\{K_w(\psi_z)\}$, where $\psi_z$ is the latent risk of each unit under treatment $z$ and $K_w$ its distribution function in the target population. The latent risk $\psi_z = R_{\psi \sim X} X_y + (1\!-\!R^2_{\psi \sim X})^{1/2}\xi_z$ is the sum of a covariate part and an independent noise $\xi_z \!\sim\! N(0,1)$, where $R^2_{\psi \sim X}$ is the $R^2$ value of regressing $\psi_z$ on $X$. The covariate part is $X_y = \rho X_e + (1\!-\!\rho^2)^{1/2} X_\perp$, where $X_\perp$ is the standardized residual from regressing $X_t=X \theta_t$ on $X_e$, with $\theta_t=(-0.4,-0.2,0.1,0.1,0.2,-0.3)^\top$, and hence $\rho = \corr(X_e, X_y)$. Thus, $R^2_{\psi \sim X}$ and $\rho$ control the strength of confounding, where $R^2_{\psi \sim X}$ measures how much $T(z)$ varies with covariates, and $\rho$ how much the covariates that drive $T(z)$ also drive $Z$. We set the end of follow-up time $t^\dagger$ so that $F_0^w(t^\dagger)=0.25$. We generate the {time to loss to follow-up} independently as $C|Z\!=\!z \sim \mathrm{Exponential}(\nu_z)$, with $\nu_z$ tuned to the specified proportion of units lost to follow-up in arm $z$, standardized to the target population}.

For randomized trials, we set $e_i \equiv r$ so $Z$ is independent of $X$, {and generate $T(z)$ with $R^2_{\psi\sim X} \!=\! \rho \!=\! 0.5$}. We consider $r \!\in\! \{1/3, 1/2, 2/3\}$ to represent allocations $1{:}2$, $1{:}1$, and $2{:}1$ of treated to control, and marginal hazard ratios $\exp(\tau_0) \!\in\! \{0.4, 0.6, 0.8\}$, which we refer to as large, moderate, and small effects, respectively. We evaluate the proposed $\widetilde{V}_{\RCT}$, Schoenfeld's $\widetilde{V}_{\SC}$, and the formula of \cite{freedman1982tables}. The proposed formula requires the event rates $d_1$ and $d_0$, and Schoenfeld's the combined event rate $d$. In contrast, Freedman's requires the potential survival probability $F_0(t^\dagger)$ and the {proportion of units lost to follow-up as separate inputs. We present the main results with $10\%$ and $30\%$ loss to follow-up in the control and treated arms, respectively, and those with $20\%$ in each arm in Appendix C.1.}

For observational studies, we {consider $r \!\in\! \{1/3,1/2\}$ and vary $\phi \in\! \{0.95, 0.90, \ldots, 0.65\}$ by calibrating $\omega_0$ and $\omega_1$}. We consider inverse probability, ATT, and overlap weights. {We present the results under the primary scenario, with $\exp(\tau_0^w)=0.6$, $10\%$ and $30\%$ loss to follow-up in the control and treated arms, confounding parameters $R^2_{\psi\sim X} \!=\! \rho \!=\! 0.5$, and the logit link; Appendix C.2 reports those that change one factor at a time to $\exp(\tau_0^w)= 0.8$, $20\%$ loss to follow-up in each arm, $\rho=0.25$ or $R^2_{\psi\sim X}\!=\! \rho \!=\! 0.25$, or the Aranda--Ordaz link.} For inverse probability weights, we compare the proposed $\widetilde{V}_{\IPW}$ with the formula of \cite{hsieh2000sample}, which inflates $\widetilde{V}_{\SC}$ by $1/(1-R^2_{Z\sim X})$, where $R^2_{Z\sim X}$ is the $R^2$ value from regressing $Z$ on $X$. We derive an approximated $R^2_{Z\sim X}$ from $(r, \phi)$ via the Beta approximation $e(X) \sim \mathrm{Beta}(a,b)$, under which $R^2_{Z\sim X} = \bV[e(X)]/\bV[Z]=1/(a+b+1)$, yielding $\widetilde{V}_{\HL}=\widetilde{V}_{\SC} \!\cdot\! [1+1/(a+b)]$. {For ATT and overlap weights, we evaluate the proposed $\widetilde{V}_w$ alone, as no alternative method is available.} We exclude \cite{scosyrev2019power} from the comparison, because they require fitting a Cox model on individual-level data.

For each configuration, we compute the required sample size $N$ {for a one-sided test with $\alpha=0.025$ and target power $0.8$}, with inputs computed from the superpopulation. We then draw $N$ units with replacement from the $M$ units, stratified by $Z$ in randomized trials and unstratified in observational studies, and repeat for {$B=5000$ replicates}. In each replicate, {we conduct the Cox-Wald test with the robust variance estimate, with $w\equiv1$ in randomized trials, or weights from the true propensity scores or those estimated by a main-effects logistic regression of $Z$ on $X$ in observational studies.} We report the proportion of rejections as the empirical power {and the coverage of the $95\%$ confidence interval.} {Appendix C.2 also reports the type I error at the $N$ of the proposed formula, with $T(z)$ generated under $\tau_0^w=0$.}

\subsubsection{Results: Randomized trials}\label{sec:simu-rct}

\vspace{-12pt}
\begin{table}[ht]
    \centering
    \small
    \setlength{\tabcolsep}{3pt}
    \renewcommand{\arraystretch}{0.96}
    \caption{{Sample size and empirical power of the proposed, Schoenfeld, and Freedman formulas in randomized trials with 10\% and 30\% loss to follow-up in the control and treated arms, respectively. The Monte Carlo error is approximately 0.006.}}
    \label{tab:rct_power}
    \begin{tabular*}{0.85\textwidth}{@{\extracolsep{\fill}}ccccccc@{}}
    \toprule
    \multirow{2}{*}[-4pt]{\shortstack[c]{Marginal \\ hazard ratio}} & \multicolumn{2}{c}{Proposed} & \multicolumn{2}{c}{Schoenfeld} & \multicolumn{2}{c}{Freedman} \\
    \cmidrule(lr){2-3} \cmidrule(lr){4-5} \cmidrule(lr){6-7}
     & $N$ & Power & $N$ & Power & $N$ & Power \\
    \midrule
    \multicolumn{7}{c}{\textit{Treatment proportion $= 1/3$}} \\
    0.4 & 103 & .847 & 73 & .687 & 117 & .892 \\
    0.6 & 268 & .815 & 221 & .744 & 288 & .844 \\
    0.8 & 1249 & .806 & 1114 & .760 & 1271 & .810 \\
    \midrule
    \multicolumn{7}{c}{\textit{Treatment proportion $= 1/2$}} \\
    0.4 & 82 & .829 & 72 & .781 & 91 & .868 \\
    0.6 & 222 & .811 & 212 & .798 & 239 & .839 \\
    0.8 & 1059 & .799 & 1037 & .792 & 1114 & .819 \\
    \midrule
    \multicolumn{7}{c}{\textit{Treatment proportion $= 2/3$}} \\
    0.4 & 82 & .819 & 92 & .857 & 86 & .836 \\
    0.6 & 232 & .797 & 257 & .833 & 247 & .819 \\
    0.8 & 1130 & .798 & 1222 & .834 & 1231 & .833 \\
    \bottomrule
    \end{tabular*}
\end{table}

Table \ref{tab:rct_power} reports results under {differential loss to follow-up between arms}. Across the scenarios considered, the proposed formula{, which coincides with the exponential-model formula,} achieves the nominal power of $0.8$, confirming that {it extends to non-exponential survival}. In contrast, the Schoenfeld formula exhibits the bidirectional pitfall discussed in Section~\ref{sec:v-rct}{, underpowering at $r=1/3$ and over-enrolling at $r=2/3$ across all effects. The Freedman formula attains adequate power but tends to over-enroll, and requires different inputs. Appendix C.1 shows that all three formulas attain coverage around $0.95$ and the Cox-Wald test retains its nominal level at the required sample sizes of the proposed formula, with similar patterns under $20\%$ loss to follow-up in each arm.}

{Proposition~\ref{prop:working-error} predicts that the proposed formula slightly underpowers with an error bounded by the range of $\pi(\tau_0,t)$ over follow-up, while Table~\ref{tab:rct_power} shows adequate power, especially at larger effects. This is in part because $\Delta/(1+\Delta)$, the relative underestimation of $V$ from omitting $\Delta$, remains negligible (Section~\ref{subsec:working-variance}). Moreover, larger effects require smaller samples with fewer events, in which the robust variance estimate tends to lead the Cox-Wald test to reject more often than asymptotic theory predicts \citep{fay2001small}, slightly raising the empirical power of all three formulas.}

\subsubsection{Results: Observational studies}\label{sec:simu-obs}

\vspace{-12pt}
\begin{table}[ht]
    \centering
    \small
    \setlength{\tabcolsep}{3pt}
    \renewcommand{\arraystretch}{0.96}
    \caption{{Sample size and empirical power across overlap levels and target populations in observational studies with $\exp(\tau_0^w)=0.6$, $R^2_{\psi\sim X}\!=\!\rho\!=\!0.5$, the logit link, and $10\%$ and $30\%$ loss to follow-up. Main power entries use propensity scores $\hat{e}_i$ estimated in each replicate, with powers using the true $e_i$ in parentheses. The Monte Carlo error is approximately 0.006.}}
    \label{tab:obs_power}
    \begin{tabular*}{\textwidth}{@{\extracolsep{\fill}}ccccccccc@{}}
    \toprule
    & \multicolumn{4}{c}{ATE} & \multicolumn{2}{c}{\multirow{2}{*}{ATT}} & \multicolumn{2}{c}{\multirow{2}{*}{ATO}} \\
    \cmidrule(lr){2-5}
    & \multicolumn{2}{c}{Proposed} & \multicolumn{2}{c}{Hsieh--Lavori} & \multicolumn{2}{c}{} & \multicolumn{2}{c}{} \\
    \cmidrule(lr){2-3} \cmidrule(lr){4-5} \cmidrule(lr){6-7} \cmidrule(lr){8-9}
    $\phi$ & $N$ & Power & $N$ & Power & $N$ & Power & $N$ & Power \\
    \midrule
    \multicolumn{9}{c}{\textit{Treatment proportion $= 1/3$}} \\
    0.95 & 271 & .852 (.820) & 224 & .752 (.748) & 271 & .857 (.816) & 271 & .864 (.818) \\
    0.90 & 274 & .848 (.810) & 224 & .745 (.728) & 271 & .861 (.816) & 271 & .859 (.817) \\
    0.85 & 284 & .852 (.818) & 229 & .750 (.730) & 278 & .867 (.827) & 276 & .868 (.821) \\
    0.80 & 298 & .861 (.820) & 234 & .745 (.715) & 286 & .851 (.823) & 281 & .865 (.823) \\
    0.75 & 323 & .861 (.824) & 242 & .719 (.705) & 299 & .846 (.817) & 290 & .864 (.817) \\
    0.70 & 368 & .882 (.846) & 253 & .706 (.696) & 318 & .846 (.812) & 302 & .855 (.812) \\
    0.65 & 458 & .916 (.893) & 264 & .680 (.677) & 347 & .843 (.813) & 316 & .865 (.816) \\
    \midrule
    \multicolumn{9}{c}{\textit{Treatment proportion $= 1/2$}} \\
    0.95 & 224 & .856 (.808) & 213 & .831 (.788) & 224 & .837 (.809) & 224 & .850 (.806) \\
    0.90 & 226 & .854 (.817) & 215 & .836 (.799) & 228 & .831 (.808) & 226 & .849 (.804) \\
    0.85 & 231 & .849 (.808) & 220 & .828 (.788) & 236 & .847 (.819) & 231 & .856 (.810) \\
    0.80 & 239 & .851 (.812) & 226 & .825 (.786) & 248 & .846 (.810) & 237 & .854 (.816) \\
    0.75 & 252 & .862 (.819) & 234 & .830 (.787) & 268 & .838 (.810) & 246 & .867 (.823) \\
    0.70 & 271 & .861 (.831) & 244 & .812 (.787) & 297 & .843 (.821) & 257 & .863 (.814) \\
    0.65 & 307 & .878 (.835) & 258 & .797 (.770) & 352 & .860 (.840) & 272 & .863 (.821) \\
    \bottomrule
    \end{tabular*}
\end{table}

Table \ref{tab:obs_power} reports results of observational studies under the primary scenario. For all three target populations, the proposed formula achieves the nominal power of $0.8$ across the overlap coefficients evaluated, {with both estimated and true propensity scores. For the study population (ATE), it tends to be more conservative under poor overlap, as the Beta distribution has heavier tails than the logit-normal and thus gives extra protection against extreme inverse probability weights.} By contrast, {the Hsieh--Lavori formula underpowers at $r=1/3$, and at $r=1/2$ with true propensity scores}, with the shortfall widening as $\phi$ decreases, because it inherits Schoenfeld's variance as the baseline and inflates it by a factor derived for a Cox model adjusting for covariates, {which is less sensitive to overlap than the proposed design effects}. We evaluate both methods using the Beta approximation without individual-level data, facilitating a fair comparison under the same inputs possible in practice. {Among the three weights, overlap weights require the smallest sample size and is the least sensitive to deteriorating overlap.} The inverse probability and {ATT} weights both place large weights on units with extreme propensities, whereas overlap weights downweight them and thus remain efficient under poor overlap. While the target population is primarily determined by the scientific question, this pattern provides a practical reference when researchers face limited feasible sample size under anticipated poor overlap.

{The estimated and true propensity scores serve different purposes. True propensity scores fix the weights as Theorem~\ref{thm:asymptotic-variance} does, and thus evaluate the formula against the variance it targets. In contrast, estimated propensity scores evaluate the analysis an investigator conducts in practice. As predicted by theory (Section~\ref{subsec:Cox-variance}), estimated propensity scores give higher power for inverse probability weights, and in our simulations, likewise for ATT and overlap weights. Accordingly, with true propensity scores, the type I error and coverage of the Cox-Wald test stay around their nominal levels, whereas with estimated propensity scores, the type I error is below the nominal $0.025$ and coverage above $0.95$ (Appendix C.2.2).} This conservativeness is favorable in observational settings, because it provides a safe margin against not only the uncertainty in design-stage inputs, but also the {confounding residual $\epsilon$}.

{Appendix C.2.1 shows largely similar patterns under additional simulation scenarios. Of note, weaker confounding with $R^2_{\psi\sim X}\!=\!\rho\!=\!0.25$ barely changes the required sample sizes or the power with true propensity scores, but lowers the power with estimated propensity scores by about $0.03$, which remains adequate. Thus, stronger confounding need not lower the power of a practical analysis. In addition, under the Aranda--Ordaz link, true propensity scores above $0.98$, which are nearly absent with the logit link, become far more frequent (Appendix C.2.3); nevertheless, the proposed formula attains the nominal power in almost all configurations. Only the ATT falls below the nominal power at lower overlap $\phi\leq0.70$ with true propensity scores (minimum $.750$ at $r=1/3$ and $\phi=0.65$); with estimated propensity scores, the ATT nearly attains the nominal power throughout (minimum $.782$ at $r=1/3$ and $\phi=0.65$). Since such tails are rarer in propensity scores estimated by logistic regression and, even if present, are commonly mitigated by approaches such as trimming or overlap weights, they alone are unlikely to leave a practical analysis underpowered.}

{The superpopulations also allow us to evaluate directly the two terms dropped by the working variance (Appendix C.2.4). Consistent with the randomized trials, $\Delta/(1+\Delta)$ is at most $0.8\%$ and remains less than $0.1\%$ across most simulation scenarios and weighting schemes. As expected in Section~\ref{subsec:v-ipw}, $\epsilon$ tends to grow as overlap deteriorates and as $R^2_{\psi\sim X}$ and $\rho$ increase, but remains small; specifically, $|\epsilon|/(V-\epsilon)$ is at most about $6\%$ in the primary simulation scenario, and $8\%$ even in additional superpopulations with $R^2_{\psi\sim X}\!=\! \rho \!=\! 0.9$.}

\subsection{Real-world data calibrated simulation}\label{sec:simu-rwd}

To complement the observational studies with synthetic data, we further conduct simulations on the publicly available benchmark observational data in \cite{Connors96}, which evaluates the causal effect of right heart catheterization among hospitalized critically ill patients. {The purpose of these simulations} is not to re-analyze the existing studies or perform \emph{post-hoc} power calculations, but to validate the proposed methods in real-world settings under unknown data generation and potential violations of assumptions and approximations invoked throughout. The treatment is whether or not the patient received right heart catheterization and the outcome is the time-to-death up to 180 days after admission. The data contain $5{,}735$ patients with 72 covariates; we follow \cite{hirano2001rhc} on covariate selection for the propensity score model. The original cohort exhibits severely poor overlap. To make the simulation more representative of common studies, we trim to patients whose estimated propensities fall in $[0.05, 0.95]$, retaining $M=5{,}393$ ($2{,}175$ treated, $3{,}218$ controls), which we regard as the full study cohort. We re-estimate the propensities on this cohort and obtain an empirical overlap coefficient of {$\phi=0.58$}. We then obtain the marginal hazard ratios {for the ATE, ATT, and ATO} by fitting a weighted marginal Cox model with estimated propensities, {and take} them as the target effects. {We also calculate the event rates of each arm standardized to each target population.}

We then treat the trimmed full study cohort as the superpopulation and design a hypothetical observational study on this population to detect the target effects. Following Section {\ref{sec:simu-design}}, we calculate the required sample size $N$ using $(r, \phi)$, the {standardized} event rates, and the target effects obtained from this cohort. We draw $N$ patients with replacement {for $B=5{,}000$ replicates, and conduct the Cox-Wald test with the robust variance estimate, using propensity scores re-estimated in each replicate or, alternatively, those fitted on the trimmed full cohort, which stay fixed and play the role of the true propensity scores}.

\begin{table}[ht]
    \centering
    \small
    \setlength{\tabcolsep}{3pt}
    \renewcommand{\arraystretch}{0.96}
    \caption{{Sample size, empirical power, and coverage of 95\% confidence interval for ATE, ATT, and ATO in simulations calibrated to the right heart catheterization cohort. Main entries of power and coverage use propensity scores estimated in each replicate, with those using the propensity scores fitted on the (trimmed) full study cohort in parentheses.}}
    \label{tab:rhc_power}
    \begin{tabular*}{0.85\textwidth}{@{\extracolsep{\fill}}ccccc@{}}
    \toprule
    & \multicolumn{2}{c}{ATE} & \multirow{2}{*}{ATT} & \multirow{2}{*}{ATO} \\
    \cmidrule(lr){2-3}
    & Proposed & Hsieh--Lavori & & \\
    \midrule
    $N$ & 4430 & 2609 & 5390 & 2180 \\
    Power & .927 (.919) & .698 (.729) & .833 (.820) & .815 (.799) \\
    Coverage & .970 (.950) & .969 (.952) & .966 (.955) & .974 (.946) \\
    \bottomrule
    \end{tabular*}
\end{table}

{Table~\ref{tab:rhc_power} reports the results.} For ATE, the proposed formula is conservative with empirical power above the nominal level, as in the synthetic data under the poorest overlap; in contrast, the Hsieh--Lavori formula considerably underpowers for the reasons elaborated in Section~{\ref{sec:simu-obs}}. For the ATT and ATO, the proposed formula attains power near the nominal level. {As in the synthetic data, re-estimating the propensity scores in each replicate slightly increases the empirical power and coverage of the proposed formula.}

{The right heart catheterization data also provide a real-world benchmark for the terms $\epsilon$ and $\Delta$ dropped by the working variance, without relying on any parametric data-generating model specified for simulations (Appendix C.3). With the weights and robust variance estimated on the full cohort with $M$ patients, $\Delta/(1+\Delta)<0.05\%$ and $|\epsilon|/(V-\epsilon)< 4\%$ approximately for all three target populations. Moreover, the observed event rates $\bar{d}_z$ are close to the standardized ones $d_z^w$, with $|d_z^w/\bar{d}_z-1|\leq 5.7\%$, so $\bar{d}_z$ would serve as a reasonable substitute here.}

\section{Discussion}\label{sec:discussion}

This paper develops analytical power and sample size formulas for causal inference with time-to-event outcomes {in observational studies}. We focus on the hazard ratio of the Cox model as it is the dominant estimand for survival analysis in medical research. {In particular, besides the canonical inputs for randomized trials, our method requires only one additional input characterizing covariate overlap for sample size determination}, extending \cite{liu2025sample} to survival outcomes. {Our formula includes randomized trials as a limiting case, extending the classical exponential-model formula to Cox analysis without assuming a parametric survival model, and quantifies the error that Schoenfeld's formula can produce.}

{Under the proportional hazards alternative, our formula needs neither survival and censoring curves nor parametric models for them as inputs.} Investigators who plan for non-proportional hazards or prefer reporting on the survival scale may turn to alternative estimands, such as the restricted mean survival time and survival probabilities at fixed time points \citep{Zhang2012DRobustMLT, mao2018propensity, Cheng2022survivals}. These estimands rely on their own asymptotic theory, and thus studies targeting them are better planned with sample size methods tailored to the corresponding analysis.

The proposed method can be extended in several directions. {In studies with multiple or factorial treatments \citep{zhao2026propensity}, each pairwise comparison can be planned with its own inputs at a significance level adjusted for multiplicity.} 
When competing events are present, the method applies directly to the cause-specific hazard, which reduces to a Cox model on the cause-specific event indicator; the subdistribution hazard would require analyzing the variance under the redefined risk set, parallel to the proposed method.
For clustered studies, {the robust variance needs to account for within-cluster correlation}; in the case of observational studies, the interaction between variance inflation due to clustering and weighting would require careful investigation.



\bibliographystyle{chicago}
\bibliography{PS-power}

\newpage
\spacingset{1.25}

\appendix
\titleformat{\section}[block]
  {\normalfont\Large\bfseries\centering}
  {Appendix \thesection.}{0.5em}{}
\titleformat{\subsection}[block]
  {\normalfont\large\bfseries\centering}
  {\thesubsection}{0.5em}{}
\titleformat{\subsubsection}[block]
  {\normalfont\normalsize\bfseries}
  {\thesubsubsection}{0.5em}{}

\section{Technical Proofs}

We follow the setup of Section 2. Units $i=1,\ldots,n$ are i.i.d.\ draws of $(T_i(0),T_i(1),C_i,Z_i,X_i)$, with potential event times $T_i(z)$ and $T_i^*=T_i(Z_i)$, censoring time $C_i$, treatment $Z_i\in\{0,1\}$ and covariates $X_i$. $r=\Pr(Z_i=1)$ is the treatment proportion and $e(x)=\Pr(Z_i=1\mid X_i=x)$ the propensity score. $G_z(t)=\Pr(C_i\geq t\mid Z_i=z)$ is the censoring distribution of arm $z$, which equals $\Pr(C_i\geq t\mid Z_i=z,X_i=x)$ for every $x$ under (A4) below. The maximum follow-up is fixed at $t^\dagger<\infty$, measured from each unit's entry. The observed data are $V_i=(T_i,\delta_i,Z_i,X_i)$, where $T_i=\min\{T_i^*,C_i,t^\dagger\}$ and $\delta_i=\bI(T_i=T_i^*)$. $Y_i(t)=\bI(T_i\geq t)$ is the at-risk process and $\mathcal{R}_i=\{l:Y_l(T_i)=1\}$ the risk set at $T_i$.

For conciseness, we prove both Theorems 1 and 2 under general weights $w_i$, which includes inverse probability weights as a special case. Specifically, Theorem \ref{thm:asymptotic-variance} is proved with any nonnegative weights that are determined by $X_i$ and $Z_i$ and satisfy other regularity conditions (stated in \ref{secapp:regularity-conditions}); Theorem \ref{thm:decomposition} further requires $w_i$ being propensity score balancing weights (defined in \ref{secapp:weight-and-population}). The associated lemmas and propositions are proved with respective general weights. Therefore, the proof serves both the inverse probability weights in Sections 2 and 3, and the extension to general balancing weights in Section 4.2.

Two groups of conditions are used in this appendix, the regularity conditions (R) and identification assumptions (A); group (A) is collected here for reference. Theorem \ref{thm:decomposition} further evaluates the asymptotic variance under the proportional hazards alternative (H1) to the null (H0), both stated in the target population of the weight $w$ defined in Section \ref{secapp:weight-and-population}.

\begin{itemize}
    \item[(A1)] SUTVA or consistency: $T_i^*=Z_iT_i(1)+(1-Z_i)T_i(0)$.
    \item[(A2)] Unconfoundedness: $\{T_i(0),T_i(1)\}\perp Z_i\mid X_i$.
    \item[(A3)] Overlap (positivity): $0<e(X_i)<1$ for all $i$.
    \item[(A4)] Arm-specific independent censoring: $\{T_i(z),X_i\}\perp C_i\mid Z_i=z$ for $z=0,1$.
\end{itemize}

\begin{itemize}
    \item[(H0)] Null: $F_1^w(t)=F_0^w(t)$ for all $t\in[0,t^\dagger]$.
    \item[(H1)] Alternative: $F_1^w(t)=F_0^w(t)^{\exp(\tau_0^w)}$ for all $t\in[0,t^\dagger]$, equivalently $\lambda_1^w(t)=\exp(\tau_0^w)\lambda_0^w(t)$.
\end{itemize}

\subsection{Theorem 1}\label{secapp:theorem1}

\subsubsection{Definitions and notations}\label{secapp:definitions-and-notations}

\textbf{Risk-set averages.}
Following Section 2, we define $S_k^*(\tau, t)$ and $s_k(\tau, t)$, $k\in\{0,1,2\}$, as the empirical and population weighted risk-set averages:
\begin{equation*}
    S_k^*(\tau, t) = n^{-1} \sum_{j=1}^n w_j Y_j(t) \exp(\tau Z_j) Z_j^k, k \in \{0,1,2\}, \quad s_k(\tau, t) = \bE\left[w_j Y_j(t)\exp(\tau Z_j) Z_j^k\right],
\end{equation*}
where the expectation is over a \textcolor{black}{generic} unit $j$ from the population distribution, and therefore $s_k(\tau, t)$ is a deterministic function of $t$ and $\tau$. These averages do not assume a correctly specified marginal Cox model in the target population, and they can differ from the ones that are defined under the true hazard (intensity). We further define the empirical and population treated arm's share of the weighted risk set at $t$ given $\tau$:
\begin{equation*}
    \pi_n^*(\tau, t) = \frac{S_1^*(\tau, t)}{S_0^*(\tau, t)}, \qquad \pi(\tau, t) = \frac{s_1(\tau, t)}{s_0(\tau, t)}.
\end{equation*}
We also define the observed counting process $N_i(t)=\bI(T_i\le t, \delta_i=1)$, which jumps exactly once at the observed possibly right-censored time $T_i$, when $\delta_i=1$.

\textbf{The weighted partial likelihood and estimand.} Denote $w_i$ the weight of unit $i$. The weighted partial likelihood is
\begin{equation*}
    \mathfrak{L}_w(\tau) = \prod_{i:\, w_i\delta_i>0}\left[\frac{w_i\exp(\tau Z_i)}{\sum_{l\in \mathcal{R}_i}w_l\exp(\tau Z_l)}\right]^{w_i\delta_i},
\end{equation*}
which, after taking logarithms and discarding the part free of $\tau$, becomes
\begin{equation*}
    \mathfrak{l}_w(\tau) = \sum_{i=1}^{n}w_i\delta_i\left[\tau Z_i-\log\sum_{l\in \mathcal{R}_i}w_l\exp(\tau Z_l)\right] = \sum_{i=1}^{n}w_i\delta_i\left[\tau Z_i-\log\left\{n\,S_0^*(\tau,T_i)\right\}\right],
\end{equation*}
where the second equation holds because $l \in \mathcal{R}_i$ if and only if $Y_l(T_i)=1$. Assume $\mathfrak{l}_w(\tau)$ is continuous and differentiable to the second order. Since $Z_j^k\exp(\tau Z_j)$ differentiates to $Z_j^{k+1}\exp(\tau Z_j)$, we have $\partial S_k^*/\partial\tau=S_{k+1}^*$, so differentiating $\mathfrak{l}_w(\tau)$ gives
\begin{equation}\label{eqapp:Un}
    U_n(\tau) = \frac{\partial \mathfrak{l}_w(\tau)}{\partial \tau} = \sum_{i=1}^{n}\psi_i^*(\tau), \qquad
    \psi_i^*(\tau)=w_i\delta_i\left[Z_i-\pi_n^*(\tau,T_i)\right],
\end{equation}
which is the weighted Cox partial-likelihood estimating equation in Section 2 with solution $\widehat{\tau}_n$ to $U_n(\widehat{\tau}_n)=0$, also seen in \cite{Binder1992CoxSurvey}. Correspondingly, the population weighted partial score is
\begin{equation}\label{eqapp:weightscore-population}
    u(\tau) := \mathbb{E}\!\left[ \psi_i(\tau) \right], \qquad \psi_i(\tau) = w_i\delta_i \left[Z_i-\pi(\tau,T_i)\right].
\end{equation}
We denote $\tau_0$ its root, that is, $u(\tau_0)=0$, and take $\tau_0$ as our estimand. This definition of $\tau_0$ does \emph{not} require a correctly specified marginal Cox model in the target population; instead, it only requires the existence of a unique root to the equation $u(\tau)=0$.

\textbf{The weighted event measure and cumulative hazard.} For a Borel $B \subseteq [0,t^\dagger]$, define
\begin{equation}\label{eqapp:weighted-event-measure}
    \mu(B) := \mathbb{E} \! \left[w_i \delta_i \, \bI(T_i\in B) \right],
\end{equation}
which is later shown (Lemma \ref{lemapp:1}) to be a finite measure on the Borel $\sigma$-algebra $\mathscr{B}([0,t^\dagger])$ over $[0,t^\dagger]$, termed the \emph{weighted event measure}. With a slight abuse of notation, we write
\begin{equation*}
    \mu(t):= \mu([0,t]) = \mathbb{E}[w_i \delta_i \bI(T_i \leq t)] = \mathbb{E}[w_iN_i(t)].
\end{equation*}
Thus, $\mu(t)$ is the population limit of the weighted cumulative event process $n^{-1}\sum_{i=1}^{n}w_iN_i(t)$, and $\mu(t^\dagger)=\mathbb{E}[w_i\delta_i]$ at the end of follow-up, later shown to be the combined event rate of the target population of the weight $w$ (Section \ref{secapp:weight-and-population}). $\mu$ can be further decomposed into the measure of each arm,
\begin{equation}\label{eqapp:weighted-event-measure-arm}
    \mu(B) = \bE[w_i\delta_i\bI(Z_i=1)\bI(T_i\in B)] + \bE[w_i\delta_i\bI(Z_i=0)\bI(T_i\in B)] =: \mu_1(B) + \mu_0(B),
\end{equation}
where $\mu_z(t):=\mu_z([0,t])$ and $\mu_1$ and $\mu_0$ have the total masses $\mu_1(t^\dagger)=\bE[w_i\delta_iZ_i]$ and $\mu_0(t^\dagger)=\bE[w_i\delta_i(1-Z_i)]$, both later shown to be positive under the regularity conditions (Remark in the proof of Lemma \ref{lemapp:4}).

We further define  the \emph{weighted cumulative hazard},
\begin{equation}\label{eqapp:weighted-cumulative-hazard}
    \Lambda_w(t):=\int_{[0,t]}\frac{d\mu(s)}{s_0(\tau_0,s)}, \qquad t\in[0,t^\dagger].
\end{equation}
As we will show in Lemma \ref{lemapp:1}, $\Lambda_w$ is non-decreasing and satisfies $d\mu(t) = s_0(\tau_0, t)d\Lambda_w(t)$.

\textbf{The weighted event time distribution and its induced moments.} The weighted event measure $\mu$ induces $\mu(t)/\mu(t^\dagger)$, the weighted distribution of observed event times. Denote $\widetilde{T} \sim \mu(t)/\mu(t^\dagger)$, then for any Borel function $f(t)$ on $[0,t^\dagger]$, $f(\widetilde{T})$ is a random variable. We can therefore define the mean and variance of $f(\widetilde{T})$ over $\mu(t)/\mu(t^\dagger)$ as
\begin{equation}\label{eqapp:weighted-event-time-moments}
    \bE_{\mu}[f] := \bE\!\left[ f(\widetilde{T}) \right] = \frac{\int_0^{t^\dagger} f(t) \,d\mu(t)}{\mu(t^\dagger)}, \quad \bV_{\mu}[f]:= \bV\!\left[ f(\widetilde{T}) \right] = \frac{\int_0^{t^\dagger} \left\{f(t) - \bE_{\mu}[f] \right\}^2 \,d\mu(t)}{\mu(t^\dagger)}.
\end{equation}
Similarly, $\mu_z(t)/\mu_z(t^\dagger)$ is the weighted distribution of observed event times in arm $z$, and its induced moments are defined as
\begin{equation*}
    \bE_{\mu_z}[f] := \bE\!\left[ f(\widetilde{T}_z) \right] 
    = \frac{\int_0^{t^\dagger} f(t) \,d\mu_z(t)}{\mu_z(t^\dagger)}, \quad
    \bV_{\mu_z}[f]:= \bV\!\left[ f(\widetilde{T}_z) \right] 
    = \frac{\int_0^{t^\dagger} \left\{f(t) - \bE_{\mu_z}[f] \right\}^2 \,d\mu_z(t)}{\mu_z(t^\dagger)},
\end{equation*}
where $\widetilde{T}_z \sim \mu_z(t)/\mu_z(t^\dagger)$. All are well defined once the total masses are positive, as later shown for $\mu_1(t^\dagger)$ and $\mu_0(t^\dagger)$, and hence for $\mu(t^\dagger)=\mu_1(t^\dagger)+\mu_0(t^\dagger)$.

\subsubsection{Regularity conditions}\label{secapp:regularity-conditions}

We maintain the following regularity conditions throughout this paper:

(R1) Existence of estimand: The equation $u(\tau)=0$ has a root $\tau_0 \in \mathbb{R}$. 

(R2) Regular weights: the weight $w_i=w(Z_i,X_i)$ is nonnegative and measurable with respect to $\sigma(Z,X)$, with $\mathbb{E}[w_i^2]<\infty$.

(R3) Continuous event times: The conditional distribution of $T_i^*$ given $(X_i,Z_i)$ almost surely has a density. Thus, $T_i^*$ has a marginal density, and $t \mapsto \mu(t)$ is continuous.

(R4) Andersen-Gill type: edited from A--D of \cite{Anderson1982Cox}, imposed on $S_k^*(\tau,t)$ and $s_k(\tau,t)$ over the interval $[0,t^\dagger]$, including the weight $w$. Let $\mathcal{B}$ be a neighborhood of $\tau_0$.
\begin{itemize}[leftmargin=3em, itemsep=0pt, topsep=0pt]
    \item[(A)] Finite interval: $\Lambda_w(t^\dagger) < \infty$.
    \item[(B)] Asymptotic stability: For $k=0,1,2$, $\sup_{t\in[0,t^\dagger],\,\tau\in\mathcal{B}}\left|S_k^*(\tau,t)-s_k(\tau,t)\right| = O_p(n^{-1/2})$.
    \item[(C)] Lindeberg condition, unchanged.
    \item[(D)] Asymptotic regularity: For all $\tau\in\mathcal{B}$ and $t\in[0,t^\dagger]$, $s_1=\partial s_0/\partial\tau$ and $s_2=\partial s_1/\partial\tau$, and $s_0(\cdot,t)$, $s_1(\cdot,t)$, $s_2(\cdot,t)$ are continuous in $\tau\in\mathcal{B}$. Furthermore, $s_0$ is bounded away from zero on $\mathcal{B}\times[0,t^\dagger]$, and the information is positive and finite, that is,
    \begin{equation*}
        0< A(\tau_0) := \mathbb{E}\!\left[w_i\delta_i\left\{\frac{s_2(\tau_0,T_i)}{s_0(\tau_0,T_i)}-\pi^2(\tau_0,T_i)\right\}\right]<\infty,
    \end{equation*}
    which ensures a unique root $\tau_0$ to $u(\tau)=0$ together with R1.
\end{itemize}

\subsubsection{Proof of Theorem \ref{thm:asymptotic-variance}}\label{secapp:theorem1-proofs}

We state the proof under any nonnegative weights $w_i$ that are determined by $X_i$ and $Z_i$ and satisfy other regularity conditions. To this end, we denote $V_w=A(\tau_0)^{-2}B(\tau_0)$, and prove that $\sqrt{n}(\widehat\tau_n-\tau_0)$ is asymptotically normal with mean $0$ and variance $V_w$, where $A(\tau_0)= -\mathbb{E}\!\left[\psi_i^\prime(\tau_0) \right] =\int_0^{t^\dagger}\pi(\tau_0,t)\{1-\pi(\tau_0,t)\}\,d\mu(t)$,  $B(\tau_0)=\mathbb{E}[\eta_i^2(\tau_0)]$, and
\begin{align*}
    \psi_i^\prime(\tau_0) & = \frac{\partial\psi_i(\tau)}{\partial\tau} \bigg|_{\tau=\tau_0} = - \,w_i\delta_i\left[\frac{s_2(\tau_0,T_i)}{s_0(\tau_0,T_i)} - \pi^2(\tau_0,T_i)\right], \\
    \eta_i(\tau_0) &= w_i\zeta_i, \qquad \zeta_i=\delta_i\left[Z_i-\pi(\tau_0,T_i)\right] -\exp(\tau_0Z_i)\int_0^{T_i}\left[Z_i-\pi(\tau_0,t)\right] d\Lambda_w(t).
\end{align*}

\begin{proof}[of Theorem \ref{thm:asymptotic-variance}]

Under the regularity conditions and invoking the continuous mapping theorem under the sup-norm over the corresponding space,
\begin{equation*}
\sup_{t,\tau}\left|S_k^*(\tau,t)-s_k(\tau,t)\right| \xrightarrow{p}0,
\qquad
\sup_{t,\tau}\left|\pi_n^*(\tau,t)-\pi(\tau,t)\right| \xrightarrow{p}0,
\label{eq:unif}
\end{equation*}
where $t\in[0,t^\dagger]$ with $t^\dagger<\infty$, and $\tau\in\mathcal{B}$ for a neighborhood of $\tau_0$.

\underline{\textit{Consistency of $\widehat{\tau}_n$ for $\tau_0$.}} Recall that $\widehat{\tau}_n$ and $\tau_0$ are the solutions to $n^{-1}U_n(\tau)=0$ and $u(\tau)=\bE[\psi_i(\tau)]=0$, respectively. To show $\widehat{\tau}_n \xrightarrow{p} \tau_0$, it suffices to show that $n^{-1}U_n(\tau) \xrightarrow{p} u(\tau)$ pointwise in $\tau$, and that $U_n(\tau)$ is monotone with $U_n(\hat{\tau}_n)=0$, and $u(\tau_0-\epsilon)>0>u(\tau_0+\epsilon)$ for every $\epsilon>0$ (\cite{Vaart1998Asymptotics} Lemma 5.10).

The contribution of unit $i$ to the empirical partial score $U_n$ can be decomposed as
\begin{align*}
    \psi_i^*(\tau) 
    & = \psi_i(\tau) - w_i\delta_i\left[\pi_n^*(\tau, T_i) - \pi(\tau, T_i)\right],
\end{align*}
where $\psi_i(\tau)=w_i\delta_i\left[Z_i - \pi(\tau, T_i)\right]$, and therefore,
\begin{equation*}
    U_n(\tau) = \sum_{i=1}^{n} \psi_i^*(\tau) = \sum_{i=1}^{n}\psi_i(\tau) - \sum_{i=1}^{n}w_i\delta_i\left[\pi_n^*(\tau,T_i)-\pi(\tau,T_i)\right].
\end{equation*}
The first sum has its summands i.i.d. across units because $\psi_i(\tau)$ is a function of unit $i$ alone. Also note that $|\psi_i| \leq w_i$ because $Z_i$ is binary and $\pi(\tau, T_i) \in [0,1]$, and thus $\bE|\psi_i| \leq \bE[w_i] < \infty$ and weak law of large numbers (WLLN) applies, yielding $n^{-1}\sum_{i=1}^{n}\psi_i(\tau) \xrightarrow{p} u(\tau)$.

The second sum is bounded termwise as
\begin{align*}
    \left| n^{-1} \sum_{i=1}^{n}w_i\delta_i \left[\pi_n^*(\tau,T_i)-\pi(\tau,T_i)\right] \right| & \leq n^{-1} \sum_{i=1}^{n}w_i\delta_i \cdot \sup_{t,\tau}\left|\pi_n^*(\tau,t)-\pi(\tau,t)\right| \\
    & = \sup_{t,\tau}\left|\pi_n^*(\tau,t)-\pi(\tau,t)\right| \cdot n^{-1} \sum_{i=1}^{n}w_i\delta_i.
\end{align*}
Notice $\sup_{t,\tau}\left|\pi_n^*(\tau,t)-\pi(\tau,t)\right|=o_p(1)$ as established above, and $n^{-1} \sum_{i=1}^{n}w_i\delta_i = O_p(1)$ by WLLN, because its summands are i.i.d. with finite $\bE[w_i\delta_i]$. Therefore, $n^{-1}U_n(\tau) = u(\tau) + o_p(1) + o_p(1)O_p(1) = u(\tau) + o_p(1)$, yielding $n^{-1}U_n(\tau) \xrightarrow{p} u(\tau)$.

For $u(\tau) = \bE[\psi_i(\tau)]$, differentiating $\psi_i(\tau)$ yields 
\begin{equation*}
    \psi_i^\prime(\tau) = \frac{\partial\psi_i(\tau)}{\partial\tau} 
    = - \, w_i\delta_i\left[\frac{s_2(\tau,T_i)}{s_0(\tau,T_i)}-\pi^2(\tau,T_i)\right]
    = - \, w_i \delta_i \pi(\tau,T_i)\left[1-\pi(\tau,T_i)\right],
\end{equation*}
where the second equation is because $\pi=s_1/s_0$ and $s_{k+1} = \partial s_k/\partial \tau$ by regularity condition, and thus $\partial \pi /\partial \tau =s_2/s_0-(s_1/s_0)^2$. The last equation holds because $Z_i$ is binary and $Z_i=Z_i^2$, and thus $s_2=s_1$. Since $|\psi_i^\prime(\tau)| \leq w_i\delta_i$ and $\bE[w_i\delta_i]<\infty$, so
\begin{equation}\label{eqapp:u-derivative}
    u^\prime(\tau) = \frac{\partial \bE[\psi_i(\tau)]}{\partial\tau} = \bE\!\left[ \frac{\partial\psi_i(\tau)}{\partial\tau} \right] = - \bE\!\left[ w_i\delta_i\left\{\frac{s_2(\tau,T_i)}{s_0(\tau,T_i)}-\pi^2(\tau,T_i)\right\}\right].
\end{equation}
Hence $u^\prime(\tau) \leq 0$ for every $\tau$, and $u^\prime(\tau_0)=-A(\tau_0)<0$ by regularity condition, so $u$ is nonincreasing and strictly decreasing at $\tau_0$. Since $u(\tau_0)=0$, for every $\epsilon > 0$ it holds $u(\tau_0-\epsilon)>0>u(\tau_0+\epsilon)$; it suffices to consider $\epsilon$ such that $\tau_0 \pm \epsilon \in \mathcal{B}$ in a neighborhood of $\tau_0$ because $\Pr(|\widehat{\tau}_n-\tau_0| > \epsilon)$ is non-increasing in $\epsilon$.

It remains to show that $U_n(\tau)$ is monotone. Differentiating each $\psi_i^*(\tau)$ yields
\begin{align}\label{eqapp:Un-derivative}
    U_n^\prime(\tau) & = \sum_{i=1}^n \frac{\partial \psi_i^*(\tau)}{\partial \tau}= -\sum_{i=1}^n w_i\delta_i \left[\frac{S_2^*(\tau, T_i)}{S_0^*(\tau, T_i)} - \left(\frac{S_1^*(\tau, T_i)}{S_0^*(\tau, T_i)}\right)^2\right] \nonumber \\
    & =  -\sum_{i=1}^n w_i\delta_i \, \pi_n^*(\tau,T_i)\left[1-\pi_n^*(\tau,T_i)\right] \leq 0.
\end{align}
where similarly, $S_{k+1}^*=\partial S_k^*/\partial\tau$ and thus $\partial \pi_n^* /\partial \tau = S_2^*/S_0^*-(S_1^*/S_0^*)^2$, and $S_2^*=S_1^*$ since $Z_i$ is binary. Hence $U_n(\tau)$ is nonincreasing with $U_n(\widehat{\tau}_n)=0$, and $\widehat{\tau}_n \xrightarrow{p} \tau_0$ follows from Lemma 5.10 of \cite{Vaart1998Asymptotics}.

We can further show that $\widehat{\tau}_n$ is unique with probability tending to one. To see this, notice that for any $\tau$, $U_n^\prime(\tau)<0$ if and only if there exists at least one $i$ such that $w_i\delta_i>0$ and $0<\pi_n^*(\tau,T_i)<1$, that is, the risk set at $T_i$ has both arms with positive weights, which does not depend on $\tau$; denote this event by $\mathcal{D}_n$, and $\Pr(\mathcal{D}_n)=\Pr\{U_n^\prime(\tau) < 0\}$ for any $\tau$. 
Under $\mathcal{D}_n$, $U_n(\tau)$ is strictly decreasing and has at most one root, so $\hat{\tau}_n$ is unique. Furthermore, taking $\tau=\tau_0$ yields $n^{-1}U_n^\prime(\tau_0) \xrightarrow{p} u^\prime(\tau_0)= - A(\tau_0) < 0$, which implies $\Pr\{U_n^\prime(\tau_0) \!<\! 0\} \to 1$ and thus $\Pr(\mathcal{D}_n) \to 1$. Therefore, $\widehat{\tau}_n$ is unique with probability tending to one.

\underline{\textit{Asymptotic linearity.}} To write $n^{-1/2}U_n(\tau_0)$ as a sum of i.i.d. terms, first adopt the decomposition of $\psi_i^*$,
\begin{equation}\label{eq:Un-decompose}
    \frac{1}{\sqrt{n}} U_n(\tau_0) = \underbrace{\frac{1}{\sqrt{n}}\sum_{i=1}^n \psi_i(\tau_0)}_{\text{Term I}} \;-\; \underbrace{\frac{1}{\sqrt{n}}\sum_{i=1}^n w_i\delta_i\left[\pi_n^*(\tau_0, T_i) - \pi(\tau_0, T_i)\right]}_{\text{Term II}}.
\end{equation}
We have shown Term I is a sum of i.i.d. $\{\psi_i(\tau_0)\}_{i=1}^n$; furthermore, $\bE[\psi_i(\tau_0)]=0$ by definition and $\bV[\psi_i(\tau_0)] < \infty$ because $\psi_i^2$ is bounded by $w_i^2$. 

For Term II, apply a first-order expansion around $(s_1(\tau_0,t),s_0(\tau_0,t))$ to $\pi_n^*(\tau_0,t)$:
\begin{equation*}
    \pi_n^*(\tau_0,t) = \pi(\tau_0,t) 
    + \frac{1}{s_0(\tau_0,t)}\left[S_1^*(\tau_0,t)-s_1(\tau_0,t)\right]
    - \frac{s_1(\tau_0,t)}{s_0^2(\tau_0,t)}\left[S_0^*(\tau_0,t)-s_0(\tau_0,t)\right] 
    + o_p(n^{-1/2}),
\end{equation*}
where the remainder is $o_p(n^{-1/2})$ uniformly in $t$. Denote $S_k^*(\tau_0,t)-s_k(\tau_0,t)=n^{-1}\sum_{j=1}^{n}\xi_{j,k}(t)$ with $\xi_{j,k}(t)=w_jY_j(t)\exp(\tau_0Z_j)Z_j^k-s_k(\tau_0,t)$. By definition, $\mathbb{E}[\xi_{j,k}(t)]=0$ and $\xi_{j,k}(t)$ depends only on unit $j$. Substituting the expansion and $\xi_{j,k}(t)$ into Term II:
\begin{equation*}
    \text{Term II} = \frac{1}{n\sqrt n} \sum_{i=1}^{n}\sum_{j=1}^{n}
    \frac{w_i\delta_i}{s_0(\tau_0,T_i)}
    \left[\xi_{j,1}(T_i)-\pi(\tau_0,T_i)\xi_{j,0}(T_i)\right]+o_p(1).
\end{equation*}
Notice that $Y_j(T_i)=\bI(T_j\ge T_i)$ and $\xi_{j,1}(t)-\pi(\tau_0,t)\xi_{j,0}(t)=w_jY_j(t)\exp(\tau_0Z_j)[Z_j-\pi(\tau_0,t)]$ for all $t$. Plug them into Term II and exchange the order of summation:
\begin{equation*}
    \text{Term II} = \frac{1}{n\sqrt n}\sum_{j=1}^{n}\sum_{i=1}^{n} \frac{w_i\delta_i\cdot w_j\exp(\tau_0Z_j) \bI(T_i\le T_j) [Z_j-\pi(\tau_0,T_i)]}{s_0(\tau_0,T_i)}+o_p(1).
\end{equation*}
The above Term II has the form $n^{-3/2}\sum_{i,j}h(V_i,V_j)$, and is asymptotically equivalent to a scaled second-order U-statistic with asymmetric kernel $h$. Each diagonal term $h(V_i,V_i)$ is $O_p(1)$ with finite mean because $|h(V_i,V_i)| \leq C w_i^2$ for some constant $C>0$, since $\delta_i$, $\bI(T_i \!\leq\! T_j)$, and $Z_j-\pi(\tau_0,T_i)$ are bounded, and $s_0(\tau_0,T_i)$ is bounded away from zero. Therefore, their sum scaled by $n^{-3/2}$ is $o_p(1)$. A symmetric reconstruction of the kernel by $\widetilde{h} =[h(v_i,v_j)+h(v_j,v_i)]/2$ further yields
\begin{equation*}
    \text{Term II} = H_n+o_p(1), \qquad 
    H_n=\frac{1}{n\sqrt n}\sum_{i\ne j}\widetilde h(V_i,V_j).
\end{equation*}
By Lemma \ref{lemapp:2}, we decompose $H_n$ into a first-order projection and a degenerate remainder:
\begin{equation*}
    H_n=\frac{1}{\sqrt n}\sum_{i=1}^{n}w_i\exp(\tau_0Z_i)\int_0^{T_i}\left[Z_i-\pi(\tau_0,t)\right]d\Lambda_w(t)+o_p(1).
\end{equation*}
Substituting $H_n$ back to Term II and then to $U_n$ yields
\begin{equation}\label{eqapp:Un-linear}
    \frac{1}{\sqrt n}U_n(\tau_0) = \frac{1}{\sqrt n}\sum_{i=1}^{n}\eta_i(\tau_0)+o_p(1),
\end{equation}
where $\eta_i(\tau_0)$ is the population-level influence function
\begin{equation}\label{eqapp:population-influence}
    \eta_i(\tau_0) = \psi_i(\tau_0) -w_i\exp(\tau_0Z_i)\int_0^{T_i}\left[Z_i-\pi(\tau_0,t)\right]d\Lambda_w(t).
\end{equation}
Both $\pi(\tau_0,t)$ and $\Lambda_w(t)$ are deterministic functions of $t$, so $\eta_i(\tau_0)$ depends only on unit $i$, and thus $\{\eta_i(\tau_0)\}_{i=1}^{n}$ are i.i.d. Moreover, $\bE[\eta_i(\tau_0)]\!=\!0$ (mean-zero, Lemma \ref{lemapp:3}) and $\bE[\eta_i^2(\tau_0)] \!<\! \infty$ (finite-variance), because $|\eta_i(\tau_0)| \leq w_i\{1+\max(1,e^{\tau_0})\Lambda_w(t^\dagger)\} \leq C w_i$ and $\bE[w_i^2]\!<\! \infty$. By standard central limit theorem (CLT), $n^{-1/2}U_n(\tau_0)$ converges to a normal distribution with mean zero and variance $\bE[\eta_i^2(\tau_0)]$.

\underline{\textit{Asymptotic distribution.}} Expand the estimating equation $U_n(\widehat{\tau}_n)=0$ around $\tau_0$:
\begin{equation*}
0=U_n(\widehat\tau_n)=U_n(\tau_0)+U_n^\prime(\widetilde{\tau}_n)(\widehat\tau_n-\tau_0),
\end{equation*}
where $\widetilde{\tau}_n$ is between $\widehat{\tau}_n$ and $\tau_0$, guaranteed by the mean value theorem; therefore, $\widetilde{\tau}_n \xrightarrow{p} \tau_0$ by $\widehat{\tau}_n \xrightarrow{p} \tau_0$. Further applying ULLN to $n^{-1} U_n^\prime(\widetilde{\tau}_n)$ yields
\begin{equation}\label{eqapp:A}
    \frac{1}{n}U_n^\prime(\widetilde{\tau}_n) \xrightarrow{p} -A(\tau_0), \quad
    A(\tau_0) = -\mathbb{E}\! \left[\psi_i^\prime(\tau_0)\right] 
    = \mathbb{E}\! \left[w_i\delta_i\left(\frac{s_2(\tau_0,T_i)}{s_0(\tau_0,T_i)}-\pi^2(\tau_0,T_i)\right)\right] > 0.
\end{equation}
Rearranging the expansion above and applying Slutsky's theorem yields
\begin{equation*}
    \sqrt n(\widehat{\tau}_n - \tau_0) = \left[-\frac{1}{n}U_n^\prime(\widetilde{\tau}_n)\right]^{-1}\!\! \cdot \frac{1}{\sqrt n}U_n(\tau_0) \; \xrightarrow{d}\; \mathcal{N}(0,V_w),
\end{equation*}
where $V_w=A(\tau_0)^{-2}B(\tau_0)$ and $B(\tau_0)=\mathbb{E}[\eta_i^2(\tau_0)]$, establishing Theorem \ref{thm:asymptotic-variance}. The integral form $A(\tau_0)=\int_0^{t^\dagger}\pi(\tau_0,t)\{1-\pi(\tau_0,t)\}\,d\mu(t)$ follows from Lemma \ref{lemapp:4}.
\end{proof}

\subsubsection{Lemmas}

\begin{lemma}\label{lemapp:1}
Under regularity conditions R1-R4, the weighted event measure $\mu$ defined in equation \eqref{eqapp:weighted-event-measure} is a finite Borel measure on
$[0,t^\dagger]$ with total mass $\mathbb{E}[w_i\delta_i]<\infty$. Furthermore, for $\mu(t) := \mu([0,t])$ with $t \in [0,t^\dagger]$ and the weighted cumulative hazard $\Lambda_w(t)$ defined in equation \eqref{eqapp:weighted-cumulative-hazard}, it holds that
\begin{equation*}
    \mathbb{E}\!\left[w_i\delta_i f(T_i)\right] = \int_0^{t^\dagger}f(t) \,d\mu(t) =\int_0^{t^\dagger} f(t) s_0(\tau_0,t) \,d\Lambda_w(t),
\end{equation*}
where $f:[0,t^\dagger] \to \mathbb{R}$ is any bounded Borel-measurable non-random function.
\end{lemma}

\begin{proof}[of Lemma \ref{lemapp:1}]
Recall that $\mu(B) := \mathbb{E} \! \left[w_i \delta_i \, \bI(T_i\in B) \right]$ for any Borel set $B \subseteq [0,t^\dagger]$, $\mu(t):= \mu([0,t]) = \mathbb{E}[w_iN_i(t)]$, and $\Lambda_w(t):=\int_{[0,t]} s_0^{-1}(\tau_0,s)d\mu(s)$ for $t\in[0,t^\dagger]$.

\underline{\textit{The set function $\mu$ is a finite measure.}} First, it is immediate that $\mu$ is non-negative because $w_i \!\geq\! 0$, and that $\mu(\varnothing) \!=\! 0$ because $\bI(T_i \!\in\! \varnothing) \!=\! 0$. Second, for any pairwise disjoint $B_1, B_2, \ldots \in \mathscr{B}([0,t^\dagger])$, $T_i$ can lie in at most one $B_k$, and thus $\bI(T_i \in \bigcup_{k=1}^\infty B_k) = \sum_{k=1}^\infty \bI(T_i \in B_k)$. Since $w_i \delta_i \geq 0$ and each of $\{w_i \delta_i \bI(T_i\in B_k)\}_{k =1}^\infty$ is a non-negative random variable, monotone convergence theorem thus implies 
\begin{equation*}
    \bE\!\left[ \sum_{k=1}^\infty w_i \delta_i \bI(T_i\in B_k) \right] = \sum_{k=1}^\infty \bE\!\left[w_i \delta_i \bI(T_i\in B_k)\right],
\end{equation*}
where the left-hand side is $\bE\!\left[w_i \delta_i \sum_{k=1}^\infty \bI(T_i\in B_k) \right] = \bE\!\left[w_i \delta_i \bI(T_i \in \bigcup_{k=1}^\infty B_k) \right]=\mu(\bigcup_{k=1}^\infty B_k)$, and the right-hand side equals $\sum_{k=1}^\infty \mu(B_k)$. Therefore, $\mu(\bigcup_{k=1}^\infty B_k)=\sum_{k=1}^\infty \mu(B_k)$ and thus $\mu$ is $\sigma$(countable)-additive. Lastly, $\mu([0,t^\dagger]) = \bE\!\left[w_i \delta_i \bI(T_i\leq t^\dagger) \right] = \bE\!\left[w_i \delta_i \right]$, which is finite by regularity condition, so $\mu$ is finite with total mass of $\bE\!\left[w_i \delta_i \right]$. In summary, $\mu$ is a finite Borel measure. Similar arguments establish $\mu_z$ is also a finite Borel measure for $z=0,1$, with total mass $\mu_1(t^\dagger)=\bE[w_i \delta_i Z_i]$ and $\mu_0(t^\dagger)=\mu(t^\dagger)-\mu_1(t^\dagger)=\bE[w_i \delta_i (1-Z_i)]$, respectively.

\underline{\textit{The first equality.}} We now show that $\mathbb{E}\!\left[w_i\delta_i f(T_i)\right] = \int_0^{t^\dagger} f(t) \,d\mu(t)$ for any $f$ stated above. We first clarify the integral on the right-hand side, for which we temporarily write $g(t)$ in place of $\mu(t)$ for clarity. Because $\mu$ is a Borel measure and $\mu((a,b]) = g(b)-g(a)$ for all $0\leq a<b \leq t^\dagger$ and $g(0)=\mu(\{0\})=0$, so $g$ is the \emph{unique} right-continuous and increasing function that satisfies these two relations with $\mu$, and its Lebesgue–Stieltjes integral is defined as $\int_B f dg := \int_B f d\mu$ for any Borel set $B$ (\cite{fleming1997Counting}, Appendix A.1). By convention, $\int_0^t f(s) \,dg(s):=\int_{(0,t]} f dg$, and since $\mu(\{0\})=0$, then $\int_0^{t} f(s)\,dg(s) = \int_{(0,t]} fd\mu+f(0)\mu(\{0\})=\int_{[0,t]} f d\mu$. Therefore, we write $\mu(t)$ instead of $g(t)$ henceforth with a slight abuse of notation, and use $d\mu$ and $d\mu(t)$ interchangeably. We also write $\int$ to denote $\int_0^{t^\dagger}$ and $\int_{[0,t^\dagger]}$, for conciseness.

First, for any indicator function $f(u) \!=\! \bI_B(u) \!=\! \bI(u \!\in\! B)$ where $B \in \mathscr{B}([0,t^\dagger])$, the integration is defined as $\int\bI_B\,d\mu := \mu(B)=\mathbb{E}\!\left[w_i\delta_i \bI(T_i \in B) \right]$, so the equality holds. By linearity, the equality extends to simple function, $f=\sum_{k=1}^Ka_k\bI_{B_k}$ with $a_k\ge0$ and Borel sets $B_k$, as
\begin{align*}
    \int f\,d\mu & = \int \sum_{k=1}^K a_k \bI_{B_k}\,d\mu = \sum_{k=1}^K a_k \mu(B_k) = \sum_{k=1}^K a_k \bE[w_i\delta_i \bI(T_i\in B_k)] \\
    & = \bE\!\left[ w_i\delta_i \sum_{k=1}^K a_k \bI(T_i\in B_k) \right] = \bE\!\left[ w_i \delta_i f(T_i) \right].
\end{align*}
Then, for a nonnegative $f \!\leq\! M \!<\! \infty$, take a sequence of simple functions for which $0 \leq f_m \uparrow f$ pointwise, and hence $w_i\delta_i f_m \uparrow w_i\delta_i f$ since $w_i\delta_i \geq 0$. Applying monotone convergence theorem on $(\Omega, \mathscr{F}, \Pr)$ and $([0,t^\dagger], \mathscr{B}([0,t^\dagger]), \mu)$ respectively yields
\begin{equation*}
    \bE[w_i \delta_i f_m(T_i)]\uparrow\bE[w_i \delta_i f(T_i)] \quad \text{and} \quad \int f_m\,d\mu \uparrow \int f\,d\mu.
\end{equation*}
Because $\bE[w_i \delta_i f_m(T_i)]=\int f_m\,d\mu$ for every $m$ since $f_m$ is simple function, the two sequences must converge to the same limit, and therefore $\bE[w_i \delta_i f(T_i)] = \int f\,d\mu$. Finally, for any bounded $f$ that $|f| \leq M <\infty$, decompose $f=f^+ - f^-$ into positive and negative parts, where $f^+$ and $f^-$ are nonnegative, measurable, and bounded by $M$. Applying the equality on nonnegative function to each part yields
\begin{equation*}
    \bE[w_i \delta_i f(T_i)] = \bE[w_i \delta_i f^+(T_i)] - \bE[w_i \delta_i f^-(T_i)] = \int f^+\,d\mu-\int f^-\,d\mu = \int f\,d\mu.
\end{equation*}

\underline{\textit{The second equality.}} Notice that the function $t \mapsto s_0(\tau_0, t) = \bE\left[w_j Y_j(t)\exp(\tau_0 Z_j)\right]$ is non-increasing in $t$, because $Y_j(b) \!=\! \bI(T_j \!\geq\! b) \leq \bI(T_j \!\geq\! a) \!=\! Y_j(a)$ for any $a\!<\!b$ and $w_j\exp(\tau_0 Z_j) \geq 0$, and thus $s_0(\tau_0, t)$ is Borel-measurable. By regularity conditions, it is bounded away from zero, that is, $0 < c_0 \leq s_0(\tau_0, t)$ for some $c_0$, and is bounded as $s_0(\tau_0, t) \leq \max\{1,\exp(\tau_0)\} \cdot \bE|w_i| = c_1 < \infty$. Therefore, $0< c_1^{-1} \leq s_0^{-1}(\tau_0, t) \leq c_0^{-1} < \infty$.

With the same slight abuse of notation as for $\mu$, let
\begin{equation*}
    \Lambda_w(B):= \int_{B} s_0^{-1}(\tau_0, t) d\mu(t),\qquad B \in \mathscr{B}([0,t^\dagger]).
\end{equation*}
In fact, since $s_0^{-1}(\tau_0, t)$ is non-negative, measurable, and integrable with respect to $\mu$, $\Lambda_w$ is also a finite measure on $([0,t^\dagger],\mathscr{B}([0,t^\dagger]))$, and $s_0^{-1}(\tau_0, t)$ is called the \emph{density} of $\Lambda_w$ with respect to $\mu$ \citep[][Theorem \ref{thm:asymptotic-variance}6.9]{billingsley1995probability}. As with $\mu(t)$, the function $\Lambda_w(t):=\Lambda_w([0,t])=\int_0^t s_0^{-1}(\tau_0, s) d\mu(s)$ is the unique right-continuous increasing function that satisfies $\Lambda_w(0)=0$ and $\Lambda_w((a,b])=\Lambda_w(b)-\Lambda_w(a)$ for all $0 \leq a<b \leq t^\dagger$, and $\int_B f(t) \,d\Lambda_w(t)$ is defined as $\int_B f \, d\Lambda_w$. Henceforth, we also use $d\Lambda_w$ and $d\Lambda_w(t)$ interchangeably.

Because any bounded Borel $f$ is $\mu$-integrable, and can be decomposed as $f(t)=f(t)s_0(\tau_0, t) \!\cdot\! s_0^{-1}(\tau_0, t)$, applying Theorem \ref{thm:asymptotic-variance}6.11 of \cite{billingsley1995probability} to $f(t)s_0(\tau_0, t)$ yields
\begin{equation*}
    \int f(t) s_0(\tau_0, t) \, d\Lambda_w(t) = \int f(t) s_0(\tau_0, t) \cdot s_0^{-1}(\tau_0, t) \, d\mu(t) = \int f(t) \, d\mu(t).
\end{equation*}
In fact, applying the same theorem to $f(t)$ by seeing that $f(t)s_0^{-1}(\tau_0, t)$ is $\mu$-integrable yields $\int f(t) d\Lambda_w(t) = \int f(t) s_0^{-1}(\tau_0, t)d\mu(t)$, which by convention is denoted as $d\Lambda_w=s_0^{-1}(\tau_0, t)\,d\mu$, which we will use later as the change of integrator.

Combining the three steps above establishes Lemma \ref{lemapp:1}.
\end{proof}

\begin{lemma}\label{lemapp:2}
    For the scaled U-statistic $H_n=n^{-3/2}\sum_{i \neq j}^n \widetilde{h}(V_i,V_j)$, where
    \begin{align*}
        \widetilde h(V_i,V_j) & = \left[h(V_i,V_j)+h(V_j,V_i)\right]/2, \\
        h(V_i,V_j) & = w_i\delta_i \cdot w_j\exp(\tau_0Z_j)\cdot \bI(T_i\le T_j)\left[Z_j-\pi(\tau_0,T_i)\right]/s_0(\tau_0,T_i),
    \end{align*}
    it holds under regularity conditions R1-R4 that
    \begin{equation*}
     H_n = \frac{1}{\sqrt n} \sum_{i=1}^{n}w_i\exp(\tau_0Z_i)\int_0^{T_i}\left[Z_i-\pi(\tau_0,t)\right]d\Lambda_w(t) + o_p(1).
    \end{equation*}
\end{lemma}

\begin{proof}[of Lemma \ref{lemapp:2}]
Rewrite the scaled U-statistic $H_n$ as
\begin{equation}\label{eqapp:standard-Ustatistic}
    H_n=n^{-3/2}\sum_{i\ne j}\widetilde h(V_i,V_j)=\frac{\sqrt n(n-1)}{n}\widetilde H_n, \qquad
    \widetilde H_n=\frac{1}{\binom{n}{2}}\sum_{i<j}\widetilde h(V_i,V_j),
\end{equation}
where $\widetilde H_n$ is a U-statistic with symmetric kernel $\widetilde h$. Notice that $V_i$ and $V_j$ are two independent draws of generic units from the population, and $\widetilde h$ is a bivariate function of them. Since $0 < c_0 \leq s_0(\tau_0,\cdot)$ for some $c_0$, then $|h(V_i,V_j)| \leq C w_iw_j$ where $C = \max(1,e^{\tau_0}) c_0^{-1} < \infty$ and $\bE[h^2] \leq C^2 \bE[w_i^2] \bE[w_j^2] < \infty$, and hence $\bE[\widetilde{h}^2]<\infty$.

Let $\bE_i$ indicate taking expectation over an independent draw of $V_i$ from the population while treating all other components in the integrand as fixed, and vice versa for $\bE_j$. Denote $\theta=\bE[h(V_i,V_j)]$ where the expectation is taken over two generic i.i.d draws $\{V_i, V_j\}$, and define the $i$-projection $h_1^{(i)}$ and the $j$-projection $h_1^{(j)}$:
\begin{equation*}
    h_1^{(i)}(V_i) = \bE_j[h(V_i, V_j)] - \theta, \quad 
    h_1^{(j)}(V_j) = \bE_i[h(V_i, V_j)] - \theta,
\end{equation*}
where $h_1^{(i)}$ fixes the \emph{first} argument of the bivariate kernel $h$ and takes expectation over the \emph{second} argument, and $h_1^{(j)}$ is in reverse. Both are univariate functions of a generic unit rather than defined for a specific unit. The superscripts $(i)$ and $(j)$ only indicate which argument (the first or the second) of $h$ is fixed. We now evaluate $\theta$, $h_1^{(i)}(V_i)$, and $h_1^{(j)}(V_j)$.

For $\theta$, by the law of total expectation, $\theta = \bE_i\!\left[ \, \bE_j\{h(V_i,V_j)\} \,\right]$, and
\begin{align*}
    \bE_j\!\left[h(V_i,V_j)\right] & =\mathbb{E}_j\!\left[\frac{w_i\delta_i\cdot w_j \exp(\tau_0Z_j) \,Y_j(T_i) \, \left\{Z_j - \pi(\tau_0,T_i)\right\}}{s_0(\tau_0,T_i)}\right] \\
    & = \frac{w_i\delta_i}{s_0(\tau_0,T_i)} \mathbb{E}_j\!\left[ w_j \exp(\tau_0Z_j) Y_j(T_i)Z_j -  w_j \exp(\tau_0Z_j) Y_j(T_i) \pi(\tau_0,T_i) \right] \\
    & = \frac{w_i\delta_i}{s_0(\tau_0,T_i)} \left[ s_1(\tau_0,T_i) - s_0(\tau_0,T_i)\pi(\tau_0,T_i) \right] = 0,
\end{align*}
where the last equation is from $\pi :=\! s_1/s_0$ and hence $s_1(\tau,t) -s_0(\tau,t)\pi(\tau,t)=0$ pointwise at all $t$ for any $\tau$. Therefore, $\theta \!=\! \bE_i[0]\!=\! 0$, and hence $h_1^{(i)}(V_i)=\bE_j[h(V_i, V_j)] - \theta=0$; because $V_i$ is a generic unit, equivalently, $h_1^{(i)} \equiv 0$ or $h_1^{(i)}(v) = 0$ for all $v$.

For $h_1^{(j)}(V_j)$, the expectation is over $V_i$ with $V_j$ held fixed. Given $V_j$, the function $f(t):=\bI(T_j\ge t)[Z_j-\pi(\tau_0,t)]/s_0(\tau_0,t)$ is deterministic in $t$, and is bounded and Borel-measurable because $s_0^{-1}(\tau_0,t) \leq c_0^{-1} < \infty$ and all its components are Borel; hence Lemma \ref{lemapp:1} applies:
\begin{align*}
    h_1^{(j)}(V_j) & = w_j\exp(\tau_0Z_j) \, \mathbb{E}_i\! \left[ w_i\delta_i\cdot \frac{\bI(T_j\ge T_i)[Z_j-\pi(\tau_0,T_i)]}{s_0(\tau_0,T_i)}\right]-0\\
    & = w_j\exp(\tau_0Z_j) \int_0^{t^\dagger} \frac{\bI(T_j\ge t)\left[Z_j-\pi(\tau_0,t)\right]}{s_0(\tau_0,t)}\,s_0(\tau_0,t)\,d\Lambda_w(t)\\
    & = w_j\exp(\tau_0Z_j)\int_0^{T_j}\left[Z_j-\pi(\tau_0,t)\right]d\Lambda_w(t).
\end{align*}

Notice that $V_i$ and $V_j$ are i.i.d. Hence, $\widetilde{\theta} := \bE[\widetilde{h}(V_i,V_j)] \!=\! 0$ since $\theta \!=\! 0$. Following \cite{Vaart1998Asymptotics} Chapter 12, define the first-order projection of $\widetilde{h}$ as $\widetilde{h}_1(V_i) := \bE_j[\widetilde{h}(V_i, V_j)]$. Then,
\begin{equation*}
    \widetilde{h}_1(V_i) = \bE_j \left[ \frac{h(V_i,V_j)+h(V_j,V_i)}{2} \right] = \frac{h_1^{(i)}(V_i) + h_1^{(j)}(V_i)}{2} = \frac{h_1^{(j)}(V_i)}{2},
\end{equation*}
where the second and last equations follow from the definition of $h_1^{(i)}$ and $h_1^{(j)}$, and from $h_1^{(i)} \equiv 0$. Similarly, $\widetilde{h}_1$ is a function of a generic unit, which is denoted by $V_i$. Further use \cite{Vaart1998Asymptotics} Theorem \ref{thm:asymptotic-variance}2.3,
\begin{equation*}
     \sqrt{n} \left(\widetilde{H}_n - \widehat{H}_n\right) = o_p(1),  \quad \widetilde{H}_n = O_p(n^{-1/2}),
\end{equation*}
where $\widehat{H}_n = 2n^{-1}\sum_{i=1}^n \widetilde{h}_1(V_i)$, and therefore
\begin{equation*}\label{eq:Vaart-theorem12.3-b}
    \sqrt{n} \widetilde{H}_n = \frac{2}{\sqrt{n}} \sum_{i=1}^n \widetilde{h}_1(V_i) + o_p(1) = \frac{1}{\sqrt{n}} \sum_{i=1}^n h_1^{(j)}(V_i) + o_p(1), \ \ \frac{1}{\sqrt{n}}\widetilde{H}_n = O_p(n^{-1}) = o_p(1).
\end{equation*}
Substituting $\sqrt{n} \widetilde{H}_n$ and $n^{-1/2}\widetilde{H}_n$ above into $H_n$ of equation \eqref{eqapp:standard-Ustatistic} leads to
\begin{equation*}
     H_n = \frac{(n-1)}{n} \sqrt{n} \widetilde{H}_n = \sqrt{n} \widetilde{H}_n - \frac{1}{\sqrt{n}} \widetilde{H}_n= \frac{1}{\sqrt{n}} \sum_{i=1}^n h_1^{(j)}(V_i) + o_p(1),
\end{equation*}
and plugging in $h_1^{(j)}(V_i)=w_i\exp(\tau_0Z_i)\int_0^{T_i}\left[Z_i-\pi(\tau_0,t)\right]d\Lambda_w(t)$ completes the proof.
\end{proof}

\begin{lemma}\label{lemapp:3}
Under regularity conditions R1-R4, the population influence function $\eta_i(\tau_0)$ is mean-zero at $\tau_0$: $\bE[\eta_i(\tau_0)]=0$, where $\eta_i(\tau_0) = \psi_i(\tau_0) - w_i\exp(\tau_0Z_i)\int_0^{T_i}\left[Z_i-\pi(\tau_0,t)\right]d\Lambda_w(t)$.
\end{lemma}

\begin{proof}[of Lemma \ref{lemapp:3}]
Because $\bE[\psi_i(\tau_0)]=0$ by definition,
\begin{equation*}
    \bE[\eta_i(\tau_0)] = - \bE\!\left[ w_i\exp(\tau_0Z_i) \int_0^{T_i}\left\{Z_i-\pi(\tau_0,t)\right\} d\Lambda_w(t) \right].
\end{equation*}
For the integral, notice that $\int_0^{T_i} \left[Z_i-\pi(\tau_0,t)\right]  d\Lambda_w(t) = \int_0^{t^\dagger} \left[Z_i-\pi(\tau_0,t)\right]Y_i(t)  d\Lambda_w(t)$ using $Y_i(t)=\bI(t \leq T_i)$. The integrand is absolutely integrable, since
\begin{equation*}
    \bE\!\left[\int_0^{t^\dagger}|w_i\exp(\tau_0Z_i)Y_i(t)\{Z_i-\pi(\tau_0,t)\}|\,d\Lambda_w(t) \right] \leq
    \max(1,e^{\tau_0})\bE[w_i]\Lambda_w(t^\dagger) < \infty
\end{equation*}
by $\bE[w_i]<\infty$ and $\Lambda_w$ being a finite measure from Lemma \ref{lemapp:1}. Fubini's theorem hence applies:
\begin{align*}
    \bE[\eta_i(\tau_0)] & = - \bE\!\left[\int_0^{t^\dagger} w_i\exp(\tau_0Z_i) \left\{Z_i-\pi(\tau_0,t)\right\} Y_i(t) d\Lambda_w(t) \right] \\
    & = - \int_0^{t^\dagger} \Big[
    \underbrace{\bE\{w_i\exp(\tau_0Z_i)Y_i(t)Z_i\}}_{=\,s_1(\tau_0,t)} 
    - \pi(\tau_0,t) \underbrace{\bE\{w_i\exp(\tau_0Z_i)Y_i(t)\}}_{=\,s_0(\tau_0,t)}
    \Big] \,d\Lambda_w(t) \\
    & = - \int_0^{t^\dagger}\left[s_1(\tau_0,t)-\pi(\tau_0,t)s_0(\tau_0,t)\right] \, d\Lambda_w(t) = 0 .
\end{align*}
\end{proof}

\subsubsection{Relations to the asymptotic variance in Lin \& Wei (1989)}

The asymptotic variance of \cite{lin1989robust} is the special case of Theorem \ref{thm:asymptotic-variance} obtained by setting $w_i=w(Z_i,X_i)\equiv1$. Under $w_i\equiv1$, our $S_k^*(\tau,t)$ and $s_k(\tau,t)$ reduce to their $S^{(k)}(\tau,t)$ and $s^{(k)}(\tau,t)$.
Applying Lemma \ref{lemapp:1} with $f(t)=s_2(\tau_0,t)/s_0(\tau_0,t)-\pi^2(\tau_0,t)$ to $A(\tau_0)$ yields
\begin{equation*}
    A(\tau_0)=\int_0^{t^\dagger}\left[\frac{s^{(2)}(\tau_0,t)}{s^{(0)}(\tau_0,t)} - \left(\frac{s^{(1)}(\tau_0,t)}{s^{(0)}(\tau_0,t)}\right)^{\!2}\right] s^{(0)}(\tau_0,t) \,d\Lambda_w(t),
\end{equation*}
which coincides with their information (matrix).

Their population influence function is given as,
\begin{equation*}
    \eta_i(\tau_0) = \int_0^{t^\dagger} \left[ Z_i - \frac{s^{(1)}(\tau_0,t)}{s^{(0)}(\tau_0,t)} \right] dN_i(t) - 
    \int_0^{t^\dagger} \frac{Y_i(t)\exp(\tau_0 Z_i)}{s^{(0)}(\tau_0,t)} \left[ Z_i - \frac{s^{(1)}(\tau_0,t)}{s^{(0)}(\tau_0,t)} \right] \,dF(t),
\end{equation*}
where $F(t):=\mathbb{E}[\sum_{i=1}^{n}N_i(t)/n]$. Since $\{N_i(t)\}_{i=1}^{n}$ at any $t$ are i.i.d., $F(t)=\mathbb{E}[N_i(t)]$, and hence $F(t)=\mu(t)$, that is, the unique Borel measure induced by $F(t)$ is our weighted event measure $\mu$ with weight $w_i \equiv 1$; therefore, $dF(t)=s^{(0)}(\tau_0,t)\,d\Lambda_w(t)$ by Lemma \ref{lemapp:1}. Applying this to the second piece yields
\begin{equation*}
    \int_0^{t^\dagger} \frac{Y_i(t)\exp(\tau_0 Z_i)}{s^{(0)}(\tau_0,t)} \left[ Z_i - \frac{s^{(1)}(\tau_0,t)}{s^{(0)}(\tau_0,t)} \right] \,dF(t) = 
    \exp(\tau_0Z_i)\int_0^{T_i}\left[Z_i-\pi(\tau_0,t)\right]d\Lambda_w(t).
\end{equation*}
Since $N_i(t)$ jumps exactly once at $T_i$ when $\delta_i=1$, the first piece
\begin{equation*}
    \int_0^{t^\dagger} \left[ Z_i - \frac{s^{(1)}(\tau_0,t)}{s^{(0)}(\tau_0,t)} \right] dN_i(t) = \delta_i\left[Z_i-\pi(\tau_0,T_i)\right] = \psi_i(\tau_0).
\end{equation*}
Therefore, the population influence function also coincides, so does the variance.

\subsubsection{Number of events and number of units}

Sample size formulas for survival outcomes have been traditionally stated in the required number of events rather than the number of units. With weight $w$ and $n$ units, it is the number of weighted events, $\sum_{i=1}^n w_i \delta_i$. To convert between these two scales, by Theorem \ref{thm:asymptotic-variance},
\begin{equation*}
    \sqrt{\sum_i w_i\delta_i} (\widehat{\tau}_n - \tau_0) = \sqrt{n^{-1}\sum_i w_i\delta_i} \cdot
    \sqrt{n} (\widehat{\tau}_n - \tau_0) \ \xrightarrow{d} \ \sqrt{\bE[w_i\delta_i]} \mathcal{N}(0, V_w) \overset{d}{=} \mathcal{N}(0, \bE[w_i\delta_i] V_w).
\end{equation*}
Therefore, the required number of units is that of weighted events divided by $\mu(t^\dagger)=\bE[w_i\delta_i]$, the combined event rate of the target population, detailed below in Section \ref{secapp:weight-and-population}.

\subsection{Theorem 2}

In this section we prove Theorem \ref{thm:decomposition} and its associated lemmas. Conditions R1–R4 and assumptions A1–A4 are maintained throughout. We introduce the propensity score balancing weights and target population (\ref{secapp:weight-and-population}), reduce the quantities in Section \ref{secapp:theorem1} under the proportional hazards alternative (H1) in the target population (\ref{secapp:marcox-model}), give the main proof (\ref{secapp:theorem2-proof}) and lemmas (\ref{secapp:theorem2-lemmas}), and interpret the design effects (\ref{secapp:design-effect}).

\subsubsection{Balancing weights and the target population}\label{secapp:weight-and-population}

\textbf{Propensity score balancing weights.} Theorem \ref{thm:asymptotic-variance} applies to any weight $w$ that is a function of covariates $X$ and treatment $Z$, and satisfies the regularity conditions. Theorem \ref{thm:decomposition}, instead, considers $w$ from the class of propensity score balancing weights \citep{li2018balancing}. 

A balancing weight $w$ is determined by a tilting function $h(x)$, with arm-specific weight functions $w_z(x)\propto h(x)/e_z(x)$, where $e_1(x)=e(x)$ and $e_0(x)=1-e(x)$. Throughout, we restrict to $h(x)$ that is a sole function of $e(x)$, bounded on $(0,1)$, so that $w_z(x)$ depends on $x$ only through $e(x)$; it covers the inverse probability weights ($h \equiv 1$), ATT weights ($h(x)=e(x)$), and overlap weights ($h(x)=e(x)(1-e(x))$), which target the full study population (ATE), the treated population (ATT) and the overlap population (ATO), respectively, all with $h$ continuous and positive on $(0,1)$. The weight $w$ is implemented with normalization within each arm,
\begin{equation}\label{eqapp:balancing-weight}
    w_1(x) = \frac{h(x) / e_1(x)}{m/r_1}, \qquad w_0(x) = \frac{h(x) / e_0(x)}{m / r_0}
\end{equation}
where $r_1=r$, $r_0=1-r$, and $m=\bE[h(X)]$; specifically, the normalizer $m/r_z$ is the population average of the unnormalized weight $h(x)/e_z(x)$ in arm $z$ by a Bayes rule. The regularity condition $\bE[w_i^2]<\infty$ in R2 can be written as
\begin{equation*}
    \bE[w_i^2] = \bE\!\left[ \bE\!\left\{ w_i^2 \mid X_i \right\} \right] = \bE\!\left[ \sum_z e_z(X) w_z^2(X) \right] = r^2\nu_1 + (1-r)^2\nu_0 < \infty,
\end{equation*}
which holds if and only if $\nu_1$ and $\nu_0$ are finite, where $\nu_z=\bE[h^2/e_z]/m^2$, since $e_zw_z^2= e_z (h/e_z)^2/(m/r_z)^2 = r_z^2 \times (h^2/e_z)/m^2$. The second equation is from conditioning on $X_i$ and then summing over $Z_i$ with probability $\Pr(Z_i=z|X_i=x)=e_z(x)$.

\textbf{Target population.} The weight $w$ defines a target population whose covariate distribution is tilted by $h$, that is, $dP^w=h\,dP/m$, where $P$ and $P^w$ denote the probability measures over the study population from which the study sample is drawn, and the target population of the weight $w$, respectively. Hence $\bE^w[g(X)]=\bE[h(X)g(X)]/m$ for every integrable $g$. We write $\bE^w$, $\bV^w$ and $\Cov^w$ for moments under $P^w$. 
Weighting does not change the conditional laws given covariates $X$, which hence do not carry superscripts. We define the potential survival, density, hazard and cumulative hazard functions in the target population as
\begin{equation*}
    F_z^w(t)=\bE^w\!\left[F_z(t\mid X)\right],\quad
    f_z^w=-\,dF_z^w/dt,\quad
    \lambda_z^w=-\left(\log F_z^w\right)^\prime,\quad
    \Lambda_0^w(t)=\int_0^t\lambda_0^w(u)\,du,
\end{equation*}
where $F_z(t | X)=\Pr(T_i(z) \!\geq\! t \mid X_i)$ as in Section 2.1. Let $f_{T(z)\mid X}$ be the conditional density of $T_i(z)$ given $X_i$, which exists by regularity condition R3 and assumptions A1-A3. Then, $f_z^w(t) = \bE^w[f_{T(z)\mid X}(t)]$ since $F_z(t | X)=\int_t^{\infty}f_{T(z)\mid X}(u)\,du$ and 
\begin{equation*}
    F_z^w(t)=\bE^w\!\left[F_z(t | X)\right]
    = \int_t^{\infty}\!\bE^w\!\left[f_{T(z)\mid X}(u)\right]du,
\end{equation*}
where the second equation is from Fubini. In contrast, $\lambda_z^w$ and $\Lambda_0^w$ are not the averages of their conditional counterparts given $X_i$ over the target population, due to non-collapsibility.

The weight $w$ also changes the event rates in the target population. Recall $G_z(t)=\Pr(C_i \geq t | Z_i=z)$. 
The conditional event rate of arm $z$ given covariates is
\begin{align}\label{eqapp:conditional-event-rate}
    d_z(x) & : = \Pr(T_i^* \leq C_i \wedge t^\dagger \mid Z_i=z, X_i=x) 
    = \Pr(T_i(z) \leq C_i \wedge t^\dagger \mid Z_i=z, X_i=x) \nonumber \\
    & = \bE\!\left[ \bI\{T_i(z) \leq t^\dagger\} \Pr\{ T_i(z) \leq C_i \mid  T_i(z), Z_i=z, X_i=x\} \mid Z_i=z, X_i=x \right] \nonumber \\
    & = \bE\!\left[ \bI\{T_i(z) \leq t^\dagger\} \, G_z(T_i(z)) \mid Z_i=z,X_i=x \right]
    = \int_0^{t^\dagger}\!f_{T(z)\mid X=x}(t)\,G_z(t)\,dt,
\end{align}
where the third equation is from conditioning on $T_i(z)$ and the outer expectation is taken over $T_i(z)$ given $Z_i$ and $X_i$, the fourth from $\{T_i(z),X_i\} \perp C_i$ given $Z_i=z$ by A4, and the last from $T_i(z) \perp Z_i$ given $X_i$ by A2, so the outer expectation is over $T_i(z)$ given $X_i$.

Averaging over the target population gives the \emph{standardized event rate} of arm $z$,
\begin{equation}\label{eqapp:standardized-event-rate}
    d_z^w := \bE^w[d_z(X)] =\int_0^{t^\dagger}\!\bE^w\!\left[f_{T(z)\mid X}(t)\right]G_z(t)\,dt
    =\int_0^{t^\dagger}\!f_z^w(t)G_z(t)\,dt.
\end{equation}
Equivalently, $d_z^w=\int\Pr(\delta_i=1\mid Z_i=z,X_i=x)\,dP^w(x)$ by the definition of $d_z(x)$, that is, the event rate that arm $z$ would show if its covariates followed the target population.

With $h \equiv 1$, the equation above reduces to the standardized event rate in the study population, $d_z=\bE[d_z(X)]=\int_0^{t^\dagger}\!f_z(t)G_z(t)\,dt$. We denote $d=rd_1+(1-r)d_0$ and $d^w=rd_1^w+(1-r)d_0^w$ for the combined event rate. 
It is worth separating the standardized event rate $d_z$ from the observed event rate $\Pr(\delta_i=1|Z=z)=\bE[d_z(X) \mid Z_i=z]$; the two coincide in a randomized trial, but generally differ in an observational study where the marginal distribution of $X$ differs from that within arm $z$. We denote the observed event rates by $\bar{d}_z:=\Pr(\delta_i=1\mid Z_i=z)$ and $\bar d:=\Pr(\delta_i=1)=r\bar{d}_1+(1-r)\bar{d}_0$.

\subsubsection{The proportional hazards alternative in the target population}\label{secapp:marcox-model}

Theorem \ref{thm:asymptotic-variance} establishes the asymptotic distribution of the estimator $\widehat{\tau}_n$ without any model assumption in the target population: $\tau_0$ is defined as the root of $u(\tau)\!=\!\bE[w_i\delta_i\{Z_i-\pi(\tau,T_i)\}] \!=\! 0$, and the cumulative hazard thereon $\Lambda_w:=\int d\mu/s_0(\tau_0,\cdot)$. The proportional hazards alternative (H1) in the target population with weight $w$ is $\lambda_1^w(t)=\lambda_0^w(t)\exp(\tau_0^w)$, with the baseline cumulative hazard $\Lambda_0^w$ in Section \ref{secapp:weight-and-population}.
Under H1, the pairs agree, and we will use them interchangeably whenever H1 is assumed,
\begin{equation*}
    \tau_0=\tau_0^w \qquad\text{and}\qquad \Lambda_w=\Lambda_0^w.
\end{equation*}
\begin{proof}[of $\tau_0=\tau_0^w$ and $\Lambda_w=\Lambda_0^w$]
For $\tau_0=\tau_0^w$, it suffices to show $u(\tau_0^w)=0$ since $u(\tau)=0$ has a unique root; then, we show $\Lambda_w(t)=\Lambda_0^w(t)$. We proceed with two parts (i) and (ii).

\underline{\textit{Part (i).}}
We first establish two useful identities, given in equations \eqref{eqapp:sk-simplified} and \eqref{eqapp:weighted-event-moments}, under identification assumptions A1-A4 and the balancing weights \eqref{eqapp:balancing-weight} only; neither of them requires the alternative (H1). We will use them repeatedly throughout the appendix.

The first identity: the weighted risk-set average $s_k(\tau,t)=\bE[w_i Y_i(t) \exp(\tau Z_i) Z_i^k]$ can be written as
\begin{align} \label{eqapp:sk-simplified}
    s_k(\tau,t) & = \bE\left[ \bE\left\{w_iY_i(t) \exp(\tau Z_i) Z_i^k \,\big| \,Z_i,X_i \right\} \right] \nonumber \\
    & = \bE_X\!\left[ \sum_{z\in\{0,1\}}e_z(X) w_z(X)  \exp(\tau z) z^k 
    \,\bE\big\{Y_i(t)\,\big|\,Z_i=z,X \big\} \right] \nonumber \\
    & = \bE_X\!\left[\sum_z e_z(X) w_z(X) \exp(\tau z) z^k \, F_z(t|X) G_z(t)\right] \nonumber \\
    & = \sum_z \frac{r_z}{m} \exp(\tau z) z^k G_z(t) \, \bE_X\!\left[h(X)F_z(t|X)\right] \nonumber \\
    & = \sum_z r_z \exp(\tau z) z^k\,F_z^w(t)\,G_z(t),
\end{align}
where $X$ is $X_i$ for brevity.
The third equation is from $\bE[Y_i(t) \mid Z_i\!=\!z, X_i\!=\!x]=F_z(t|x)G_z(t)$; to see this, notice $\bE[Y_i(t)]=\Pr(T_i \geq t)$ since $Y_i(t)=\bI(T_i \geq t)$, and decompose $Y_i(t)= \bI(T_i^*\ge t)\bI(C_i\ge t)$ for all $t \leq t^\dagger$, and therefore $\Pr(T_i^*\ge t, C_i\ge t \mid Z_i=z,\! X_i=x) = \Pr(T(z) \geq t|X_i=x) \Pr(C_i \geq t|Z_i=z) = F_z(t|x)G_z(t)$ by A1, A2, and A4. The fourth substitutes $e_z(x)w_z(x)=r_zh(x)/m$ by equation \eqref{eqapp:balancing-weight}, and the fifth is from $\bE_X\!\left[h(X)F_z(t|X)\right]/m=\bE^w\!\left[F_z(t|X)\right]=F_z^w(t)$. Immediately,
\begin{equation*}
    s_1(\tau,t) = r_1 \exp(\tau) F_1^w(t)G_1(t), \qquad s_0(\tau,t) = r_1 \exp(\tau) F_1^w(t)G_1(t) + r_0 F_0^w(t)G_0(t).
\end{equation*}
The second identity: the moment $\bE[w_i\delta_iZ_i^kf(T_i)]$ for any bounded measurable $f$ on $[0,t^\dagger]$ can be written as
\begin{align}\label{eqapp:weighted-event-moments}
    \bE[w_i\delta_iZ_i^kf(T_i)]
    &=\bE\left[\bE\left\{w_i\delta_iZ_i^kf(T_i)\,\big|\,Z_i,X_i\right\}\right] \nonumber \\
    &=\bE_X\!\left[\sum_z e_z(X) w_z(X) z^k \, \bE\!\left\{ \delta_if(T_i) \,\big|\,Z_i=z,X\right\} \right] \nonumber \\
    &=\bE_X\!\left[\sum_z e_z(X) w_z(X) z^k \int_0^{t^\dagger}\!f(t)f_{T(z)|X}(t)G_z(t)\,dt\right] \nonumber \\
    &=\sum_z\frac{r_z}{m} z^k \int_0^{t^\dagger}\! f(t) 
    \bE_X\!\left[h(X)f_{T(z)|X}(t)\right] G_z(t) \,dt \nonumber \\
    &=\sum_z r_z z^k \int_0^{t^\dagger}\!f(t)f_z^w(t)G_z(t)\,dt.
\end{align}
The fourth equation is again from $e_z(x)w_z(x)=r_z h(x)/m$ and Fubini, and the fifth from $\bE_X\!\left[h(X)f_{T(z)|X}(t)\right]/m = \bE^w[f_{T(z)|X}(t)]=f_z^w(t)$.
The third is from $\delta_i=\bI(T_i^*\le C_i\wedge t^\dagger)$ being binary and $T_i=T_i^*$ when $\delta_i=1$, so conditioning on $T_i(z)$ as in equation \eqref{eqapp:conditional-event-rate} gives
\begin{align}\label{eqapp:conditional-event-moments}
    & \bE\!\left[\delta_if(T_i) \mid Z_i \!=\! z, X_i \!=\! x \right] \nonumber \\
    & = \bE\!\left[ \bI\{T_i(z) \leq t^\dagger\} f(T_i(z)) \Pr\{T_i(z) \leq C_i \mid T_i(z), Z_i \!=\! z, X_i \!=\! x \} \mid Z_i=z, X_i=x \right] \nonumber \\
    & = \bE\!\left[ \bI\{T_i(z) \leq t^\dagger\} f(T_i(z)) G_z(T_i(z)) \mid X_i \!=\! x \right]
    = \int_0^{t^\dagger}\!f(t)f_{T(z)\mid X=x}(t)G_z(t)\,dt.
\end{align}
Equation \eqref{eqapp:weighted-event-moments} has the following immediate results on the total masses $\mu(t^\dagger)$ and $\mu_z(t^\dagger)$:
\begin{equation}\label{eqapp:mu-total-mass}
    \mu(t^\dagger) = \bE[w_i\delta_i] = \sum_{z\in \{0,1\}} r_z \int_0^{t^\dagger}\! f_z^w(t)G_z(t)\,dt = rd_1^w+(1-r)d_0^w = d^w,
\end{equation}
and with $\mu=\mu_1+\mu_0$,
\begin{equation}\label{eqapp:muz-total-mass}
    \mu_1(t^\dagger) = \bE[w_i \delta_i Z_i] = \sum_{z\in \{0,1\}} r_z \,z \int_0^{t^\dagger}\! f_z^w(t)G_z(t)\,dt = rd_1^w, \quad
    \mu_0(t^\dagger) = (1-r)d_0^w.
\end{equation}

\underline{\textit{Part (ii).}} The rest of the proof assumes the proportional hazards alternative (H1).
We now show $u(\tau_0^w) = \bE[w_i\delta_i Z_i] - \bE[w_i\delta_i \pi(\tau_0^w,T_i)] = 0$ under (H1). For $\bE[w_i\delta_i Z_i]$, taking $k=1$ and $f \equiv 1$ in $\bE[w_i\delta_iZ_i^kf(T_i)]$ yields
\begin{equation*}
    \bE[w_i\delta_i Z_i] = 0 + r_1 \int_0^{t^\dagger}\!f_1^w(t)G_1(t)\,dt = r_1 d_1^w.
\end{equation*}
For $\bE[w_i\delta_i \pi(\tau_0^w,T_i)]$, notice that $\pi(\tau_0^w,t)$ is a bounded measurable function on $[0,t^\dagger]$ and hence Lemma \ref{lemapp:1} applies, 
\begin{equation*}
    \bE[w_i\delta_i \pi(\tau_0^w,T_i)] = \int_0^{t^\dagger}\!\pi(\tau_0^w,t)\,d\mu(t).
\end{equation*}
Taking $k\!=\!0$ and $f\!=\!\bI_B$ in $\bE[w_i\delta_iZ_i^kf(T_i)]$ yields $\mu(B) = \sum_z r_z \int_B f_z^w(t) G_z(t)\,dt$, and setting $B=[0,t]$ gives $\mu(t)=\int_0^t\sum_zr_zf_z^w(s)G_z(s)\,ds$, so $\mu$ has a density $d\mu(t)/dt=\sum_z r_z f_z^w(t) G_z(t)$. Furthermore, under H1, $\lambda_1^w(t)=\lambda_0^w(t)\exp(\tau_0^w)$, and therefore
\begin{equation}\label{eqapp:marcox-density}
    f_z^w(t) = \exp(\tau_0^w z) F_z^w(t) \lambda_0^w(t), \qquad z=0,1,
\end{equation}
which follows from $\lambda_z^w=-(\log F_z^w)'=f_z^w/F_z^w$. Substituting into the density $d\mu(t)/dt$ yields
\begin{equation}\label{eqapp:weighted-event-density-marcox}
    d\mu(t)=\lambda_0^w(t)\sum_z r_z \exp(\tau_0^w z) F_z^w(t)G_z(t)\,dt = s_0(\tau_0^w,t) \lambda_0^w(t) \,dt,
\end{equation}
because $s_0(\tau_0^w,t)= \sum_z r_z \exp(\tau_0^w z)\,F_z^w(t)\,G_z(t)$ by equation \eqref{eqapp:sk-simplified}. Therefore,
\begin{equation*}
    \int_0^{t^\dagger}\!\pi(\tau_0^w,t)\,d\mu(t) = \int_0^{t^\dagger}\!\frac{s_1(\tau_0^w,t)}{s_0(\tau_0^w,t)}\,s_0(\tau_0^w,t) \lambda_0^w(t)\,dt = \int_0^{t^\dagger}\! s_1(\tau_0^w,t) \lambda_0^w(t)\,dt.
\end{equation*}
Further plugging in $s_1(\tau_0^w,t) = r_1 \exp(\tau_0^w) F_1^w(t)G_1(t)$,
\begin{equation*}
    \int_0^{t^\dagger}\! s_1(\tau_0^w,t) \lambda_0^w(t)\,dt
    = \int_0^{t^\dagger}\! r_1 \underbrace{\exp(\tau_0^w) F_1^w(t) \lambda_0^w(t)}_{=f_1^w(t)} G_1(t)\,dt
    = \int_0^{t^\dagger}\! r_1 f_1^w(t) G_1(t) \, dt = r_1 d_1^w, 
\end{equation*}
which yields $\bE[w_i\delta_i \pi(\tau_0^w,T_i)] = r_1 d_1^w = \bE[w_i\delta_i Z_i]$. Therefore, $u(\tau_0^w)=0$ and hence $\tau_0^w=\tau_0$.

Finally, for $t\in[0,t^\dagger]$,
\begin{equation*}
    \Lambda_w(t)=\int_0^t \frac{d\mu(s)}{s_0(\tau_0,s)} = \int_0^t \frac{1}{s_0(\tau_0,s)} s_0(\tau_0,s) \lambda_0^w(s) \,ds
    =\int_0^t\lambda_0^w(s)\,ds=\Lambda_0^w(t),
\end{equation*}
establishing $\Lambda_w = \Lambda_0^w$. 
\end{proof}

\subsubsection{Proof of Theorem 2}\label{secapp:theorem2-proof}

We prove Theorem \ref{thm:decomposition} under any balancing weights $w_i$ from equation \eqref{eqapp:balancing-weight}. To this end, we prove
\begin{equation*}
    V_w =\frac{\kappa_1(1+\Delta_1)}{rd_1^w} + \frac{\kappa_0(1+\Delta_0)}{(1-r)d_0^w} + \epsilon,
\end{equation*}
where $\kappa_z = 1 + \bV[w_i| Z_i\!=\!z]$ and $\epsilon=\{rC_1 + (1\!-r)C_0\}/A(\tau_0)^2$, with $C_z=\!\Cov^w\!\!\left\{w_z(e_i),Q_z(X_i)\right\}$ taken over the target population and $Q_z(x):=\bE[\zeta_i^2\mid Z_i=z,X_i=x]$ for $\zeta_i$ from Theorem~\ref{thm:asymptotic-variance}. The terms
\begin{equation*}
    \Delta_1=\frac{\bV\{\pi(\tau_0,\widetilde{T}_1)\}}{[1-\bE\{\pi(\tau_0,\widetilde{T}_1)\}]^2},\qquad
    \Delta_0=\frac{\bV\{\pi(\tau_0,\widetilde{T}_0)\}}{[\bE\{\pi(\tau_0,\widetilde{T}_0)\}]^2}
\end{equation*}
are the squared coefficients of variation of $1\!-\!\pi(\tau_0,\widetilde{T}_1)$ and $\pi(\tau_0,\widetilde{T}_0)$, respectively, where $\widetilde{T}_z$ follows the weighted distribution of observed event times in arm $z$, $\Pr(\widetilde{T}_z\leq t)=\mu_z(t)/\mu_z(t^\dagger)$.

\begin{proof}[of Theorem \ref{thm:decomposition}]

From Section \ref{secapp:marcox-model}, $\tau_0=\tau_0^w$ and $\Lambda_w=\Lambda_0^w$ under H1; we adopt $\tau_0$ and $\Lambda_0^w$ for notations. Write $\pi=\pi(\tau_0,t)$ in the integrands, and recall that $\bE[\pi(\tau_0,\widetilde{T}_z)]$ and $\bV[\pi(\tau_0,\widetilde{T}_z)]$ are denoted as $\bE_{\mu_z}[\pi]$ and $\bV_{\mu_z}[\pi]$ by the definition in \eqref{eqapp:weighted-event-time-moments}. Notice that $w_z(e_i)=w_z(X_i)$ since $w_z$ is a function of $e_i=e(X_i)$ alone. Then, Lemmas \ref{lemapp:4} and \ref{lemapp:5} gives,
\begin{align*}
    A(\tau_0) & = \int_0^{t^\dagger}(1-\pi)\,d\mu_1 
    = \mu_1(t^\dagger)\,\{1-\bE_{\mu_1}[\pi] \} \\
    & = \int_0^{t^\dagger}\pi\,d\mu_0 
    = \mu_0(t^\dagger)\, \bE_{\mu_0}[\pi], \\
    B(\tau_0) & = \int_0^{t^\dagger}\pi(1-\pi)\,\kappa_w(\pi)\,d\mu+ B_\epsilon(\tau_0),
\end{align*}
where $\kappa_w(\pi) = (1-\pi)\kappa_1 + \pi\kappa_0$ with $\kappa_z:=\bE[w_i^2|Z=z]$, and $B_\epsilon(\tau_0) = rC_1 + (1-r)C_0$. From Section \ref{secapp:design-effect}, $\kappa_z$ is termed as design effect and $\kappa_z=\bE[w_i^2|Z=z]=1+\bV[w_i|Z_i=z]$. Further expanding $\kappa_w(\pi)$ in the leading term of $B(\tau_0)$,
\begin{align*}
    \int_0^{t^\dagger}\! \pi(1-\pi)\,\kappa_w(\pi)\,d\mu
    & = \kappa_1 \int_0^{t^\dagger}\! \pi(1-\pi)^2\,d\mu + \kappa_0\int_0^{t^\dagger}\! \pi^2(1-\pi)\,d\mu \\
    & = \kappa_1 \int_0^{t^\dagger}(1-\pi)^2\,d\mu_1 + \kappa_0\int_0^{t^\dagger}\pi^2\,d\mu_0 \\
    & = \kappa_1\,\mu_1(t^\dagger)\,\bE_{\mu_1}\!\left[(1-\pi)^2\right] + \kappa_0\,\mu_0(t^\dagger)\,\bE_{\mu_0}\!\left[\pi^2\right] \\
    & = \kappa_1 \,\mu_1(t^\dagger) \left\{(1-\bE_{\mu_1}[\pi])^2 + \bV_{\mu_1}[\pi] \right\} + \kappa_0 \,\mu_0(t^\dagger) \left\{\bE_{\mu_0}[\pi]^2 + \bV_{\mu_0}[\pi]\right\} ,
\end{align*}
where the last equation applies $\bE[\cdot^2]=\bV[\cdot] + \bE[\cdot]^2$. Plugging into $V_w \!=\! A(\tau_0)^{-2}B(\tau_0)$ from Theorem \ref{thm:asymptotic-variance}, with the first term divided by $A(\tau_0)^2=\mu_1(t^\dagger)^2 (1-\bE_{\mu_1}[\pi])^2$ and the second term by $A(\tau_0)^2 = \mu_0(t^\dagger)^2 \bE_{\mu_0}[\pi]^2$,
\begin{align*}
    V_w-\epsilon & = \frac{\kappa_1\,\mu_1(t^\dagger) \left\{(1-\bE_{\mu_1}[\pi])^2 + \bV_{\mu_1}[\pi] \right\} }{\mu_1(t^\dagger)^2 (1-\bE_{\mu_1}[\pi])^2} + \frac{\kappa_0\,\mu_0(t^\dagger) \left\{\bE_{\mu_0}[\pi]^2 + \bV_{\mu_0}[\pi] \right\} }{\mu_0(t^\dagger)^2\, \bE_{\mu_0}[\pi]^2} \\
    & = \frac{\kappa_1(1+\Delta_1)}{\mu_1(t^\dagger)} + \frac{\kappa_0(1+\Delta_0)}{\mu_0(t^\dagger)}
\end{align*}
where $\mu_1(t^\dagger) = rd_1^w$ and $\mu_0(t^\dagger)= (1-r)d_0^w $ by equation \eqref{eqapp:muz-total-mass}, and $\Delta_1$ and $\Delta_0$ are as above. The residual follows from $\epsilon = A(\tau_0)^{-2} B_\epsilon(\tau_0)$ with $A(\tau_0)$ from Lemma \ref{lemapp:4} (i).
\end{proof}

\subsubsection{Lemmas}\label{secapp:theorem2-lemmas}

\begin{lemma}[Reduction of $A(\tau_0)$]\label{lemapp:4} 

\leavevmode
\begin{enumerate}
    \item [(i)] Under the conditions of Theorem \ref{thm:asymptotic-variance},
    \begin{equation*}
        A(\tau_0)
        = \int_0^{t^\dagger} \pi(\tau_0,t) \left\{1-\pi(\tau_0,t) \right\}d\mu(t).
    \end{equation*}

    \item [(ii)] Under the conditions of Theorem \ref{thm:decomposition},
    \begin{align*}
    A(\tau_0) & = \int_0^{t^\dagger} \{1-\pi(\tau_0,t)\} \,d\mu_1(t)  
    = \mu_1(t^\dagger)\,\{1-\bE_{\mu_1}[\pi] \} \\
    & = \int_0^{t^\dagger} \pi(\tau_0,t) \,d\mu_0(t)
    = \mu_0(t^\dagger)\, \bE_{\mu_0}[\pi].
    \end{align*}
\end{enumerate}
\end{lemma}

\begin{proof}[of Lemma \ref{lemapp:4}]
\leavevmode

\textit{(i)} Recall from Theorem \ref{thm:asymptotic-variance} that
\begin{equation*}
    A(\tau_0) = \bE\!\left[ w_i\delta_i \left\{ 
    \frac{s_2(\tau_0,T_i)}{s_0(\tau_0,T_i)} - \pi^2(\tau_0,T_i) \right\} \right].
\end{equation*}
Since $Z_i$ is binary, $Z_i^2=Z_i$ exactly, thus $s_2(\tau_0,t)=s_1(\tau_0,t)$, $s_2/s_0 - \pi^2=\pi(1-\pi)$, and $A(\tau_0)=\bE\left[w_i\delta_i\,\pi(\tau_0,T_i)\{1-\pi(\tau_0,T_i)\}\right]$. The function $f(t)=\pi(\tau_0,t)\{1-\pi(\tau_0,t)\}$ is bounded, measurable, and deterministic in $t$, so Lemma \ref{lemapp:1} applies,
\begin{equation*}
    A(\tau_0) = \int_0^{t^\dagger}\pi(\tau_0,t)\left\{1-\pi(\tau_0,t)\right\}d\mu(t),
\end{equation*}
which shows the first claim. 

\textit{(ii)} Recall that $\mu_z$ is the arm-specific weighted event measure defined in equation \eqref{eqapp:weighted-event-measure-arm}, and $\mu=\mu_1+\mu_0$.
Taking $k\!=\!1$ and $f\!=\!\bI_B$ in equation \eqref{eqapp:weighted-event-moments} gives 
\begin{equation*}
    \mu_1(B) = \bE[w_i\delta_iZ_i\bI(T_i\in B)] = r\int_Bf_1^w(t)G_1(t)\,dt,
\end{equation*}
so $\mu_1$ has density $rf_1^wG_1$. Under (H1), $\tau_0=\tau_0^w$ by Section \ref{secapp:marcox-model}.  Then, using equations \eqref{eqapp:marcox-density} and \eqref{eqapp:sk-simplified},
\begin{equation*}
    r f_1^w(t)G_1(t) = r\exp(\tau_0)F_1^w(t)G_1(t)\,\lambda_0^w(t) 
    = s_1(\tau_0,t)\,\lambda_0^w(t) 
    = \pi(\tau_0,t) \,s_0(\tau_0,t)\lambda_0^w(t),
\end{equation*}
and $s_0(\tau_0,t)\lambda_0^w(t)\,dt=d\mu(t)$ from equation \eqref{eqapp:weighted-event-density-marcox}. Hence,
\begin{equation}\label{eqapp:arm-event-measure-density}
    d\mu_1(t)=\pi(\tau_0,t)\,d\mu(t), \qquad d\mu_0(t)=\left\{1-\pi(\tau_0,t)\right\}d\mu(t).
\end{equation}
Therefore, $A(\tau_0)$ has two additional forms,
\begin{align*}
    A(\tau_0) & = \int_0^{t^\dagger} \left\{1-\pi(\tau_0,t)\right\} \pi(\tau_0,t) \,d\mu(t) \\ 
    & = \int_0^{t^\dagger} \left\{1-\pi(\tau_0,t)\right\} \,d\mu_1(t) 
    = \mu_1(t^\dagger) - \mu_1(t^\dagger) \bE_{\mu_1}[\pi], \\
    A(\tau_0) & = \int_0^{t^\dagger}\pi(\tau_0,t)\left\{1-\pi(\tau_0,t)\right\}d\mu(t) \\
    & = \int_0^{t^\dagger} \pi(\tau_0,t) \,d\mu_0(t)
    = \mu_0(t^\dagger)\, \bE_{\mu_0}[\pi],
\end{align*}
which shows the second claim. 

\textit{Remark}. The total masses $\mu_z(t^\dagger)$ are positive without requiring (H1), which follows from the proof. Specifically, by the score equation $u(\tau_0)=0$ and Lemma \ref{lemapp:1},
\begin{align*}
    \mu_1(t^\dagger) & = \bE[w_i\delta_iZ_i] = \int_0^{t^\dagger} \pi d\mu \geq A(\tau_0) > 0, \\
    \mu_0(t^\dagger) & = \bE[w_i\delta_i(1-Z_i)] = \int_0^{t^\dagger} (1-\pi) d\mu \geq A(\tau_0) > 0.
\end{align*}
Moreover, $A(\tau_0)>0$, hence $\bE_{\mu_1}[\pi]<1$ and $\bE_{\mu_0}[\pi]>0$ by part (ii).
\end{proof}

\begin{lemma}[Reduction of $B(\tau_0)$]\label{lemapp:5}
Under the conditions of Theorem \ref{thm:decomposition},
\begin{align*}
    B(\tau_0) & = \int_0^{t^\dagger}\pi(\tau_0,t)\left\{1-\pi(\tau_0,t)\right\}\kappa_w\!\left(\pi(\tau_0,t)\right)d\mu(t)+B_\epsilon(\tau_0),
\end{align*}
where $B_\epsilon(\tau_0) = r C_1 + (1-r) C_0$ with $C_z$ as in Theorem~\ref{thm:decomposition}, and $\kappa_w(\pi) := (1-\pi)\kappa_1 + \pi\kappa_0$.
\end{lemma}

\begin{proof}[of Lemma \ref{lemapp:5}]
Recall from Theorem \ref{thm:asymptotic-variance} that $B(\tau_0)=\bE[\eta_i^2(\tau_0)]$ and 
\begin{equation*}
    \eta_i(\tau_0) = w_i\zeta_i, \quad \zeta_i = \delta_i[Z_i-\pi(\tau_0,T_i)] - \exp(\tau_0Z_i) \int_0^{T_i} [Z_i-\pi(\tau_0,t)]\,d\Lambda_w(t).
\end{equation*}
We have shown $\Lambda_w=\Lambda_0^w$ under H1 in Section \ref{secapp:marcox-model}. Define $\Omega_z(t):=\int_0^t[z-\pi(\tau_0,s)]\,d\Lambda_0^w(s)$ for $z=0,1$, then $\zeta_i$ can be written as
\begin{equation*}
    \zeta_i =  \,\delta_i\!\left\{Z_i-\pi(\tau_0,T_i)\right\} - \exp(\tau_0Z_i)\,\Omega_{Z_i}(T_i).
\end{equation*}
Therefore, $Q_z(x) = \bE\left[ \, [\delta_i\{z-\pi(\tau_0,T_i)\} - \exp(\tau_0z)\,\Omega_z(T_i)]^2\mid Z_i=z,X_i=x \right]$ for $Q_z(x):=\bE\left[ \zeta_i^2 \mid Z_i=z,X_i=x \right]$ from Theorem \ref{thm:decomposition}. Notice that $w_i$ is determined after being conditioned on $(Z_i,X_i)$, and hence by law of total expectation
\begin{equation*}
    B(\tau_0)=\bE[\eta_i^2(\tau_0)] = \bE\!\left[ w_i^2 \, Q_{Z_i}(X_i) \right],
\end{equation*}
where the outer expectation is over $(Z_i,X_i)$. Summing over $Z_i\!=\!0,1$ with probability $e_z(X_i)$,
\begin{equation}\label{eqapp:Bmatrix-two-terms}
    B(\tau_0) = \bE\!\left[ e_1(X_i) w_1^2(X_i) Q_1(X_i) \right] + \bE\!\left[ e_0(X_i) w_0^2(X_i) Q_0(X_i) \right]
\end{equation}
where the two outer expectations are taken over $X_i$. It remains to evaluate $Q_z(x)$ and then the two outer expectations.

\underline{\textit{Evaluating $Q_z(x)$.}} Expanding the square,
\begin{align*}
    Q_z(x) & = \bE\left[\delta_i^2 \{z-\pi(\tau_0,T_i)\}^2 \mid \cdot \right] \,
    - \, 2 \exp(\tau_0z) \,\bE\left[\delta_i \{z-\pi(\tau_0,T_i)\}\,\Omega_z(T_i) \mid \cdot \right] \\
    & \quad + \, \exp(2\tau_0z) \,\bE\left[ \Omega_z(T_i)^2 \mid \cdot \right],
\end{align*}
where all three expectations are conditioned on $Z_i=z$ and $X_i=x$, suppressed by dots. For the first, $\delta_i^2=\delta_i$ by binary, and plugging in $f(t)=\{z-\pi(\tau_0,t)\}^2$ into equation \eqref{eqapp:conditional-event-moments} yields
\begin{equation*}
    \bE\left[\delta_i^2\{z-\pi(\tau_0,T_i)\}^2\mid\cdot\right] = \int_0^{t^\dagger}\!\{z-\pi(\tau_0,t)\}^2 \, f_{T(z)\mid X=x}(t)G_z(t)\,dt =:\Upsilon_z(x).
\end{equation*}
For the second, taking $f(t) = \{z-\pi(\tau_0,t)\}\,\Omega_z(t)$ in equation \eqref{eqapp:conditional-event-moments} yields
\begin{equation*}
    \bE\left[\delta_i \{z-\pi(\tau_0,T_i)\}\,\Omega_z(T_i) \mid \cdot \right] = \int_0^{t^\dagger}\! \{z-\pi(\tau_0,t)\}\,\Omega_z(t) \, f_{T(z)\mid X=x}(t)G_z(t)\,dt =:\Psi_z(x).
\end{equation*}
For the last expectation, notice that $\Omega_z(t)$ is absolutely continuous with $d\Omega_z(t) = \{z-\pi(\tau_0,t)\} \,d\Lambda_0^w(t)$ and $\Omega_z(0)=0$, so
\begin{equation*}
    \Omega_z(T_i)^2 = \Omega_z(T_i)^2 - \Omega_z(0)^2 = \int_0^{T_i} 2\, \Omega_z(t) \,d\Omega_z(t) 
    = 2 \int_0^{t^\dagger} \Omega_z(t)\{z-\pi(\tau_0,t)\}\, Y_i(t) \,d\Lambda_0^{w}(t).
\end{equation*}
Applying Fubini and $\bE[Y_i(t) \mid Z_i=z,\! X_i=x] = F_z(t|x)G_z(t)$ as in equation \eqref{eqapp:sk-simplified} yields
\begin{align*}
    \bE\!\left[ \Omega_z(T_i)^2 \mid \cdot \right] & = 2 \! \int_0^{t^\dagger} \! \Omega_z(t)\{z-\pi(\tau_0,t)\} \, \bE[Y_i(t) \mid \cdot ] \,d\Lambda_0^{w}(t) \\
    & = 2 \! \int_0^{t^\dagger} \! \Omega_z(t)\{z-\pi(\tau_0,t)\} \, F_z(t|x)G_z(t) \,d\Lambda_0^{w}(t) =: 2\Phi_z(x).
\end{align*}
Collecting the three terms,
\begin{equation*}
    Q_z(x) = \Upsilon_z(x) - 2\exp(\tau_0z) \Psi_z(x) + 2\exp(2\tau_0z) \Phi_z(x).
\end{equation*}

\underline{\textit{Outer expectations.}} Notice that $e_z w_z = r_zh/m$ by definition of $w_z$. Write the two outer expectations of $B(\tau_0)$ in equation \eqref{eqapp:Bmatrix-two-terms} under $P^w$ using $\bE[h\,\cdot\,]/m = \bE^w[\cdot]$, and then split the expectation of the product, for $z=0,1$,
\begin{align}\label{eqapp:Bmatrix-twoterms-exact}
    \bE\!\left[ e_z w_z^2 Q_z \right] & = \frac{r_z}{m} \bE\!\left[ h \cdot w_z Q_z \right] = r_z \,\bE^w\!\left[ w_z Q_z \right] \nonumber \\
    & = r_z \left\{ \bE^w[w_z]\,\bE^w[Q_z] + C_z \right\} = r_z \kappa_z \bE^w[Q_z] + r_z C_z,
\end{align}
where $C_z:=\Cov^w\{w_z(X),\,Q_z(X)\}$ as in Theorem \ref{thm:decomposition}. The last equation is from $\bE^w[w_z] = \bE[hw_z]/m = \bE[e_z w_z^2]/r_z$, and that by the definition of $\bE\!\left[ g(X_i) \mid Z_i=z \right]$ for any $g$,
\begin{align*}
    \bE\!\left[ e_z(X_i) g(X_i) \right]/r_z & = 
    \bE\!\left[ \bE\{\bI(Z_i=z) g(X_i) \mid X_i\} \right]/r_z \\
    & = \bE\!\left[ \bI(Z_i=z) g(X_i) \right]/\Pr(Z_i=z)
    = \bE\!\left[ g(X_i) \mid Z_i=z \right],
\end{align*}
so that taking $g=w_z^2$ yields $\bE^w[w_z]=\bE[e_z w_z^2]/r_z = \bE[w_z^2 \mid Z_i=z] = \kappa_z$.

For $\bE^w[Q_z]$, averaging over $P^w$ with $\bE^w[f_{T(z)\mid X}(t)]=f_z^w(t)$ and $\bE^w[F_z(t\mid X)]=F_z^w(t)$ from Section \ref{secapp:weight-and-population}, and applying Fubini gives
\begin{align*}
    \bE^w[Q_z] & = \bE^w[\Upsilon_z(X)] - 2\exp(\tau_0z) \bE^w[\Psi_z(X)] + 2\exp(2\tau_0z) \bE^w[\Phi_z(X)] \\
    & = \int_0^{t^\dagger}\! \{z-\pi\}^2 G_z f_z^w\,dt
    - 2\exp(\tau_0z) \int_0^{t^\dagger}\! \{z-\pi\}\,\Omega_z G_z f_z^w\,dt \\
    & \quad + 2\exp(2\tau_0z) \int_0^{t^\dagger}\! \{z-\pi\}\,\Omega_z F_z^w G_z \,d\Lambda_0^w,
\end{align*}
where all factors of integrands are functions of $t$ and $\pi=\pi(\tau_0,t)$. Under H1, $f_z^w = \exp(\tau_0 z) F_z^w \lambda_0^w$ by equation \eqref{eqapp:marcox-density}, and thus
\begin{align*}
    \int_0^{t^\dagger}\! \{z-\pi\}^2 G_z f_z^w\,dt & = \exp(\tau_0 z) \int_0^{t^\dagger}\! \{z-\pi\}^2 F_z^w G_z \,d\Lambda_0^w, \\
    \int_0^{t^\dagger}\! \{z-\pi\}\,\Omega_z G_z f_z^w\,dt & = \exp(\tau_0 z) \int_0^{t^\dagger}\! \{z-\pi\}\,\Omega_z F_z^w  G_z \,d\Lambda_0^w.
\end{align*}
The second equation implies that the last two terms of $\bE^w[Q_z]$ cancel, and hence 
\begin{equation*}
    \bE^w[Q_z] = \exp(\tau_0 z) \int_0^{t^\dagger} \{z-\pi\}^2 F_z^w G_z \,d\Lambda_0^w.
\end{equation*}
To write it with $\mu$, notice that $r\exp(\tau_0)F_1^wG_1 = s_1(\tau_0,t) = \pi\,s_0(\tau_0,t)$ and $(1-r)F_0^wG_0 = s_0(\tau_0,t)-s_1(\tau_0,t) = (1-\pi)\,s_0(\tau_0,t)$ by equation \eqref{eqapp:sk-simplified}, and $s_0(\tau_0,t)\,d\Lambda_0^w(t) = d\mu(t)$ under the proportional hazards alternative (H1) from Section \ref{secapp:marcox-model}. Therefore,
\begin{align*}
    r \,\bE^w[Q_1] & = \int_0^{t^\dagger} (1-\pi)^2\, r\exp(\tau_0)F_1^wG_1\,d\Lambda_0^w = \int_0^{t^\dagger} (1-\pi)^2\, \pi\, s_0\,d\Lambda_0^w \\
    & = \int_0^{t^\dagger} \pi (1-\pi)^2\,d\mu, \\
    (1-r) \bE^w[Q_0] & = \int_0^{t^\dagger} \pi^2\, (1-r)F_0^wG_0\,d\Lambda_0^w = \int_0^{t^\dagger} \pi^2 (1-\pi)\, s_0\,d\Lambda_0^w \\
    & = \int_0^{t^\dagger} \pi^2 (1-\pi) \,d\mu.
\end{align*}
Plug the expressions for $\bE^w[Q_z]$ above back to equations \eqref{eqapp:Bmatrix-twoterms-exact} and \eqref{eqapp:Bmatrix-two-terms}, 
\begin{align*}
    B(\tau_0) & = \kappa_1 r \,\bE^w[Q_1] + \kappa_0 (1-r) \,\bE^w[Q_0] + r C_1 + (1-r) C_0 \\
    & = \kappa_1 \int_0^{t^\dagger}\! \pi(1-\pi)^2 \,d\mu + \kappa_0\int_0^{t^\dagger}\! \pi^2(1-\pi)\,d\mu + B_\epsilon(\tau_0) \\
    & = \int_0^{t^\dagger}\! \pi (1-\pi) \{(1-\pi)\kappa_1 + \pi\,\kappa_0\} \,d\mu + B_\epsilon(\tau_0) \\
    & = \int_0^{t^\dagger}\! \pi (1-\pi) \kappa_w(\pi) \,d\mu + B_\epsilon(\tau_0),
\end{align*}
where $\kappa_w(\pi) = (1-\pi)\kappa_1 + \pi\kappa_0$ and $B_\epsilon(\tau_0) = r C_1 + (1-r) C_0$, which finishes the proof.
\end{proof}

\subsubsection{Design effect}\label{secapp:design-effect}

The design effect $\kappa_z$ measures the inflation of arm $z$'s contribution to $V_w$ due to weighting, relative to that in the target randomized trial enrolling from the target population with the same treatment proportion $r$; in that trial, $w_i\equiv1$ and thus $\kappa_1=\kappa_0=1$ and $\epsilon=0$, so
\begin{equation*}
    V_w =\frac{1+\Delta_1}{rd_1^w} + \frac{1+\Delta_0}{(1-r)d_0^w}.
\end{equation*}
Moreover, besides $\kappa_z:=\bE[w_i^2|Z_i=z]$ and $\bE^w[w_z]=\kappa_z$ from equation \eqref{eqapp:Bmatrix-twoterms-exact}, we also have $\kappa_z=1+\bV[w_i| Z_i=z]$. To see this, by conditioning on $(X_i,Z_i)$ and summing over $Z_i$,
\begin{align*}
    &\bE[Z_iw_i] = r_1\,\bE[h]/m = r, &\quad
    &\bE[Z_iw_i^2] = r_1^2 \,\bE[h^2/e]/m^2 = r^2 \nu_1, \\
    &\bE[(1-Z_i)w_i] = r_0\,\bE[h]/m = 1-r, &\quad
    &\bE[(1-Z_i)w_i^2] = r_0^2 \,\bE[h^2/(1-e)]/m^2 = (1-r)^2\nu_0.
\end{align*}
Hence, $\bE[w_i\mid Z_i\!=\!z] = 1$ and $\bE[w_i^2\mid Z_i\!=\! z] = r_z\nu_z$, and thus
\begin{equation*}
    \kappa_z = r_z \nu_z = 1+\bV[w_i\mid Z_i=z].
\end{equation*}
Therefore, $\kappa_z \geq 1$, with equality if and only if $w_i$ is constant within arm $z$.

\subsection{Working variance and sample size formula}

Dropping $\Delta_1$, $\Delta_0$ and $\epsilon$ from $V_w$ in Theorem \ref{thm:decomposition} gives the working variance
\begin{equation}\label{eqapp:working-variance-unit}
    \widetilde{V}_w  
    = \frac{\kappa_1}{rd_1^w} + \frac{\kappa_0}{(1-r)d_0^w}.
\end{equation}
This section details the beta approximation to propensity score distributions that is needed to calculate $\kappa_1$ and $\kappa_0$ in the design stage (\ref{secapp:beta-approximation}), the resulting sample size formulas for randomized trials and observational studies (\ref{secapp:workingvar-corollaries}), propositions on properties of these formulas (\ref{secapp:workingvar-propositions}) and lemma they invoke (\ref{secapp:workingvar-lemma}).

\subsubsection{Beta approximation of the propensity score distribution}\label{secapp:beta-approximation}

Under a logistic propensity score model, $e(X)$ is approximately $\mathrm{Beta}(a,b)$ \citep{liu2025sample}, with density $p(u)=u^{a-1}(1-u)^{b-1}/B(a,b)$ where $B(a,b)$ is the Beta function. Let $p_z$ be the density of $e(X)$ given $Z=z$ and $\phi:=\int_0^1\min\{p_1(u),p_0(u)\}\,du$ the overlap coefficient. Then
\begin{itemize}
    \item[(i)] $(r,\phi)$ and $(a,b)$ satisfy
    \begin{equation}\label{eqapp:ab-rphi}
        r = \frac{a}{a + b}, \qquad \phi=1-\frac{\Gamma(a+b+1)}{\Gamma(a+1)\Gamma(b+1)}\,r^a(1-r)^b;
    \end{equation}
    \item[(ii)] for fixed $r\in(0,1)$, $\phi$ is continuous and strictly increasing in $a+b$, tending to $0$ as $a+b\to0$ and to $1$ as $a+b\to\infty$; hence, equation \eqref{eqapp:ab-rphi} maps $(r,\phi)\in(0,1)^2$ bijectively to $(a,b)\in(0,\infty)^2$, and $\phi\to1$ if and only if $a+b\to\infty$;
    \item[(iii)] for $r\geq1/2$, $a,b>1$ if and only if $\phi>1-r^{r/(1-r)}$, a threshold symmetric about $r=1/2$ that increases from $1/2$ toward $1-e^{-1}\approx0.63$ as $r\to1$.
\end{itemize}

\begin{proof}[of (i)--(iii)]
\leavevmode

\textit{(i)} 
Under Beta approximation $e(X) \!\sim\! \mathrm{Beta}(a,b)$. Since $\Pr\{Z\!=\!1\mid e(X)\!=\!u\}=u$, $r=\Pr(Z\!=\!1)=\bE[e(X)]=a/(a+b)$, giving the first identity. By Bayes' rule, $p_1(u)=u\,p(u)/r$ and $p_0(u)=(1-u)\,p(u)/(1-r)$, so $p_1(u)-p_0(u)=(u-r)\,p(u) /\{r(1-r)\}$ and is positive if and only if $u>r$. Moreover, $\min\{p_1(u),p_0(u)\}=p_1(u)-\{p_1(u)\!-\!p_0(u)\}^+$, and thus $\phi=1-\int_0^1\{p_1(u)\!-\!p_0(u)\}^+du$. With $d\{u^a(1-u)^b\}/du=(a+b)(r-u)\,u^{a-1}(1-u)^{b-1}$ and $r(1-r)=ab/(a+b)^2$,
\begin{align*}
    1 - \phi & = \int_0^1\{p_1(u)-p_0(u)\}^+du 
    = \int_r^1 \left\{p_1(u)-p_0(u)\right\}du \\
    & = \int_r^1 (u-r) \, u^{a-1}(1-u)^{b-1} \big/ \{r(1-r)B(a,b)\} \,du \\
    & = - \frac{1}{r(1-r) (a+b) B(a,b)}
     \int_r^1  d\{u^a(1-u)^b\} \\ 
    & = \frac{a+b}{ab} \frac{\Gamma(a+b)}{\Gamma(a)\Gamma(b)}r^a (1-r)^b = \frac{\Gamma(a+b+1)}{\Gamma(a+1)\Gamma(b+1)} r^a (1-r)^b,
\end{align*}
which established equation \eqref{eqapp:ab-rphi}.

\textit{(ii)} Fix $r$ and write $n=a+b$, so $a=rn$ and $b=(1-r)n$. By equation \eqref{eqapp:ab-rphi} and its proof,
\begin{align*}
    g(n)=g(n;r):=\log(1-\phi) & =\log\Gamma(n+1)-\log\Gamma(rn+1)-\log\Gamma\{(1-r)n+1\} \\
    & \quad + \, n\{r\log r+(1-r)\log(1-r)\}.
\end{align*}
Differentiate with respect to $n$,
\begin{align*}
    g^\prime(n) & = d\log\Gamma(n+1)/dn - r\cdot d\log\Gamma(rn+1)/d(rn) \\ 
    & \quad \, - (1-r)\cdot d\log\Gamma\{(1-r)n+1\}/d\{(1-r)n\} + r\log r+(1-r)\log(1-r) \\
    & = k(n) - r \, k(rn) - (1-r)\, k\{(1-r)n\},
\end{align*}
where $k(x):=d\log\Gamma(x+1)/dx-\log x$. By Binet's Formula \citep[][Chapter 12]{WhittakerWatson1996},
\begin{equation*}
    \log\Gamma(x+1) = \left(x+\frac12\right) \log x - x 
    + \frac{1}{2}\log(2\pi) 
    + \int_0^\infty \left(\frac{1}{\exp(t)-1} - \frac{1}{t}+\frac{1}{2} \right) \frac{\exp(-xt)}{t}\,dt,
\end{equation*}
and differentiating with respect to $x$ and with $\int_0^\infty \exp(-xt)dt=1/x$ gives
\begin{equation*}
    k(x)= \int_0^\infty \left\{\frac{1}{t} - \frac{1}{\exp(t)-1}  \right\} \exp(-xt) \,dt,
\end{equation*}
which is positive and strictly decreasing because $\exp(t)-1>t$. Therefore, $k(n)<k(rn)$ and $k(n) < k\{(1-r)n\}$ and thus $g^\prime(n)<0$, so $g(n)$ is decreasing in $n$. Moreover, $g(n)\to 0$ as $n \to 0$; Stirling's formula of the form $\log\Gamma(x+1) = (x+1/2)\log x - x + \log(2\pi)/2 + o(1)$ further gives $g(n) = -\log\{ 2\pi r(1-r)n\}/2 + o(1) \to -\infty$ as $n \to \infty$. Therefore, by $\phi=1-\exp\{g(n)\}$, $\phi$ increases strictly and continuously from $0$ to $1$ as $n$ increases from $0$ to $\infty$, which also gives the bijection.

\textit{(iii)} When $r\geq1/2$, $a\geq b$, so $a,b>1$ if and only if $b>1$, that is, $n>1/(1-r)$. By (ii), at any fixed $r$, this is equivalent to $\phi$ exceeding its value at $b=1$, $a=r/(1-r)$, which is $1-(a+1)r^a(1-r)=1-r^{r/(1-r)}$ by equation \eqref{eqapp:ab-rphi}, and thus $\phi>1-r^{r/(1-r)}$. Symmetry follows from exchanging $(a,r)$ with $(b,1-r)$, and monotonicity in $r$ from writing $r^{r/(1-r)}=\exp\{r(1-r)^{-1} \log r\}$ and $d\{r(1-r)^{-1} \log r\}/dr=(\log r+1-r)/(1-r)^2<0$.
\end{proof}

\subsubsection{Corollaries}\label{secapp:workingvar-corollaries}

\begin{corollary}[Randomized trials]\label{corapp:1}
For a randomized controlled trial, $\epsilon=0$ and
\begin{equation*}
    \widetilde{V}_{\RCT} = \frac{1}{r d_1}+\frac{1}{(1-r)d_0}
\end{equation*}
\end{corollary}

\begin{proof}[of Corollary \ref{corapp:1}]
In a randomized trial $e(X)\equiv r$, so $h$ is constant and balancing weights of equation \eqref{eqapp:balancing-weight} reduce to $w_i\equiv1$ for every weighting scheme. Hence, $P^w=P$, so $d_z^w=d_z$ and $d^w=d$. Because $X \perp Z$, $d_z =\bar{d}_z$ and $d=\bar{d}$. Moreover $w_i \equiv 1$, so $C_z=0$ and $\epsilon=0$ exactly. Plugging $\kappa_z=1$ into equation \eqref{eqapp:working-variance-unit} completes the proof.
\end{proof}

\begin{corollary}[Observational studies]\label{corapp:2}
Under a logistic model for the propensity scores with $e(X)\sim\mathrm{Beta}(a,b)$, where $(a,b)$ is determined by $(r,\phi)$ through equation \eqref{eqapp:ab-rphi}, the working variance is $\widetilde{V}_w=\kappa_1/(rd_1^w)+\kappa_0/\{(1-r)d_0^w\}$, where $\kappa_1$ and $\kappa_0$ take the values below.
\begin{center}\normalfont\small
\setlength{\tabcolsep}{4pt}
\begin{tabular}{ccccc}
\toprule
weight & $\widetilde{V}_w$  & $\kappa_1$ & $\kappa_0$ & regularity condition \\
\midrule
inverse probability & $\widetilde{V}_{\IPW}$ & $\dfrac{a(a+b-1)}{(a+b)(a-1)}$ & $\dfrac{b(a+b-1)}{(a+b)(b-1)}$ & $a>1,\;b>1$\\[12pt]
overlap & $\widetilde{V}_{\OW}$ & $\dfrac{(b+1)(a+b+1)}{b\,(a+b+2)}$ &
$\dfrac{(a+1)(a+b+1)}{a\,(a+b+2)}$ & $a>0,\;b>0$\\[12pt]
ATT & $\widetilde{V}_{\ATT}$ & $1$ & $\dfrac{b(a+1)}{a\,(b-1)}$ & $b>1$\\
\bottomrule
\end{tabular}
\end{center}
In particular, for inverse probability weights, $d_z^w=d_z$ and
\begin{equation*}
    \widetilde{V}_{\IPW} = (a+b-1) \left[\frac{1}{(a-1)d_1} + \frac{1}{(b-1)d_0}\right].
\end{equation*}
For all three weights, $\kappa_1\to1$ and $\kappa_0\to1$ as $\phi\to1$.
\end{corollary}

\begin{proof}[of Corollary \ref{corapp:2}]
The form of $\widetilde{V}_w$ is equation \eqref{eqapp:working-variance-unit}. Recall Section \ref{secapp:design-effect} that $\kappa_z=r_z\nu_z$ with $\nu_z=\bE[h^2/e_z]/\bE[h]^2$, where the tilting function $h$ is a polynomial of $e$ or $(1-e)$ for all of inverse probability, overlap, and ATT weights. Therefore, under $e\sim\mathrm{Beta}(a,b)$, $\nu_z$ is determined by the Beta moments $\bE[e^p(1-e)^q]=B(a+p,b+q)/B(a,b)$, where $B(a,b)$ is the Beta function with parameters $a$ and $b$, and $r_1=r=a/(a+b)$ and $r_0=1-r=b/(a+b)$.

For inverse probability weights ($h\equiv1$), $\nu_z=\bE[1/e_z]$ and thus $\nu_1=\bE[e^{-1}]=B(a-1,b)/B(a,b)=(a+b-1)/(a-1)$ and $\nu_0=\bE[(1-e)^{-1}]=B(a,b-1)/B(a,b)=(a+b-1)/(b-1)$, and multiplying them by $r_1$ and $r_0$, respectively, gives $\kappa_1$ and $\kappa_0$ in the table. Moreover, $P^w=P$ and $d_z^w=d_z$. Plugging in $\kappa_1/r=(a+b-1)/(a-1)$ and $\kappa_0/(1-r)=(a+b-1)/(b-1)$ to $\widetilde{V}_w$ gives $\widetilde{V}_{\IPW}$.

For overlap weights ($h=e(1-e)$), $\bE[h]=B(a+1,b+1)/B(a,b)=ab/[(a+b)(a+b+1)]$, $\bE[h^2/e]=B(a+1,b+2)/B(a,b)$, and $\bE[h^2/(1-e)]=B(a+2,b+1)/B(a,b)$, which gives
\begin{equation*}
    \nu_1 = \frac{(b+1)(a+b)(a+b+1)}{ab\,(a+b+2)}, \quad \nu_0 = \dfrac{(a+1)(a+b)(a+b+1)}{ab\,(a+b+2)},
\end{equation*}
and similarly multiplying with $r_1$ and $r_0$ gives the table values.

For ATT weights ($h=e$), $\bE[h]=r$, $\bE[h^2/e]=\bE[e]=r$, and $\bE[h^2/(1-e)]=B(a+2,b-1)/B(a,b)=a(a+1)/[(b-1)(a+b)]$, and hence
\begin{equation*}
    \nu_1 = 1/r = \frac{a+b}{a}, \quad \nu_0 = \frac{a(a+1)}{r^2(b-1)(a+b)} = \frac{(a+1)(a+b)}{a(b-1)}.
\end{equation*}
The regularity conditions ensure that the corresponding Beta functions are finite, so $\kappa_1$ and $\kappa_0$, and hence $V$, are finite. Finally, as $\phi\to1$, $a+b\to\infty$ at fixed $r$ (Section \ref{secapp:beta-approximation}), so $a=r(a+b)\to\infty$ and $b=(1-r)(a+b)\to\infty$, and every $\kappa_z$ in the table tends to one.
\end{proof}

\subsubsection{Lemmas}\label{secapp:workingvar-lemma}

\begin{lemma}[Concentration under the Beta approximation]\label{lemapp:6}
Let $e\sim\mathrm{Beta}(a,b)$, where $(a,b)$ is determined by $(r,\phi)$ through equation \eqref{eqapp:ab-rphi} with $r \in (0,1)$ fixed. For any Borel function $q$ on $(0,1)$ that is continuous at $r$ and satisfies $\limsup_{\phi\to1}\bE[q(e)^2]<\infty$, it holds that $\bE[q(e)]\to q(r)$ as $\phi \to 1$. In particular, for any function $g$ bounded on $(0,1)$ and continuous at $r$, and integer $k\geq 0$,
\begin{equation*}
    \bE\!\left[\frac{g(e)}{e_z^k}\right]\to\frac{g(r)}{r_z^k}, \qquad z\in\{0,1\},
\end{equation*}
as $\phi\to1$, where $e_1=e$, $e_0=1-e$, $r_1=r$, and $r_0=1-r$.
\end{lemma}

\begin{proof}[of Lemma \ref{lemapp:6}]
By fact (ii) of Section \ref{secapp:beta-approximation}, at a fixed $r$, $\phi\to1$ if and only if $a+b\to\infty$. Under $e\sim\mathrm{Beta}(a,b)$, $\bE[e]=r$ and $\bV[e]=r(1-r)/(a+b+1)$, so $\bV[e]\to0$ as $\phi\to1$.

Fix $\delta\in(0,\min\{r,1-r\})$ and let $\mathcal{A}_\delta=\{|e-r|\le\delta\}$. By Chebyshev's inequality, $\Pr(\mathcal{A}_\delta^c)\le\bV[e]/\delta^2\to0$. Splitting the expectation on $\mathcal{A}_\delta$ and $\mathcal{A}_\delta^c$,
\begin{align*}
    \left|\bE[q(e)]-q(r)\right|
    & \le \bE\!\left[ |q(e)-q(r)| \cdot \bI(\mathcal{A}_\delta) \right] + \bE\!\left[|q(e)-q(r)| \cdot \bI(\mathcal{A}_\delta^c)\right] \\
    & \le \sup_{|u-r|\le\delta}|q(u)-q(r)| + \left\{\bE\!\left[\{q(e)-q(r)\}^2\right]\right\}^{1/2} \cdot \left\{\Pr(\mathcal{A}_\delta^c)\right\}^{1/2} \\
    & \le \sup_{|u-r|\le\delta}|q(u)-q(r)| + \left[\left\{\bE[q(e)^2]\right\}^{1/2} + |q(r)| \right] \cdot \left\{\Pr(\mathcal{A}_\delta^c)\right\}^{1/2},
\end{align*}
where the first inequality applies the triangle inequality, the second applies the Cauchy-Schwarz inequality to the second term, and the third is Minkowski's inequality. The bracket in the last step is bounded as $\phi \to 1$ by assumption and $\Pr(\mathcal{A}_\delta^c)\to0$, so $\limsup_{\phi\to1}|\bE[q(e)]-q(r)|\le\sup_{|u-r|\le\delta}|q(u)-q(r)|$ for every such $\delta$. Letting $\delta\downarrow0$, then the right-hand side converges to zero because $q$ is continuous at $r$.

For the special case with $z=1$ so that $e_z=e$, let $q(u)=g(u)/u^k$; then $q$ is continuous at $r$ because $r>0$ and $g$ is continuous at $r$. Let $|g| \leq G < \infty$ for some constant $G$; then $\bE[q(e)^2] = \bE[g(e)^2/e^{2k}] \leq G^2 \,\bE[e^{-2k}]$, and for $a>2k$,
\begin{equation*}
    \bE[e^{-2k}] = \frac{B(a-2k,b)}{B(a,b)} = \frac{\Gamma(a-2k)\,\Gamma(a+b)}{\Gamma(a)\,\Gamma(a+b-2k)} = \prod_{i=1}^{2k} \frac{a+b-i}{a-i} \ \to\ \frac{1}{r^{2k}},
\end{equation*}
because each factor $(a+b-i)/(a-i)\to 1/r$ as $a+b \to \infty$ with $a=r(a+b)$. Therefore, $\limsup_{\phi \to 1} \bE[q(e)^2] \leq G^2r^{-2k}<\infty$, so $\bE[q(e)]= \bE[g(e)/e^k] \to g(r)/r^k$. The case $z=0$ with $e_z=1-e$ follows in the same way with $q(u)=g(u)/(1-u)^k$ and $\bE[(1-e)^{-2k}]=B(a,b-2k)/B(a,b)$ for $b>2k$, with $b$ in place of $a$ and $1-r$ in place of $r$.
\end{proof}

\subsubsection{Propositions}\label{secapp:workingvar-propositions}

\textbf{Proposition \ref{prop:working-error}}. We prove proposition \ref{prop:working-error} under any balancing weights $w_i$ from equation \eqref{eqapp:balancing-weight}. To this end, proposition \ref{prop:working-error} is restated as:

\textit{Under the conditions and with $\epsilon$ of Theorem~\ref{thm:decomposition}, $V_w$ can be expressed in $\widetilde{V}_w$ as $V_w=\widetilde{V}_w(1+\Delta) + \epsilon$, where $\Delta\geq0$ and $\Delta=0$ if and only if $\pi(\tau_0,t)$ is constant $\mu$-almost everywhere. Moreover, 
\begin{align}
    0 \; \leq \; \Delta \; 
    & \; \leq \; \frac{(\pi_u-\pi_l)^2}{4} \cdot \widetilde{V}_w^{-1} \left[\frac{\kappa_1}{rd_1^w (1-\pi_u)(1-\pi_l)} + \frac{\kappa_0}{(1-r)d_0^w \,\pi_u\pi_l}\right] \label{eqapp:delta-bound-tight} \\
    & \; \leq \; (\pi_u-\pi_l)^2/[4\min\{\pi_l\pi_u,(1-\pi_l)(1-\pi_u)\}] \label{eqapp:delta-bound-loose}
\end{align}
where $\pi_l$ and $\pi_u$ are the infimum and supremum of $\pi(\tau_0,t)$ over $[0,t^\dagger]$, respectively, assumed to lie in $(0,1)$.}

\begin{proof}[of Proposition \ref{prop:working-error}]

From Theorem \ref{thm:decomposition} and $\widetilde{V}_w$ in equation \eqref{eqapp:working-variance-unit},
\begin{equation*}
    V_w-\epsilon-\widetilde{V}_w = \frac{\kappa_1\Delta_1}{rd_1^w} + \frac{\kappa_0\Delta_0}{(1-r)d_0^w} =: \widetilde{V}_w \Delta,
\end{equation*}
which gives $V_w = \widetilde{V}_w (1+\Delta) + \epsilon$ and thus
\begin{equation*}
    \Delta = \frac{1}{\widetilde{V}_w} \left[ \frac{\kappa_1\Delta_1}{rd_1^w} + \frac{\kappa_0\Delta_0}{(1-r)d_0^w} \right].
\end{equation*}

Now we bound $\Delta$. Since $\Delta_z \geq 0$, $\Delta \geq 0$. Assume $\Delta\!=\!0$, then $\Delta_1\!=\!\Delta_0\!=\!0$, and thus $\bV_{\mu_0}[\pi]\!=\!\bV_{\mu_1}[\pi]\!=\!0$, that is, $\pi \!=\! \bE_{\mu_1}[\pi]$ $\mu_1$-almost everywhere and $\pi \!=\! \bE_{\mu_0}[\pi]$ $\mu_0$-almost everywhere. By equation \eqref{eqapp:arm-event-measure-density}, they imply $\pi \in \{0,\bE_{\mu_1}[\pi]\}$ and $\pi\in\{\bE_{\mu_0}[\pi],1\}$ $\mu$-almost everywhere; since $\bE_{\mu_0}[\pi]>0$ and $\bE_{\mu_1}[\pi]<1$ (Lemma \ref{lemapp:4}), it must be $\pi\!=\!\bE_{\mu_0}[\pi]\!=\!\bE_{\mu_1}[\pi]$ $\mu$-almost everywhere. Conversely, if $\pi$ is a constant $\mu$-almost everywhere, then $\bV_{\mu_0}[\pi]=\bV_{\mu_1}[\pi]=0$, implying $\Delta=0$. So $\Delta=0$ if and only if $\pi$ is a constant $\mu$-almost everywhere.

For the upper bound, suppose a random variable $\pi$ is on $[\ell,u]$ with $\ell>0$. Taking the expectation of $(u-\pi)(\pi-\ell)\geq 0$,
\begin{equation*}
    0\le\bE[(u-\pi)(\pi-\ell)] = -\bE[\pi^2] + (u+\ell)\,\bE[\pi] - u\ell
    = -\bV[\pi]+(u-\bE[\pi])(\bE[\pi]-\ell),
\end{equation*}
and hence, with $s:= 1/\bE[\pi] \in [1/u, 1/\ell]$,
\begin{align*}
    \frac{\bV[\pi]}{\bE[\pi]^2} & \leq 
    \frac{(u-\bE[\pi])(\bE[\pi]-\ell)}{\bE[\pi]^2} = (us-1)(1-\ell s)
    \leq (us^*-1)(1-\ell s^*) \\
    & = \frac{u-\ell}{2\ell}\cdot\frac{u-\ell}{2u} = \frac{(u-\ell)^2}{4u\ell},
\end{align*}
where $s^*= (1/u+1/\ell)/2$ is the point where the quadratic $(us-1)(1-\ell s)$ in $s$ is maximized. Taking $\pi_l:= \inf_{t}\pi(\tau_0,t)$ and $\pi_u: = \sup_{t}\pi(\tau_0,t)$ over $t \in [0,t^\dagger]$, and applying the inequality to $1-\pi(\tau_0,\widetilde{T}_1) \in [1-\pi_u,1-\pi_l]$ and $\pi(\tau_0,\widetilde{T}_0) \in [\pi_l,\pi_u]$, with $\widetilde{T}_z \sim \mu_z(t)/\mu_z(t^\dagger)$,
\begin{equation*}
    \Delta_1 = \frac{\bV_{\mu_1}[1-\pi]}{\left(\bE_{\mu_1}[1-\pi]\right)^2} \leq \frac{(\pi_u-\pi_l)^2}{4(1-\pi_u)(1-\pi_l)}, \qquad
    \Delta_0 = \frac{\bV_{\mu_0}[\pi]}{\left(\bE_{\mu_0}[\pi]\right)^2} \leq \frac{(\pi_u-\pi_l)^2}{4\pi_u\pi_l},
\end{equation*}
and plugging both back to $\Delta$ gives inequality \eqref{eqapp:delta-bound-tight}.

Because $\kappa_1/(rd_1^w\widetilde{V}_w)$ and $\kappa_0/\{(1-r)d_0^w\widetilde{V}_w\}$ are positive and sum to one by equation \eqref{eqapp:working-variance-unit}, the right-hand side of inequality \eqref{eqapp:delta-bound-tight} is a product of $(\pi_u-\pi_l)^2/4$ and a weighted average of $1/\{(1-\pi_u)(1-\pi_l)\}$ and $1/(\pi_u\pi_l)$, which is at most $1/\min\{\pi_l\pi_u,(1-\pi_l)(1-\pi_u)\}$. This further gives inequality \eqref{eqapp:delta-bound-loose}.
\end{proof}

\vspace{12pt}

\noindent \textbf{Proposition \ref{prop:Schoenfeld}}.

\begin{proof}[of Proposition \ref{prop:Schoenfeld}]
Because $\widetilde{V}_{\SC} = 1/\{r(1-r)d\}=4/d$,
\begin{align*}
    \frac{\widetilde{V}_{\RCT}}{\widetilde{V}_{\SC}} & =
    \left(\frac{2}{d_1} + \frac{2}{d_0} \right)\! \Big/ \frac{4}{d} 
    = \left( \frac{2}{d_1} + \frac{2}{d_0} \right) \frac{d_1+d_0}{4 \cdot 2}
    = \frac{(d_1+d_0)^2}{4\,d_1d_0} \\
    &= \frac{1}{4}\left\{\left(\frac{d_1}{d_0}\right)^{\!1/2} \! + \left(\frac{d_0}{d_1}\right)^{\!1/2}\right\}^2 = \left\{\frac{\exp(\vartheta/2)+\exp(-\vartheta/2)}{2}\right\}^2 = \cosh^2(\vartheta/2) \geq 1,
\end{align*}
with equality if and only if $\vartheta=0$, that is, $d_1=d_0$.

For the second claim under $G_1=G_0=:G$ and H1, let $\widetilde{C}_i := C_i \wedge t^\dagger$. Hence, $\Pr(\widetilde{C}_i \geq t \mid Z_i=z)= \Pr(\widetilde{C}_i \geq t) = G(t)$ for $t \leq t^\dagger$. In a randomized trial $P^w=P$, so we drop the superscript $w$. From equation \eqref{eqapp:standardized-event-rate} and Fubini, for $z=0,1$,
\begin{equation*}
    d_z = \int_0^{t^\dagger}f_z(t) G(t)\,dt
    = \int_0^{t^\dagger}f_z(t) \, \bE[\bI(\widetilde{C}\geq t)] \,dt 
    = \bE\!\left[\int_0^{\widetilde{C}} f_z(t)\,dt\right]
    = \bE\!\left[1-F_z(\widetilde{C})\right].
\end{equation*}
Under H1, $F_0(t)^{\exp(\tau_0)}=F_1(t)$, and thus
\begin{equation*}
    d_1 - d_0 = \bE\!\left[1 - F_0(\widetilde{C})^{\exp(\tau_0)}\right] - \bE\!\left[1-F_0(\widetilde{C})\right] = \bE\!\left[F_0(\widetilde{C}) - F_0(\widetilde{C})^{\exp(\tau_0)}\right].
\end{equation*}
Notice that $F_0(t)\!>\!0$ on $[0, t^\dagger]$, i.e., $F_0(t^\dagger)\!>\!0$, because $F_0(t^\dagger)\!=\!\exp\{-\Lambda_0(t^\dagger)\}$ and $\Lambda_0(t^\dagger) < \infty$ by regularity conditions. Moreover, $\Pr\{F_0(\widetilde{C})\!<\!1\} > 0$, because otherwise $d_0 = \bE[1-F_0(\widetilde{C})] = 0$, contradicting $A(\tau_0)>0$. Therefore, $d_1-d_0$ has the sign of $\tau_0$, and $d_1=d_0$ if and only if $\tau_0=0$.

\textit{Remark}. This comparison holds for both the number of units and the events, because the variance used to calculate the latter multiplies both $\widetilde{V}_{\RCT}$ and $\widetilde{V}_{\SC}$ by $d$, which cancels in their ratio.
\end{proof}

\vspace{12pt}

\begin{proposition}\label{propapp:3}
Under the proportional hazards alternative (H1), let
\begin{equation*}
    \sigma_z = \frac{\{\bV^w[w_z(X)]\}^{1/2}}{\bE^w[w_z(X)]}, \qquad
    M_z: = \max\{1,\,-\exp(\tau_0z)\log\gamma\},
\end{equation*}
for $z=0,1$, where $\sigma_z$ is the coefficient of variation of $w_z$ and $\gamma:=F_0^w(t^\dagger)$ is the survival probability at $t^\dagger$ under control in the target population. Let
\begin{equation*}
    \chi_z := \frac{\{\bV^w[Q_z(X)]\}^{1/2}}{\bE^w[Q_z(X)]}, \qquad
    \rho_z := \left\{ \frac{\bV^w[\bE^w\{Q_z(X)\mid e(X)\}]}{\bV^w[Q_z(X)]} \right\}^{1/2},
\end{equation*}
where similarly, $\chi_z$ is the coefficient of variation of $Q_z$ and $\rho_z^2$ is the $R$-square statistic of the regression $Q_z \sim e$ in the target population, with $Q_z$ as in Theorem \ref{thm:decomposition}. Let $\pi_l$ and $\pi_u$ be as in Proposition \ref{prop:working-error} and
\begin{equation*}
    L_1:= \left(\frac{1-\pi_u}{1-\pi_l}\right)^{2} d_1^w, \quad 
    L_0:= \left(\frac{\pi_l}{\pi_u}\right)^{2} d_0^w.
\end{equation*}
Then,
\begin{itemize}
    \item [(i)] $0 \leq \rho_z \leq 1$, and if $\bE[w_i^3|Z_i=z] < \infty$ for $z=0,1$, then
    \begin{align*}
    \frac{|\epsilon|}{V_w-\epsilon} 
    & \;\leq\; \max_{z\in\{0,1\}}  \sigma_z\rho_z\chi_z
    \;\leq\; \max_{z\in\{0,1\}} \sigma_z\rho_z\left(\frac{M_z^2}{L_z}-1\right)^{1/2}
\end{align*}

    \item [(ii)] Under $e(X) \!\sim\! \mathrm{Beta}(a,b)$, $\sigma_z\to0$ as $\phi\to1$ for all propensity-score balancing weights continuous and positive at $r$. Hence, for a given event process, $\epsilon\to0$ and $|\epsilon|/(V_w-\epsilon)\to0$ as $\phi \to 1$.
\end{itemize}
\end{proposition}

\begin{proof}[of Proposition \ref{propapp:3}]
\leavevmode

\textit{(i)}
Recall from the proof of Theorem \ref{thm:decomposition} that
\begin{equation*}
    \epsilon = \frac{r\,C_1+(1-r)\,C_0}{A(\tau_0)^2}, \qquad
    C_z = \Cov^w\!\left\{w_z(X), Q_z(X)\right\},
\end{equation*}
with $Q_z(x) = \bE\!\left[\left\{\delta_i\left[z-\pi(\tau_0,T_i)\right] - \exp(\tau_0 z)\,\Omega_z(T_i)\right\}^2 \,\middle|\,Z_i=z,X_i=x\right]$,
and the last step of the proof of Lemma \ref{lemapp:5} that $B(\tau_0) - B_\epsilon(\tau_0) = \kappa_1 r\,\bE^w[Q_1] + \kappa_0 (1-r)\,\bE^w[Q_0]$, so that
\begin{align}\label{eqapp:epsilon-relative-bound-step1}
    \frac{|\epsilon|}{V_w-\epsilon} = \frac{ \left| r C_1 \!+\! (1-r) C_0 \right|}{A(\tau_0)^2} \cdot \frac{A(\tau_0)^2}{B(\tau_0) \!-\! B_\epsilon(\tau_0)} 
    & = \frac{\left| r C_1 + (1-r) C_0 \right|}{\kappa_1 r\,\bE^w[Q_1] + \kappa_0 (1-r)\,\bE^w[Q_0]} \nonumber \\
    & \leq \frac{r \left|C_1\right| + (1-r) \left|C_0\right| }{\kappa_1 r\,\bE^w[Q_1] + \kappa_0 (1-r)\,\bE^w[Q_0]},
\end{align}
so it remains to bound $|C_z|$ relative to $\bE^w[Q_z]$.

We first deal with $|C_z|$. First, because $w_z(X)$ is a function of $e(X)$,
\begin{equation*}
    \bE^w[w_zQ_z] = \bE^w\!\left[ \, \bE^w\{w_zQ_z \mid e \} \right] = \bE^w\!\left[ \,w_z \,\bE^w\{Q_z \mid e \} \right],
\end{equation*}
where each integrand is a function of $X$. Denote $q_z(X):= \bE^w[Q_z(X) \mid e(X)]$. By the total expectation, $\bE^w[q_z(X)]=\bE^w[Q_z(X)]$, so
\begin{align*}
    C_z = \Cov^w(w_z, Q_z) 
    & = \bE^w[w_z Q_z] - \bE^w[w_z]\bE^w[Q_z] \\
    & = \bE^w[w_z q_z] - \bE^w[w_z]\bE^w[q_z]
    = \Cov^w(w_z, q_z).
\end{align*}
By the Cauchy-Schwarz inequality and the definitions of $\sigma_z$, $\rho_z$ and $\chi_z$,
\begin{align*}
    |C_z| & \leq \left\{\bV^w\!\left[w_z\right]\right\}^{1/2} \cdot \left\{\bV^w[q_z]\right\}^{1/2} \\
    & = \sigma_z \bE^w\!\left[w_z\right] \cdot \rho_z\left\{\bV^w[Q_z]\right\}^{1/2}
    = \sigma_z \kappa_z \,\rho_z\, \chi_z\bE^w[Q_z],
\end{align*}
where the second equation uses $\bE^w\!\left[w_z\right]=\kappa_z$ (Section \ref{secapp:design-effect}). Furthermore, $\bE^w[Q_z] > 0$, because by the proof of Lemma \ref{lemapp:5},
\begin{equation*}
    \bE^w[Q_1] 
    = r^{-1} \int_0^{t^\dagger}(1-\pi)^2 d\mu_1, \qquad 
    \bE^w[Q_0] 
    = (1-r)^{-1} \int_0^{t^\dagger} \pi^2 d\mu_0,
\end{equation*}
which if zero for any of $z=0,1$ would imply $A(\tau_0) = 0$ (Lemma \ref{lemapp:4}) and contradict the regularity condition.
Therefore, plugging the inequality of $|C_z|$ above into equation \eqref{eqapp:epsilon-relative-bound-step1}, and applying the inequality $(x_1+x_0)/(y_1+y_0) \leq \max\{x_1/y_1, x_0/y_0\}$ for $x_z \geq 0$ and $y_z>0$,
\begin{equation*}
    \frac{|\epsilon|}{V_w-\epsilon} \quad\leq\quad \max_{z \in \{0,1\}}\; \sigma_z \,\rho_z \,\chi_z,
\end{equation*}
which gives the first inequality.

It remains to bound $\chi_z$. We now show $0 \leq Q_1(x) \leq M_1^{*2}$ and $0 \leq Q_0(x) \leq M_0^{*2}$, where $M_1^* := (1-\pi_l)M_1$ and $M_0^* := \pi_uM_0$, where $\pi_l$ and $\pi_u$ are bounds for $\pi(\tau_0,t)$ as in Proposition \ref{prop:working-error}. By definition, $Q_z(x) \geq 0$. For $Q_z(x) \leq M_z^{*2}$, see that $Q_z(x)$ depends on $\Omega_z(t) = \int_0^t\left[z-\pi(\tau_0,s)\right] d\Lambda_0^w(s)$. By definitions in Section \ref{secapp:weight-and-population}, $\Lambda_0^w(t) = - \log F_0^w(t)$, and thus $\Lambda_0^w(t^\dagger) = - \log\gamma$ since $\log F_0^w(0)=0$; under H1, $\Lambda_0^w=\Lambda_w$, so $\Lambda_0^w(t^\dagger)$ is finite, i.e., $\gamma > 0$. Therefore, for all $t \leq t^\dagger$,
\begin{align*}
    \Omega_1(t) & = \int_0^t \left\{1 - \pi(\tau_0,s)\right\} d\Lambda_0^w(s)   \in   [0,\,-(1-\pi_l)\log\gamma], \\
    \Omega_0(t) & = -\int_0^t\pi(\tau_0,s)\,d\Lambda_0^w(s) \in   [\pi_u\log\gamma,\,0].
\end{align*}
Since $T_i\leq t^\dagger$ and $\delta_i \in \{0,1\}$, for the integrand of $Q_1$ and $Q_0$, respectively,
\begin{align*}
    1-\pi_l & \geq \delta_i\left\{1-\pi(\tau_0,T_i)\right\}-\exp(\tau_0)\,\Omega_1(T_i) \geq (1-\pi_l)\exp(\tau_0)\log\gamma, \\
    -\pi_u & \leq -\delta_i\,\pi(\tau_0,T_i)-\Omega_0(T_i) \leq -\pi_u\log\gamma.
\end{align*}
Taking the square bounds the middle terms by $M_1^{*2}$ and $M_0^{*2}$, respectively, and so  are their conditional means, i.e., $Q_z(x) \leq M_z^{*2}$. Hence $Q_z(x)^2 \leq M_z^{*2} \,Q_z(x)$, which implies
\begin{equation*}
    \bV^w[Q_z] = \bE^w[Q_z^2] - \left(\bE^w[Q_z]\right)^2 \;\leq\; M_z^{*2} \, \bE^w[Q_z] - \left(\bE^w[Q_z]\right)^2
    = \bE^w[Q_z] \left(M_z^{*2}-\bE^w[Q_z]\right),
\end{equation*}
that is, $\chi_z\leq(M_z^{*2}/\bE^w[Q_z]-1)^{1/2}$.

Hence, it remains to bound $\bE^w[Q_z]$. Using the integral expression for $\bE^w[Q_z]$ again gives
\begin{align*}
    \bE^w[Q_1]
    & \geq r^{-1}(1-\pi_u)^2 \mu_1(t^\dagger)
    = (1-\pi_u)^2 d_1^w, 
    \qquad
    \bE^w[Q_0]
    \geq (1-r)^{-1} \pi_l^2 \mu_0(t^\dagger) 
    = \pi_l^2 d_0^w,
\end{align*}
and hence
\begin{equation*}
    \frac{M_1^{*2}}{\bE^w[Q_1]}\le\frac{(1-\pi_l)^2M_1^2}{(1-\pi_u)^2\,d_1^w} =: \frac{M_1^2}{L_1}, \qquad
    \frac{M_0^{*2}}{\bE^w[Q_0]}\le\frac{\pi_u^2M_0^2}{\pi_l^2\,d_0^w} =: \frac{M_0^2}{L_0},
\end{equation*}
where $L_z$ follows this Proposition's statement. Therefore $\chi_z \leq (M_z^2/L_z-1)^{1/2}$, which establishes the second inequality.

\textit{(ii)} Positivity and continuity of $w$ at $e=r$ implies the same for each $w_z= (r_z/m) \times(h/e_z)$, and thus for $h$. Hence, Lemma \ref{lemapp:6} applies. 
Notice that $\sigma_z$ is invariant to rescaling, and thus $\sigma_z=\{\bV^w[w_z]\}^{1/2}/\bE^w[w_z] = \{\bV^w[h/e_z]\}^{1/2}/\bE^w[h/e_z]$. Since $\bE^w[h/e_z]=\bE[h^2/e_z]/\bE[h]$ and $\bE^w[(h/e_z)^2]=\bE[h^3/e_z^2]/\bE[h]$ by definition of $\bE^w$, Lemma \ref{lemapp:6} with $g=h$, $h^2$, and $h^3$ (see that they are all bounded on $(0,1)$), and $k=0$, $1$, and $2$, respectively, gives
\begin{align*}
    & \bE^w\!\left[\frac{h}{e_z}\right] = \frac{\bE[h^2/e_z]}{\bE[h]}
    \to \frac{h(r)^2/r_z}{h(r)} > 0 \\
    & \bV^w\!\left[\frac{h}{e_z}\right] = \frac{\bE[h^3/e_z^2]}{\bE[h]} -\left(\frac{\bE[h^2/e_z]}{\bE[h]}\right)^{\!2} 
    \to \frac{h(r)^3 / r_z^2}{h(r)} - \left\{ \frac{h(r)^2/r_z}{h(r)} \right\}^{\!2} = 0,
\end{align*}
so $\sigma_z \to 0$.

In contrast, $M_z$ and $L_z$ are functions of the event process only that do not depend on the dispersion of the weights. Holding the event process fixed, that is, $(r, F_z^w, G_z)$ is held fixed, then $\tau_0$, $\pi(\tau_0,\cdot)$, $d_z^w$, $\gamma$, and $(\pi_l, \pi_u)$ are all hence fixed and invariant to $\phi$, so are $M_z$ and $L_z$. Therefore, $|\epsilon|/(V_w-\epsilon) \leq \max_z \rho_z \sigma_z(M_z^2/L_z -1)^{1/2} \to 0$ as $\phi \to 1$.

Moreover, by Theorem \ref{thm:decomposition},
\begin{equation*}
    V_w-\epsilon = \frac{\kappa_1(1+\Delta_1)}{rd_1^w} + \frac{\kappa_0(1+\Delta_0)}{(1-r)d_0^w}
    \quad\to\quad \frac{1+\Delta_1}{rd_1^w} + \frac{1+\Delta_0}{(1-r)d_0^w}<\infty,
\end{equation*}
where the convergence follows from Lemma \ref{lemapp:6} on $\nu_z$, so $\kappa_z \to 1$ for both $z=0,1$ as $\phi \to 1$, while $\Delta_z$ and $d_z^w$ are fixed by the event process. Hence, $\epsilon \to 0$, which completes the proof.

\textit{Remark}. $\bE[w_i^3|Z_i=z] < \infty$ is equivalent to $\bV^w[w_z(X)]<\infty$, which ensures $\sigma_z < \infty$. Under $e \sim \mathrm{Beta}(a,b)$, it is always satisfied for large enough $\phi$ as $\phi \to 1$ (Section \ref{secapp:cv-weights-beta}).
\end{proof}

\vspace{12pt}

\begin{proposition}\label{propapp:4} 
For $z\in\{0,1\}$ and assume $0 < \bar{d}_z < 1$, then 
\begin{align*}
    \left|d_z^w-\bar{d}_z \right|^2 
    & \;=\; \Cov\{w_z(X),d_z(X)\mid Z=z\}^2 \\
    & \;\leq\; (\kappa_z-1) R_{\delta \sim e \mid Z=z}^2 \, \bar{d}_z(1-\bar{d}_z),
\end{align*}
where $R_{\delta \sim e \mid Z=z}^2$ is the $R^2$ of regressing $\delta$ on $e(X)$ within arm $Z=z$. Equivalently, the bound can be written on the relative scale as
\begin{equation*}
    \left| d_z^w/\bar{d}_z - 1 \right|  \;\leq\; (\kappa_z-1)^{1/2}\,c_z,
\end{equation*}
where $c_z:=\left\{\bV[\bE\{\delta_i\mid e(X_i),Z_i=z\}\mid Z_i=z]\right\}^{1/2}/\bar{d}_z$ is the coefficient of variation of $\bE\{\delta_i\mid e(X_i),Z_i=z\}$ in arm $z$. In particular, $d_1^w=\bar{d}_1$ exactly for ATT weights, and $d_z^w - \bar{d}_z\to0$ as $\phi \to 1$ under $e(X)\sim\mathrm{Beta}(a,b)$ for all propensity score balancing weights continuous and positive at $r$.
\end{proposition}

\begin{proof}[of Proposition \ref{propapp:4}]
We first express $d_z^w$ and $\bar{d}_z$ with the conditional event rate $d_z(x):= \bE[\delta_i \mid Z_i=z, X_i=x]$ of equation \eqref{eqapp:conditional-event-rate}.
Recall the observed event rate $\bar{d}_z := \Pr(\delta_i=1 \mid Z=z) = \bE[\delta_i \mid Z_i=z]$, and therefore the law of total expectation gives
\begin{equation*}
    \bar{d}_z = \bE\left[ \bE\{\delta_i \mid Z_i=z, X_i\} \mid Z_i=z \right] = \bE\left[d_z(X_i) \mid Z_i=z \right],
\end{equation*}
where both outer expectations are taken over the distribution of $X$ given $Z=z$. For $d_z^w$, recall that $\bE[w_i\delta_iZ_i] =rd_1^w$ and $\bE[w_i\delta_i(1-Z_i)] =(1-r)d_0^w$, and hence
\begin{align*}
    & \bE[w_i\delta_iZ_i] = \bE[w_i\delta_i \mid Z_i=1]\Pr(Z_i=1) + 0 = r d_1^w, \\ 
    & \bE[w_i\delta_i(1-Z_i)] = 0 + \bE[w_i\delta_i \mid Z_i=0]\Pr(Z_i=0) = (1-r)d_0^w,
\end{align*}
which implies $\bE[w_i \delta_i \mid Z_i=z]= d_z^w$. Since $w_i=w(Z_i,X_i)$ is a function of $(Z_i,X_i)$,
\begin{equation}\label{eqapp:std-event-rate-conditionZ}
    d_z^w 
    = \bE\left[ w_i \,\bE\{ \delta_i \mid Z_i=z, X_i \} \mid Z_i=z \right]
    = \bE\left[w_id_z(X_i) \mid Z_i=z \right].
\end{equation}
Recall $\bE[w_i\mid Z_i=z]=1$ by the definition of $w$ in equation \eqref{eqapp:balancing-weight}, and hence
\begin{align*}
    |d_z^w-\bar{d}_z|^2 & = \left|\; \bE\!\left[w_i\,d_z(X_i) \mid Z_i=z\right] - \bE[w_i\mid Z_i=z] \cdot \bE\!\left[d_z(X_i)\mid Z_i=z\right] \;\right|^2 \\
    & = \Cov\!\left\{w_i,\,d_z(X_i)\mid Z_i=z\right\}^2 \\
    & = \Cov\left[w_i,\, \bE\{ d_z(X_i) \mid e(X_i), Z_i=z \} \mid Z_i=z\right]^2,
\end{align*}
where the last equation is from the total expectation within arm $z$, in which $w_i=w_z(X_i)$ is a function of $e(X_i)$. Further applying Cauchy-Schwarz gives
\begin{align}\label{eqapp:event-rate-diff-CauchySchwarz}
    |d_z^w-\bar{d}_z|^2 & \leq \bV[w_i\mid Z_i=z] \cdot \bV[\, \bE\{ d_z(X_i) \mid e(X_i), Z_i=z \} \mid Z_i=z] \nonumber \\
    & = \bV[w_i\mid Z_i=z] \cdot R_{\delta \sim e \mid Z=z}^2 \cdot \bar{d}_z (1-\bar{d}_z),
\end{align}
where $R_{\delta \sim e \mid Z=z}^2$ is the $R$-square statistic of regressing $\delta$ against $e$ within arm $z$. To see this, by definition, $R_{\delta \sim e \mid Z=z}^2 = \bV[\bE\{\delta_i \mid e(X_i),Z_i\!=\!z\} \mid Z_i\!=\!z] / \bV[\delta_i \mid Z_i\!=\!z]$. Furthermore, total expectation gives, 
\begin{align*}
    & \bE[\delta_i \mid e(X_i),Z_i=z] = \bE[ \bE\{\delta_i \mid X_i, Z_i=z\} \mid e(X_i),Z_i=z] = \bE[ d_z(X_i)\mid e(X_i),Z_i=z], \\
    & \bE[\delta_i \mid Z_i=z] = \bE[ \bE\{\delta_i \mid Z_i=z, X_i\} \mid Z_i=z] = \bE[d_z(X_i) \mid Z_i=z] = \bar{d}_z,
\end{align*}
where the first is from $e(X_i)$ being a function of $X_i$ so the smaller $\sigma$-algebra $\sigma(e(X_i),Z_i)$ is being conditioned on. Therefore, the numerator and denominator of $R_{\delta \sim e \mid Z=z}^2$ are
\begin{align*}
    & \bV[\bE\{\delta_i \mid e(X_i),Z_i=z\} \mid Z_i=z] = \bV[\bE\{ d_z(X_i)\mid e(X_i),Z_i=z\} \mid Z_i=z], \\
    & \bV[\delta_i \mid Z_i=z] = \bE[\delta_i^2 \mid Z_i=z] - \bE[\delta_i \mid Z_i=z]^2 = \bar{d}_z (1-\bar{d}_z),
\end{align*}
which establishes equation \eqref{eqapp:event-rate-diff-CauchySchwarz}. Then, plugging in $\bV[w_i\mid Z_i=z] = \kappa_z-1$ from Section \ref{secapp:design-effect} establishes $\left|d_z^w-\bar{d}_z \right|^2 \leq (\kappa_z-1) R_{\delta \sim e \mid Z=z}^2 \, \bar{d}_z(1-\bar{d}_z)$. On the relative scale, dividing both sides by $(\bar{d}_z)^2$ and taking the root yields
\begin{align*}
    \left| d_z^w/\bar{d}_z - 1 \right| & \leq (\kappa_z-1)^{1/2} \left\{\frac{\bV[\bE\{\delta_i \mid e(X_i),Z_i=z\} \mid Z_i=z]}{\bar{d}_z (1-\bar{d}_z)}\bar{d}_z (1-\bar{d}_z) \right\}^{1/2} \Big/ \bar{d}_z \\
    & = (\kappa_z-1)^{1/2} \left\{\bV[\bE\{\delta_i\mid e(X_i),Z_i=z\}\mid Z_i=z]\right\}^{1/2}/\bar{d}_z = (\kappa_z-1)^{1/2} c_z.
\end{align*}

Finally, $\kappa_z \to 1$ as $\phi \to 1$ (the proof of Proposition \ref{propapp:3}), so $d_z^w-\bar{d}_z \to 0$. For ATT weights, $h=e$ and $m=\bE[e(X)]=r$, so $w_1=\{e/e\}/\{r/r\}=1$, and thus $\kappa_1=1$; hence, $d_1^w = \bar{d}_1$, which completes the proof.
\end{proof}

\newpage

\section{Computational details} \label{secapp:calculation}

\subsection{\texorpdfstring{$\Delta$}{Delta} under complete follow-up}\label{secapp:delta-complete-followup}

Under H1, by equation \eqref{eqapp:sk-simplified} and $F_1^w=(F_0^w)^{\exp(\tau_0)}$,
\begin{equation*}
    \pi(\tau_0,t) = \frac{r_1 \exp(\tau_0) F_1^w(t)G_1(t)}{r_1 \exp(\tau_0) F_1^w(t)G_1(t) + r_0 F_0^w(t)G_0(t)},
\end{equation*}
which generally depends on $G_z$ and $F_z^w$ for $z=0,1$, but under $G_1=G_0=1$ simplifies to 
\begin{equation*}
    \pi(\tau_0,t) = \frac{r\theta \, \{F_0^w(t)\}^{\theta-1} }{r \theta \, \{F_0^w(t)\}^{\theta-1} + 1-r},
\end{equation*}
with $\theta=\exp(\tau_0)$. At any fixed $\tau_0$, $\pi(\tau_0,t)$ is a function of $F_0^w(t)$ only; let $u=F_0^w(t)$, then
\begin{equation*}
    \pi(\tau_0,t) = \pi(u) = \frac{r\theta\,u^{\theta-1}}{r\theta\,u^{\theta-1}+1-r}, \qquad u\in[F_0^w(t^\dagger),1],
\end{equation*}
which is monotone in $u$.
Denote $\gamma:= F_0^w(t^\dagger)$. Similar to Section \ref{secapp:marcox-model}, taking $k=1$ and $f=\mathbb{I}_{[0,t]}$ in equation \eqref{eqapp:weighted-event-moments} gives $\mu_1$, and substrating it from $\mu$ gives $\mu_0$, 
\begin{equation*}
    d\mu_0=(1-r)f_0^w\,dt=-(1-r)\,du, \quad d\mu_1=rf_1^w\,dt=-r\theta u^{\theta-1}\,du,
\end{equation*}
so the $\mu_z$-moments of $\pi$ are integrals over $[\gamma, 1]$,
\begin{align*}
    \bE_{\mu_0}[\pi] &= \frac{1}{d_0^w}\int_\gamma^1 \pi(u)\,du, &\qquad
    \bV_{\mu_0}[\pi] &= \frac{1}{d_0^w}\int_\gamma^1 \left\{\pi(u)-\bE_{\mu_0}[\pi]\right\}^2du, \\
    \bE_{\mu_1}[\pi] &= \frac{1}{d_1^w}\int_\gamma^1 \pi(u)\,\theta u^{\theta-1}\,du, &\qquad
    \bV_{\mu_1}[\pi] &= \frac{1}{d_1^w}\int_\gamma^1 \left\{\pi(u)-\bE_{\mu_1}[\pi]\right\}^2\theta u^{\theta-1}\,du,
\end{align*}
where the negative signs in $d\mu_z$ cancel with reversing the integration limits from $1$ to $\gamma$ to $\gamma$ to $1$. Evaluating these integrals numerically gives $\Delta_1=\bV_{\mu_1}[\pi]/(1-\bE_{\mu_1}[\pi])^2$ and $\Delta_0=\bV_{\mu_0}[\pi]/(\bE_{\mu_0}[\pi])^2$, and thus 
\begin{equation*}
    \Delta = \frac{1}{\widetilde{V}_w}\left[\frac{\kappa_1\Delta_1}{rd_1^w}+\frac{\kappa_0\Delta_0}{(1-r)d_0^w}\right],
\end{equation*}
from the proof of Proposition \ref{prop:working-error}, with $\kappa_1=\kappa_0=1$ in a randomized trial. Table \ref{tabapp:rct_delta_noloss} reports results for randomized trials with $G_1=G_0\equiv1$.

\begin{table}[ht]
    \centering
    \small
    \setlength{\tabcolsep}{3pt}
    \renewcommand{\arraystretch}{0.96}
    \caption{Percentage underestimation of the asymptotic variance $V_{\RCT}$ from omitting $\Delta$ in randomized trials without loss to follow-up. Entries are $\Delta/(1+\Delta)$ expressed as percentages (\%), with the upper bound from from inequality \eqref{eqapp:delta-bound-tight} in parentheses.}
    \label{tabapp:rct_delta_noloss}
    \begin{tabular}{@{}cccc@{}}
    \toprule
    \multirow{2}{*}{\shortstack[c]{Marginal \\ hazard ratio}} & \multicolumn{3}{c}{Control event probability $d_0$} \\
    \cmidrule(lr){2-4}
     & \makebox[0.15\textwidth][c]{0.2} & \makebox[0.15\textwidth][c]{0.5} & \makebox[0.15\textwidth][c]{0.8} \\
    \midrule
    \multicolumn{4}{c}{\textit{Treatment proportion $= 1/3$}} \\
    0.2 & 0.02 (0.07) & 0.26 (0.79) & 1.73 (5.33) \\
    0.4 & 0.02 (0.06) & 0.22 (0.68) & 1.32 (4.09) \\
    0.6 & 0.01 (0.04) & 0.12 (0.36) & 0.65 (2.07) \\
    0.8 & 0.00 (0.01) & 0.03 (0.10) & 0.17 (0.55) \\
    \midrule
    \multicolumn{4}{c}{\textit{Treatment proportion $= 1/2$}} \\
    0.2 & 0.04 (0.12) & 0.42 (1.26) & 2.55 (7.71) \\
    0.4 & 0.03 (0.09) & 0.31 (0.95) & 1.71 (5.28) \\
    0.6 & 0.02 (0.05) & 0.15 (0.46) & 0.78 (2.51) \\
    0.8 & 0.00 (0.01) & 0.04 (0.12) & 0.19 (0.64) \\
    \midrule
    \multicolumn{4}{c}{\textit{Treatment proportion $= 2/3$}} \\
    0.2 & 0.06 (0.17) & 0.57 (1.71) & 3.14 (9.42) \\
    0.4 & 0.04 (0.11) & 0.36 (1.07) & 1.78 (5.59) \\
    0.6 & 0.02 (0.05) & 0.15 (0.47) & 0.75 (2.44) \\
    0.8 & 0.00 (0.01) & 0.04 (0.11) & 0.18 (0.59) \\
    \bottomrule
    \end{tabular}
\end{table}

\subsection{Overlap coefficient from fitted propensity scores}\label{secapp:phi-estimation}

Let $p$ be the density of $e(X)$ and $p_z$ the conditional one given $Z=z$. Since $\Pr\{Z=1\mid e(X)\!=\!u\}=u$, by Bayes' rule, $p_1(u)=u\,p(u)/r$ and $p_0(u)=(1-u)\,p(u)/(1-r)$, and hence
\begin{equation*}
    \phi = \int_0^1\min\{p_1(u),p_0(u)\}\,du 
    = \int_0^1 p(u) \min\left\{ \frac{u}{r}, \frac{1-u}{1-r}\right\}\,du 
    = \bE\!\left[ \min\left\{\frac{e(X)}{r},\frac{1-e(X)}{1-r}\right\} \right].
\end{equation*}
Similar to \cite{liu2025sample}, we estimate $\phi$ from fitted propensity scores $\hat e_1,\ldots,\hat e_n$ non-parametrically,
\begin{equation*}
    \widehat{\phi}=\frac{1}{n}\sum_{i=1}^n\min\left\{\frac{\hat{e}_i}{\hat{r}},\frac{1-\hat{e}_i}{1-\hat{r}}\right\}, \qquad \hat{r}=\frac{1}{n}\sum_{i=1}^n\hat{e}_i.
\end{equation*}
For the right heart catheterization cohort, $\hat{e}_i$ are the fitted logistic propensity scores. For the synthetic superpopulation of $M$ units, we replace $\hat{e}_i$ by the true $e(X_i)$.

\subsection{\texorpdfstring{$\Delta$}{Delta} and \texorpdfstring{$\epsilon$}{epsilon} estimated from simulated and real-world data}\label{secapp:delta-epsilon-data}

We estimate the population (asymptotic) quantities $\Delta$ and $\epsilon$ from a large dataset with $M$ units with any censoring. Let $n_z=\sum_{i=1}^M \mathbb{I}(Z_i\!=\!z)$ and $\hat{r}=n_1/M$. We estimate the following population quantities by
\begin{equation*}
    \hat{d}_z^w = \frac{\sum_{i:Z_i=z}w_i\delta_i}{\sum_{i:Z_i=z}w_i}, 
    \qquad
    \hat{\kappa}_z = 1 + \frac{1}{n_z-1} \sum_{i:Z_i=z}(w_i-1)^2,
    \qquad
    \widehat{\widetilde{V}}_w=\frac{\hat{\kappa}_1}{\hat{r}\hat{d}_1^w}+\frac{\hat{\kappa}_0}{(1-\hat{r})\hat{d}_0^w},
\end{equation*}
where $\hat{\kappa}_z$ follows from $w_i$ being normalized to mean one within each arm $z$. Further let $\widehat{\tau}_M$ the weighted Cox partial-likelihood estimate, and $\widehat{V}_w$ be the robust estimate of the asymptotic variance $V_w$, both from this dataset. Then, with $\pi_M^*(\widehat{\tau}_M,t)=S_1^*(\widehat{\tau}_M,t)/S_0^*(\widehat{\tau}_M,t)$, we estimate $\bE_{\mu_z}[\pi]$ and $\bV_{\mu_z}[\pi]$ by 
\begin{equation*}
    \bar{\pi}_z=\frac{\sum_{i:Z_i=z} w_i\delta_i\,\pi_M^*(\widehat{\tau}_M,T_i)}{\sum_{i:Z_i=z} w_i\delta_i}, \qquad
    s_z^2=\frac{\sum_{i:Z_i=z} w_i\delta_i\{\pi_M^*(\widehat{\tau}_M,T_i)-\bar{\pi}_z\}^2}{\sum_{i:Z_i=z} w_i\delta_i},
\end{equation*}
the sample mean and variance of $\pi_M^*(\widehat{\tau}_M,T_i)$ over the event times of arm $z$ weighted by $w_i$, respectively. Therefore, $\widehat{\Delta}_1=s_1^2/(1-\bar\pi_1)^2$ and $\widehat{\Delta}_0=s_0^2/\bar\pi_0^2$, and hence
\begin{equation*}
    \widehat{\Delta}=\frac{1}{\widehat{\widetilde{V}}_w}\left[\frac{\hat{\kappa}_1\widehat{\Delta}_1}{\hat{r}\hat{d}_1^w}+\frac{\hat{\kappa}_0\widehat{\Delta}_0}{(1-\hat{r})\hat{d}_0^w}\right], \qquad
    \widehat{\epsilon}=\widehat{V}_w-\widehat{\widetilde{V}}_w(1+\widehat{\Delta}).
\end{equation*}

For synthetic data, we use all $M=1,000,000$ units in the simulated superpopulation, and use the true $e(X_i)$ to calculate the weight $w_i$. For the real-world right heart catheterization data, we use the trimmed cohort of $M=5,393$ patients and estimate $w_i$ with $\hat{e}_i$ from the main-effects logistic propensity score model.

\subsection{Coefficient of variation of weights under the Beta approximation}\label{secapp:cv-weights-beta}

Each $\sigma_z$ in Proposition \ref{propapp:3}, the coefficient of variation of $w_z(X)$, is determined by $(r,\phi)$ under $e(X)\sim\mathrm{Beta}(a,b)$ through ratios of Beta moments. For inverse probability weights, $h\equiv1$ and $P^w=P$,
\begin{align*}
    \sigma_1^2 &
    = \frac{\bE[(e^{-1})^2] - \bE[e^{-1}]^2}{\bE[e^{-1}]^2}
    = \frac{\bE[e^{-2}]}{\bE[e^{-1}]^2}-1 \\
    & = \frac{(a+b-1)(a+b-2)/\{(a-1)(a-2)\}}{(a+b-1)^2/(a-1)^2}-1 \\
    & = \frac{(a+b-2)(a-1)}{(a-2)(a+b-1)}-1
    = \frac{b}{(a-2)(a+b-1)}.
\end{align*}
For overlap weights, $h=e(1-e)$ and $dP^w \propto e(1-e)\,dP$, so $e \sim \mathrm{Beta}(a+1,b+1)$ under $P^w$, and $h/e=1-e$; hence
\begin{equation*}
    \sigma_1^2 = \frac{\bV^w[1-e]}{\bE^w[1-e]^2} = \frac{(a+1)(b+1)/\{(a+b+2)^2(a+b+3)\}}{(b+1)^2/(a+b+2)^2} = \frac{a+1}{(b+1)(a+b+3)}.
\end{equation*}
For ATT weights, $\sigma_1^2=0$ since $w_1 \equiv1$. For $\sigma_0^2$, $h=e$, so $e \sim \mathrm{Beta}(a+1,b)$ under $P^w$, and $h/e\equiv1$ and $h/(1-e)=e/(1-e)$. For $Y \sim \mathrm{Beta}(\alpha,\beta)$, $\bE[Y/(1-Y)]=\alpha/(\beta-1)$ and $\bE[Y^2/(1-Y)^2]=\alpha(\alpha+1)/\{(\beta-1)(\beta-2)\}$, which with $(\alpha,\beta)=(a+1,b)$ give
\begin{equation*}
    \sigma_0^2 = \frac{\alpha(\alpha+1)/\{(\beta-1)(\beta-2)\}}{\alpha^2/(\beta-1)^2}-1 = \frac{\alpha+\beta-1}{\alpha(\beta-2)}=\frac{a+b}{(a+1)(b-2)}.
\end{equation*}
The remaining $\sigma_0^2$ follow by exchanging $(a,e)$ with $(b,1-e)$:
\begin{center}\normalfont\small
\setlength{\tabcolsep}{6pt}
\begin{tabular}{cccc}
\toprule
weight & $\sigma_1^2$ & $\sigma_0^2$ & $\sigma_z$ finite if \\
\midrule
inverse probability & $\dfrac{b}{(a-2)(a+b-1)}$ & $\dfrac{a}{(b-2)(a+b-1)}$ & $a>2,\;b>2$\\[12pt]
overlap & $\dfrac{a+1}{(b+1)(a+b+3)}$ & $\dfrac{b+1}{(a+1)(a+b+3)}$ & $a>0,\;b>0$\\[12pt]
ATT & $0$ & $\dfrac{a+b}{(a+1)(b-2)}$ & $b>2$\\
\bottomrule
\end{tabular}
\end{center}

\newpage

\section{Additional simulation results}\label{secapp:additional-simulations}

\subsection{Randomized trials}\label{secapp:simu-rct}

\vspace{-12pt}
\begin{table}[H]
    \centering
    \footnotesize
    \setlength{\tabcolsep}{3pt}
    \renewcommand{\arraystretch}{0.96}
    \caption{Sample size and empirical coverage of 95\% confidence intervals for the proposed, Schoenfeld, and Freedman formulas in randomized trials with 10\% and 30\% loss to follow-up in the control and treated arms, respectively. The Monte Carlo error is approximately 0.003.}
    \label{tabapp:rct_coverage}
    \begin{tabular*}{0.85\textwidth}{@{\extracolsep{\fill}}ccccccc@{}}
    \toprule
    \multirow{2}{*}[-4pt]{\shortstack[c]{Marginal \\ hazard ratio}} & \multicolumn{2}{c}{Proposed} & \multicolumn{2}{c}{Schoenfeld} & \multicolumn{2}{c}{Freedman} \\
    \cmidrule(lr){2-3} \cmidrule(lr){4-5} \cmidrule(lr){6-7}
     & $N$ & Coverage & $N$ & Coverage & $N$ & Coverage \\
    \midrule
    \multicolumn{7}{c}{\textit{Treatment proportion $= 1/3$}} \\
    0.4 & 103 & .951 & 73 & .954 & 117 & .948 \\
    0.6 & 268 & .953 & 221 & .951 & 288 & .950 \\
    0.8 & 1249 & .944 & 1114 & .948 & 1271 & .948 \\
    \midrule
    \multicolumn{7}{c}{\textit{Treatment proportion $= 1/2$}} \\
    0.4 & 82 & .951 & 72 & .952 & 91 & .950 \\
    0.6 & 222 & .950 & 212 & .950 & 239 & .950 \\
    0.8 & 1059 & .946 & 1037 & .947 & 1114 & .944 \\
    \midrule
    \multicolumn{7}{c}{\textit{Treatment proportion $= 2/3$}} \\
    0.4 & 82 & .952 & 92 & .947 & 86 & .949 \\
    0.6 & 232 & .951 & 257 & .950 & 247 & .951 \\
    0.8 & 1130 & .943 & 1222 & .950 & 1231 & .952 \\
    \bottomrule
    \end{tabular*}
\end{table}

\begin{table}[H]
    \centering
    \footnotesize
    \setlength{\tabcolsep}{1.5pt}
    \renewcommand{\arraystretch}{0.96}
    \caption{Sample size, empirical power, and coverage of 95\% confidence intervals for the proposed, Schoenfeld, and Freedman formulas in randomized trials with 20\% loss to follow-up in each arm. The Monte Carlo errors are about 0.006 for power and 0.003 for coverage.}
    \label{tabapp:rct_results_loss20}
    \begin{tabular*}{\textwidth}{@{\extracolsep{\fill}}cccccccccc@{}}
    \toprule
    \multirow{2}{*}[-4pt]{\shortstack[c]{Marginal \\ hazard ratio}} & \multicolumn{3}{c}{Proposed} & \multicolumn{3}{c}{Schoenfeld} & \multicolumn{3}{c}{Freedman} \\
    \cmidrule(lr){2-4} \cmidrule(lr){5-7} \cmidrule(lr){8-10}
     & $N$ & Power & Coverage & $N$ & Power & Coverage & $N$ & Power & Coverage \\
    \midrule
    \multicolumn{10}{c}{\textit{Treatment proportion $= 1/3$}} \\
    0.4 & 98 & .842 & .955 & 77 & .750 & .951 & 122 & .916 & .953 \\
    0.6 & 256 & .812 & .950 & 231 & .760 & .948 & 300 & .871 & .952 \\
    0.8 & 1195 & .808 & .945 & 1152 & .795 & .944 & 1321 & .853 & .948 \\
    \midrule
    \multicolumn{10}{c}{\textit{Treatment proportion $= 1/2$}} \\
    0.4 & 80 & .818 & .948 & 75 & .796 & .950 & 91 & .867 & .950 \\
    0.6 & 218 & .805 & .949 & 214 & .804 & .952 & 238 & .843 & .954 \\
    0.8 & 1038 & .792 & .947 & 1035 & .790 & .947 & 1106 & .818 & .947 \\
    \midrule
    \multicolumn{10}{c}{\textit{Treatment proportion $= 2/3$}} \\
    0.4 & 82 & .792 & .949 & 92 & .827 & .949 & 83 & .798 & .949 \\
    0.6 & 234 & .803 & .949 & 252 & .823 & .951 & 236 & .803 & .951 \\
    0.8 & 1144 & .797 & .951 & 1181 & .810 & .952 & 1172 & .809 & .950 \\
    \bottomrule
    \end{tabular*}
\end{table}

\begin{table}[ht]
    \centering
    \small
    \setlength{\tabcolsep}{3pt}
    \renewcommand{\arraystretch}{0.96}
    \caption{Empirical type I error of the Wald test at sample sizes determined by the proposed formula in randomized trials with 10\% and 30\% loss to follow-up in the control and treated arms, respectively. Type I error is evaluated at a marginal hazard ratio of 1.}
    \label{tabapp:rct_type1}
    \begin{tabular}{@{}cccc@{}}
    \toprule
    \multirow{2}{*}{\shortstack[c]{Marginal hazard ratio \\ used to plan $N$}} & \multicolumn{3}{c}{Type I error by treatment proportion} \\
    \cmidrule(lr){2-4}
     & \makebox[0.15\textwidth][c]{$1/3$} & \makebox[0.15\textwidth][c]{$1/2$} & \makebox[0.15\textwidth][c]{$2/3$} \\
    \midrule
    0.4 & .020 & .025 & .030 \\
    0.6 & .025 & .027 & .029 \\
    0.8 & .022 & .024 & .027 \\
    \bottomrule
    \end{tabular}
\end{table}

\subsection{Observational studies}\label{secapp:simu-obs}

\subsubsection{Power under additional scenarios}\label{secapp:simu-obs-other-powers}

The primary observational scenario presented in the main text uses a marginal hazard ratio $\exp(\tau_0^w)=0.6$, confounding parameters $R^2_{\psi\sim X}\!=\! \rho \!=\!0.5$, a logit link for the propensity score model, and $10\%$ and $30\%$ loss to follow-up in the control and treated arms, respectively. Each of the following tables change one factor from the primary scenario. Each table reports the empirical power that uses estimated propensity scores $\hat{e}_i$ in each replicate, with power using the true $e_i$ in parentheses. The Monte Carlo error is approximately 0.006.

\begin{table}[H]
    \centering
    \footnotesize
    \setlength{\tabcolsep}{3pt}
    \renewcommand{\arraystretch}{0.96}
    \caption{Sample size and empirical power for ATE, ATT, and ATO when the marginal hazard ratio is changed from the primary scenario to $\exp(\tau_0^w)=0.8$.}
    \label{tabapp:obs_power_hr08}
    \begin{tabular*}{\textwidth}{@{\extracolsep{\fill}}ccccccccc@{}}
    \toprule
    & \multicolumn{4}{c}{ATE} & \multicolumn{2}{c}{\multirow{2}{*}{ATT}} & \multicolumn{2}{c}{\multirow{2}{*}{ATO}} \\
    \cmidrule(lr){2-5}
    & \multicolumn{2}{c}{Proposed} & \multicolumn{2}{c}{Hsieh--Lavori} & \multicolumn{2}{c}{} & \multicolumn{2}{c}{} \\
    \cmidrule(lr){2-3} \cmidrule(lr){4-5} \cmidrule(lr){6-7} \cmidrule(lr){8-9}
    $\phi$ & $N$ & Power & $N$ & Power & $N$ & Power & $N$ & Power \\
    \midrule
    \multicolumn{9}{c}{\textit{Treatment proportion $= 1/3$}} \\
    0.95 & 1275 & .853 (.805) & 1135 & .805 (.746) & 1274 & .865 (.818) & 1274 & .851 (.796) \\
    0.90 & 1275 & .844 (.802) & 1127 & .788 (.746) & 1271 & .861 (.815) & 1264 & .852 (.800) \\
    0.85 & 1328 & .854 (.801) & 1157 & .792 (.754) & 1307 & .850 (.805) & 1296 & .861 (.820) \\
    0.80 & 1392 & .855 (.804) & 1182 & .778 (.733) & 1347 & .850 (.805) & 1324 & .861 (.811) \\
    0.75 & 1510 & .854 (.803) & 1225 & .759 (.721) & 1409 & .859 (.817) & 1371 & .860 (.800) \\
    0.70 & 1718 & .875 (.826) & 1282 & .742 (.704) & 1525 & .863 (.810) & 1433 & .864 (.812) \\
    0.65 & 2117 & .899 (.847) & 1332 & .689 (.648) & 1676 & .856 (.803) & 1498 & .855 (.803) \\
    \midrule
    \multicolumn{9}{c}{\textit{Treatment proportion $= 1/2$}} \\
    0.95 & 1066 & .849 (.795) & 1044 & .839 (.784) & 1072 & .855 (.802) & 1066 & .848 (.801) \\
    0.90 & 1071 & .846 (.807) & 1048 & .841 (.799) & 1087 & .849 (.806) & 1070 & .845 (.798) \\
    0.85 & 1104 & .855 (.802) & 1078 & .846 (.791) & 1138 & .851 (.815) & 1103 & .848 (.795) \\
    0.80 & 1138 & .853 (.794) & 1105 & .839 (.784) & 1197 & .850 (.802) & 1133 & .850 (.809) \\
    0.75 & 1200 & .847 (.802) & 1147 & .832 (.787) & 1300 & .851 (.809) & 1177 & .856 (.811) \\
    0.70 & 1287 & .852 (.803) & 1191 & .822 (.776) & 1445 & .853 (.821) & 1224 & .845 (.796) \\
    0.65 & 1471 & .871 (.818) & 1273 & .810 (.760) & 1739 & .868 (.829) & 1304 & .864 (.809) \\
    \bottomrule
    \end{tabular*}
\end{table}

\begin{table}[H]
    \centering
    \footnotesize
    \setlength{\tabcolsep}{3pt}
    \renewcommand{\arraystretch}{0.96}
    \caption{Sample size and empirical power for ATE, ATT, and ATO when loss to follow-up is changed from the primary scenario to $20\%$ in each arm.}
    \label{tabapp:obs_power_loss20}
    \begin{tabular*}{\textwidth}{@{\extracolsep{\fill}}ccccccccc@{}}
    \toprule
    & \multicolumn{4}{c}{ATE} & \multicolumn{2}{c}{\multirow{2}{*}{ATT}} & \multicolumn{2}{c}{\multirow{2}{*}{ATO}} \\
    \cmidrule(lr){2-5}
    & \multicolumn{2}{c}{Proposed} & \multicolumn{2}{c}{Hsieh--Lavori} & \multicolumn{2}{c}{} & \multicolumn{2}{c}{} \\
    \cmidrule(lr){2-3} \cmidrule(lr){4-5} \cmidrule(lr){6-7} \cmidrule(lr){8-9}
    $\phi$ & $N$ & Power & $N$ & Power & $N$ & Power & $N$ & Power \\
    \multicolumn{9}{c}{\textit{Treatment proportion $= 1/3$}} \\
    0.95 & 259 & .851 (.816) & 233 & .804 (.775) & 259 & .848 (.813) & 258 & .847 (.805) \\
    0.90 & 262 & .854 (.814) & 234 & .799 (.771) & 261 & .854 (.818) & 260 & .864 (.821) \\
    0.85 & 271 & .855 (.822) & 239 & .802 (.771) & 268 & .852 (.818) & 264 & .859 (.815) \\
    0.80 & 284 & .856 (.829) & 244 & .793 (.773) & 276 & .850 (.819) & 270 & .859 (.815) \\
    0.75 & 308 & .860 (.831) & 253 & .773 (.749) & 290 & .846 (.817) & 279 & .872 (.820) \\
    0.70 & 349 & .881 (.850) & 264 & .756 (.737) & 312 & .844 (.823) & 291 & .865 (.823) \\
    0.65 & 431 & .912 (.880) & 275 & .724 (.705) & 344 & .844 (.811) & 304 & .859 (.815) \\
    \midrule
    \multicolumn{9}{c}{\textit{Treatment proportion $= 1/2$}} \\
    0.95 & 220 & .843 (.805) & 216 & .836 (.797) & 220 & .827 (.798) & 220 & .849 (.806) \\
    0.90 & 221 & .851 (.813) & 217 & .841 (.807) & 224 & .824 (.806) & 221 & .845 (.800) \\
    0.85 & 227 & .848 (.813) & 222 & .839 (.801) & 234 & .840 (.808) & 226 & .853 (.811) \\
    0.80 & 235 & .852 (.820) & 229 & .842 (.810) & 247 & .840 (.806) & 233 & .862 (.819) \\
    0.75 & 247 & .855 (.816) & 237 & .838 (.797) & 268 & .832 (.805) & 242 & .848 (.806) \\
    0.70 & 266 & .851 (.817) & 247 & .820 (.789) & 300 & .833 (.815) & 252 & .857 (.826) \\
    0.65 & 301 & .865 (.830) & 261 & .810 (.776) & 359 & .862 (.844) & 266 & .851 (.799) \\
    \bottomrule
    \end{tabular*}
\end{table}

\begin{table}[H]
    \centering
    \footnotesize
    \setlength{\tabcolsep}{3pt}
    \renewcommand{\arraystretch}{0.96}
    \caption{Sample size and empirical power for ATE, ATT, and ATO when the confounding parameter $\rho$ is changed from the primary scenario to $\rho=0.25$.}
    \label{tabapp:obs_power_rho25}
    \begin{tabular*}{\textwidth}{@{\extracolsep{\fill}}ccccccccc@{}}
    \toprule
    & \multicolumn{4}{c}{ATE} & \multicolumn{2}{c}{\multirow{2}{*}{ATT}} & \multicolumn{2}{c}{\multirow{2}{*}{ATO}} \\
    \cmidrule(lr){2-5}
    & \multicolumn{2}{c}{Proposed} & \multicolumn{2}{c}{Hsieh--Lavori} & \multicolumn{2}{c}{} & \multicolumn{2}{c}{} \\
    \cmidrule(lr){2-3} \cmidrule(lr){4-5} \cmidrule(lr){6-7} \cmidrule(lr){8-9}
    $\phi$ & $N$ & Power & $N$ & Power & $N$ & Power & $N$ & Power \\
    \midrule
    \multicolumn{9}{c}{\textit{Treatment proportion $= 1/3$}} \\
    0.95 & 270 & .858 (.824) & 223 & .765 (.747) & 270 & .852 (.821) & 269 & .857 (.816) \\
    0.90 & 273 & .849 (.809) & 224 & .746 (.722) & 271 & .854 (.819) & 270 & .864 (.812) \\
    0.85 & 283 & .855 (.819) & 228 & .748 (.726) & 277 & .855 (.812) & 275 & .868 (.819) \\
    0.80 & 297 & .850 (.830) & 233 & .737 (.727) & 285 & .854 (.815) & 281 & .852 (.810) \\
    0.75 & 323 & .856 (.826) & 242 & .721 (.707) & 298 & .848 (.822) & 290 & .865 (.820) \\
    0.70 & 367 & .872 (.845) & 252 & .697 (.690) & 318 & .842 (.807) & 301 & .861 (.819) \\
    0.65 & 455 & .911 (.889) & 262 & .675 (.666) & 346 & .836 (.814) & 315 & .866 (.825) \\
    \midrule
    \multicolumn{9}{c}{\textit{Treatment proportion $= 1/2$}} \\
    0.95 & 222 & .850 (.802) & 212 & .827 (.789) & 223 & .846 (.806) & 222 & .846 (.816) \\
    0.90 & 225 & .848 (.807) & 214 & .832 (.788) & 227 & .832 (.808) & 225 & .845 (.802) \\
    0.85 & 230 & .851 (.816) & 218 & .827 (.792) & 235 & .841 (.818) & 230 & .858 (.818) \\
    0.80 & 238 & .849 (.810) & 224 & .818 (.783) & 247 & .831 (.809) & 236 & .857 (.816) \\
    0.75 & 251 & .844 (.814) & 233 & .813 (.785) & 267 & .821 (.811) & 245 & .856 (.807) \\
    0.70 & 270 & .855 (.819) & 243 & .800 (.776) & 296 & .836 (.819) & 256 & .859 (.819) \\
    0.65 & 305 & .868 (.840) & 257 & .801 (.780) & 351 & .853 (.835) & 271 & .856 (.805) \\
    \bottomrule
    \end{tabular*}
\end{table}

\begin{table}[H]
    \centering
    \footnotesize
    \setlength{\tabcolsep}{3pt}
    \renewcommand{\arraystretch}{0.96}
    \caption{Sample size and empirical power for ATE, ATT, and ATO when confounding parameters are changed from the primary scenario to $R^2_{\psi\sim X}=\rho=0.25$.}
    \label{tabapp:obs_power_r2rho25}
    \begin{tabular*}{\textwidth}{@{\extracolsep{\fill}}ccccccccc@{}}
    \toprule
    & \multicolumn{4}{c}{ATE} & \multicolumn{2}{c}{\multirow{2}{*}{ATT}} & \multicolumn{2}{c}{\multirow{2}{*}{ATO}} \\
    \cmidrule(lr){2-5}
    & \multicolumn{2}{c}{Proposed} & \multicolumn{2}{c}{Hsieh--Lavori} & \multicolumn{2}{c}{} & \multicolumn{2}{c}{} \\
    \cmidrule(lr){2-3} \cmidrule(lr){4-5} \cmidrule(lr){6-7} \cmidrule(lr){8-9}
    $\phi$ & $N$ & Power & $N$ & Power & $N$ & Power & $N$ & Power \\
    \midrule
    \multicolumn{9}{c}{\textit{Treatment proportion $= 1/3$}} \\
    0.95 & 270 & .826 (.820) & 223 & .733 (.734) & 270 & .819 (.803) & 270 & .837 (.819) \\
    0.90 & 274 & .820 (.810) & 224 & .725 (.727) & 271 & .828 (.816) & 271 & .840 (.826) \\
    0.85 & 283 & .833 (.831) & 228 & .728 (.742) & 277 & .836 (.830) & 275 & .831 (.818) \\
    0.80 & 296 & .818 (.818) & 233 & .719 (.722) & 285 & .823 (.812) & 280 & .824 (.813) \\
    0.75 & 322 & .826 (.825) & 241 & .699 (.711) & 298 & .818 (.811) & 289 & .829 (.810) \\
    0.70 & 366 & .854 (.852) & 251 & .688 (.700) & 317 & .829 (.824) & 300 & .828 (.815) \\
    0.65 & 454 & .890 (.882) & 261 & .657 (.664) & 346 & .811 (.815) & 314 & .830 (.808) \\
    \midrule
    \multicolumn{9}{c}{\textit{Treatment proportion $= 1/2$}} \\
    0.95 & 223 & .817 (.806) & 212 & .796 (.789) & 224 & .812 (.809) & 223 & .828 (.815) \\
    0.90 & 225 & .825 (.820) & 214 & .805 (.803) & 227 & .800 (.799) & 225 & .822 (.813) \\
    0.85 & 231 & .821 (.811) & 219 & .804 (.794) & 236 & .798 (.803) & 230 & .820 (.802) \\
    0.80 & 238 & .817 (.809) & 225 & .791 (.783) & 247 & .802 (.806) & 236 & .821 (.806) \\
    0.75 & 251 & .821 (.813) & 233 & .783 (.781) & 267 & .811 (.812) & 246 & .820 (.808) \\
    0.70 & 270 & .821 (.818) & 243 & .776 (.773) & 297 & .806 (.810) & 257 & .827 (.810) \\
    0.65 & 305 & .844 (.837) & 256 & .772 (.771) & 351 & .818 (.823) & 271 & .825 (.816) \\
    \bottomrule
    \end{tabular*}
\end{table}

\begin{table}[H]
    \centering
    \footnotesize
    \setlength{\tabcolsep}{3pt}
    \renewcommand{\arraystretch}{0.96}
    \caption{Sample size and empirical power for ATE, ATT, and ATO when the propensity-score link is changed from logit to the (reflected) Aranda--Ordaz link with $\lambda=5$.}
    \label{tabapp:obs_power_ao3}
    \begin{tabular*}{\textwidth}{@{\extracolsep{\fill}}ccccccccc@{}}
    \toprule
    & \multicolumn{4}{c}{ATE} & \multicolumn{2}{c}{\multirow{2}{*}{ATT}} & \multicolumn{2}{c}{\multirow{2}{*}{ATO}} \\
    \cmidrule(lr){2-5}
    & \multicolumn{2}{c}{Proposed} & \multicolumn{2}{c}{Hsieh--Lavori} & \multicolumn{2}{c}{} & \multicolumn{2}{c}{} \\
    \cmidrule(lr){2-3} \cmidrule(lr){4-5} \cmidrule(lr){6-7} \cmidrule(lr){8-9}
    $\phi$ & $N$ & Power & $N$ & Power & $N$ & Power & $N$ & Power \\
    \midrule
    \multicolumn{9}{c}{\textit{Treatment proportion $= 1/3$}} \\
    0.95 & 271 & .849 (.819) & 224 & .758 (.737) & 271 & .855 (.817) & 271 & .862 (.818) \\
    0.90 & 274 & .854 (.818) & 224 & .759 (.739) & 271 & .862 (.826) & 271 & .859 (.817) \\
    0.85 & 283 & .860 (.825) & 228 & .747 (.729) & 277 & .863 (.824) & 275 & .866 (.821) \\
    0.80 & 298 & .857 (.819) & 234 & .748 (.726) & 285 & .852 (.819) & 281 & .872 (.824) \\
    0.75 & 325 & .869 (.839) & 243 & .734 (.720) & 300 & .845 (.813) & 291 & .864 (.818) \\
    0.70 & 368 & .885 (.854) & 252 & .703 (.693) & 317 & .825 (.792) & 302 & .865 (.817) \\
    0.65 & 457 & .912 (.891) & 263 & .683 (.669) & 342 & .782 (.750) & 316 & .858 (.800) \\
    \midrule
    \multicolumn{9}{c}{\textit{Treatment proportion $= 1/2$}} \\
    0.95 & 224 & .850 (.814) & 213 & .829 (.790) & 224 & .834 (.805) & 224 & .850 (.804) \\
    0.90 & 226 & .854 (.815) & 215 & .837 (.796) & 228 & .835 (.809) & 226 & .852 (.806) \\
    0.85 & 231 & .850 (.815) & 219 & .822 (.796) & 236 & .836 (.820) & 231 & .863 (.814) \\
    0.80 & 239 & .857 (.819) & 225 & .834 (.794) & 248 & .830 (.813) & 237 & .849 (.806) \\
    0.75 & 251 & .855 (.816) & 234 & .817 (.794) & 267 & .811 (.796) & 245 & .855 (.817) \\
    0.70 & 270 & .849 (.811) & 243 & .799 (.772) & 294 & .800 (.773) & 256 & .846 (.810) \\
    0.65 & 301 & .863 (.817) & 253 & .785 (.751) & 341 & .796 (.760) & 270 & .855 (.818) \\
    \bottomrule
    \end{tabular*}
\end{table}

\subsubsection{Coverage and Type I error}

\vspace{-12pt}
\begin{table}[H]
    \centering
    \footnotesize
    \setlength{\tabcolsep}{3pt}
    \renewcommand{\arraystretch}{0.99}
    \caption{\small Empirical coverage of 95\% confidence intervals at the sample sizes determined by the proposed methods. The Scenario column identifies the primary scenario or the factor of other scenarios changed from it. Results use estimated propensity scores in each replicate, with results using true propensity scores in parentheses. The Monte Carlo error is approximately 0.003.}
    \label{tabapp:obs_coverage}
    \setlength{\aboverulesep}{0.30ex}
    \setlength{\belowrulesep}{0.45ex}
    \begin{tabular*}{0.8\textwidth}{@{\extracolsep{\fill}}p{0.18\textwidth}cccc@{}}
    \toprule
    Scenario & $\phi$ & ATE & ATT & ATO \\
    \midrule
    \multirow{8}{=}{\textit{Primary scenario}} & \multicolumn{4}{c}{\textit{Treatment proportion $= 1/3$}} \\
    & 0.95 & .980 (.950) & .982 (.951) & .985 (.954) \\
    & 0.80 & .981 (.952) & .983 (.953) & .982 (.951) \\
    & 0.65 & .979 (.945) & .978 (.943) & .983 (.950) \\
    & \multicolumn{4}{c}{\textit{Treatment proportion $= 1/2$}} \\
    & 0.95 & .982 (.947) & .981 (.949) & .986 (.950) \\
    & 0.80 & .982 (.947) & .978 (.945) & .985 (.948) \\
    & 0.65 & .983 (.947) & .978 (.948) & .981 (.945) \\
    \midrule
    \multirow{8}{=}{$\exp(\tau_0^w)=0.8$} & \multicolumn{4}{c}{\textit{Treatment proportion $= 1/3$}} \\
    & 0.95 & .980 (.946) & .987 (.947) & .986 (.951) \\
    & 0.80 & .985 (.948) & .983 (.950) & .986 (.948) \\
    & 0.65 & .982 (.944) & .983 (.949) & .982 (.950) \\
    & \multicolumn{4}{c}{\textit{Treatment proportion $= 1/2$}} \\
    & 0.95 & .988 (.951) & .986 (.946) & .987 (.945) \\
    & 0.80 & .987 (.950) & .984 (.946) & .987 (.954) \\
    & 0.65 & .982 (.945) & .977 (.942) & .985 (.947) \\
    \midrule
    \multirow{8}{=}{\textit{\shortstack[l]{$20\%$ loss to\\follow-up in each arm}}} & \multicolumn{4}{c}{\textit{Treatment proportion $= 1/3$}} \\
    & 0.95 & .985 (.949) & .983 (.944) & .985 (.952) \\
    & 0.80 & .984 (.947) & .981 (.949) & .983 (.951) \\
    & 0.65 & .980 (.943) & .982 (.954) & .982 (.954) \\
    & \multicolumn{4}{c}{\textit{Treatment proportion $= 1/2$}} \\
    & 0.95 & .981 (.945) & .983 (.951) & .983 (.951) \\
    & 0.80 & .987 (.951) & .981 (.952) & .984 (.952) \\
    & 0.65 & .980 (.948) & .974 (.942) & .981 (.946) \\
    \midrule
    \multirow{8}{=}{$\rho=0.25$} & \multicolumn{4}{c}{\textit{Treatment proportion $= 1/3$}} \\
    & 0.95 & .982 (.947) & .984 (.950) & .986 (.950) \\
    & 0.80 & .982 (.949) & .983 (.943) & .984 (.946) \\
    & 0.65 & .977 (.944) & .978 (.949) & .984 (.950) \\
    & \multicolumn{4}{c}{\textit{Treatment proportion $= 1/2$}} \\
    & 0.95 & .984 (.951) & .983 (.950) & .986 (.953) \\
    & 0.80 & .984 (.949) & .980 (.941) & .988 (.954) \\
    & 0.65 & .978 (.942) & .976 (.948) & .983 (.947) \\
    \midrule
    \multirow{8}{=}{$R^2_{\psi\sim X}=\rho=0.25$} & \multicolumn{4}{c}{\textit{Treatment proportion $= 1/3$}} \\
    & 0.95 & .968 (.953) & .968 (.947) & .971 (.955) \\
    & 0.80 & .960 (.946) & .966 (.946) & .965 (.942) \\
    & 0.65 & .958 (.941) & .970 (.947) & .971 (.952) \\
    & \multicolumn{4}{c}{\textit{Treatment proportion $= 1/2$}} \\
    & 0.95 & .968 (.947) & .969 (.957) & .967 (.948) \\
    & 0.80 & .970 (.951) & .963 (.946) & .967 (.948) \\
    & 0.65 & .964 (.947) & .960 (.949) & .973 (.954) \\
    \midrule
    \multirow{8}{=}{\textit{\shortstack[l]{Aranda--Ordaz link\\($\lambda=5$)}}} & \multicolumn{4}{c}{\textit{Treatment proportion $= 1/3$}} \\
    & 0.95 & .984 (.949) & .982 (.949) & .983 (.949) \\
    & 0.80 & .982 (.945) & .985 (.958) & .982 (.947) \\
    & 0.65 & .975 (.944) & .971 (.933) & .980 (.945) \\
    & \multicolumn{4}{c}{\textit{Treatment proportion $= 1/2$}} \\
    & 0.95 & .985 (.952) & .985 (.948) & .982 (.949) \\
    & 0.80 & .984 (.953) & .982 (.947) & .985 (.945) \\
    & 0.65 & .984 (.946) & .962 (.931) & .979 (.950) \\
    \bottomrule
    \end{tabular*}
\end{table}

\begin{table}[H]
    \centering
    \footnotesize
    \setlength{\tabcolsep}{3pt}
    \renewcommand{\arraystretch}{0.99}
    \caption{Empirical type I errors under the primary scenario at the sample sizes from the proposed method. Results use estimated propensity scores in each replicate, with results using true propensity scores in parentheses. The Monte Carlo error is approximately 0.002.}
    \label{tabapp:obs_type1}
    \setlength{\aboverulesep}{0.25ex}
    \setlength{\belowrulesep}{0.35ex}
    \begin{tabular*}{0.55\textwidth}{@{\extracolsep{\fill}}cccc@{}}
    \toprule
    $\phi$ & ATE & ATT & ATO \\
    \midrule
    \multicolumn{4}{c}{\textit{Treatment proportion $= 1/3$}} \\
    0.95 & .006 (.021) & .005 (.025) & .006 (.020) \\
    0.90 & .007 (.025) & .005 (.025) & .007 (.023) \\
    0.85 & .006 (.023) & .005 (.023) & .007 (.021) \\
    0.80 & .006 (.025) & .009 (.023) & .008 (.025) \\
    0.75 & .005 (.022) & .007 (.024) & .004 (.023) \\
    0.70 & .007 (.020) & .008 (.026) & .006 (.023) \\
    0.65 & .006 (.020) & .007 (.021) & .005 (.024) \\
    \midrule
    \multicolumn{4}{c}{\textit{Treatment proportion $= 1/2$}} \\
    0.95 & .007 (.025) & .005 (.024) & .006 (.029) \\
    0.90 & .006 (.023) & .009 (.030) & .007 (.020) \\
    0.85 & .006 (.026) & .011 (.027) & .008 (.026) \\
    0.80 & .006 (.026) & .010 (.024) & .007 (.023) \\
    0.75 & .008 (.027) & .009 (.030) & .007 (.025) \\
    0.70 & .009 (.025) & .010 (.026) & .009 (.027) \\
    0.65 & .007 (.022) & .010 (.026) & .008 (.024) \\
    \bottomrule
    \end{tabular*}
\end{table}

\subsubsection{Tails of the true propensity scores at the Aranda--Ordaz link}

\vspace{-12pt}
\begin{table}[H]
    \centering
    \footnotesize
    \setlength{\tabcolsep}{3pt}
    \renewcommand{\arraystretch}{0.96}
    \caption{Proportions (per mille, \textperthousand) of the true propensity scores in the indicated ranges at overlap $\phi=0.65$ of different links. AO is the Aranda--Ordaz link.}
    \label{tabapp:obs_AO_tails}
    \begin{tabular*}{0.9\textwidth}{@{\extracolsep{\fill}}ccccc@{\hspace{1.5em}}ccc@{}}
    \toprule
    \multirow{2}{*}{\shortstack{Treatment\\proportion}} & \multirow{2}{*}{Link} & \multicolumn{2}{c}{Range of $e$} & \multirow{2}{*}{Sum} & \multicolumn{2}{c}{Range of $e$} & \multirow{2}{*}{Sum} \\
    \cmidrule(lr){3-4} \cmidrule(lr){6-7}
     & & $[0,0.02]$ & $[0.98,1.00]$ & & $[0,0.01]$ & $[0.99,1.00]$ & \\
    \midrule
    \multirow{3}{*}{$1/3$} & Logit & 0.85 & 0.00 & 0.85 & 0.06 & 0.00 & 0.06 \\
     & AO ($\lambda=5$) & 0.01 & 2.96 & 2.97 & 0.00 & 1.40 & 1.40 \\
    \midrule
    \multirow{3}{*}{$1/2$} & Logit & 0.03 & 0.03 & 0.06 & 0.00 & 0.00 & 0.00 \\
     & AO ($\lambda=5$) & 0.00 & 8.22 & 8.22 & 0.00 & 3.52 & 3.52 \\
    \bottomrule
    \end{tabular*}
\end{table}

\subsubsection{Diagnosis of \texorpdfstring{$\Delta$}{Delta} and \texorpdfstring{$\epsilon$}{epsilon}}

\vspace{-12pt}
\begin{table}[H]
    \centering
    \small
    \setlength{\tabcolsep}{3pt}
    \renewcommand{\arraystretch}{0.83}
    \caption{\small Superpopulation values of $\Delta/(1+\Delta)$ and $|\epsilon|/(V-\epsilon)$ across observational simulation settings, expressed as percentages (\%). Calculations follow Section \ref{secapp:delta-epsilon-data}.}
    \label{tabapp:obs_approx}
    \footnotesize
    \begin{tabular*}{1.00\textwidth}{@{\extracolsep{\fill}}p{0.18\textwidth}ccccccc@{}}
    \toprule
     & & \multicolumn{2}{c}{ATE} & \multicolumn{2}{c}{ATT} & \multicolumn{2}{c}{ATO} \\
    \cmidrule(lr){3-4} \cmidrule(lr){5-6} \cmidrule(lr){7-8}
    Scenario & $\phi$ & $\Delta/(1+\Delta)$ & $|\epsilon|/(V-\epsilon)$ & $\Delta/(1+\Delta)$ & $|\epsilon|/(V-\epsilon)$ & $\Delta/(1+\Delta)$ & $|\epsilon|/(V-\epsilon)$ \\
    \midrule
    \multirow{8}{=}{\textit{Primary scenario}} & \multicolumn{7}{c}{\textit{Treatment proportion $= 1/3$}} \\
    & 0.95 & 0.07 & 0.62 & 0.07 & 0.58 & 0.07 & 0.63 \\
    & 0.80 & 0.07 & 1.99 & 0.08 & 2.04 & 0.07 & 1.98 \\
    & 0.65 & 0.06 & 3.17 & 0.10 & 4.24 & 0.07 & 2.88 \\
     & \multicolumn{7}{c}{\textit{Treatment proportion $= 1/2$}} \\
    & 0.95 & 0.09 & 0.62 & 0.09 & 0.76 & 0.09 & 0.64 \\
    & 0.80 & 0.09 & 2.33 & 0.09 & 2.98 & 0.09 & 2.24 \\
    & 0.65 & 0.09 & 4.13 & 0.10 & 5.85 & 0.09 & 3.13 \\
    \midrule
    \multirow{8}{=}{$\exp(\tau_0^w)=0.8$} & \multicolumn{7}{c}{\textit{Treatment proportion $= 1/3$}} \\
    & 0.95 & 0.02 & 0.06 & 0.02 & 0.55 & 0.02 & 0.33 \\
    & 0.80 & 0.02 & 0.74 & 0.03 & 1.72 & 0.02 & 0.81 \\
    & 0.65 & 0.02 & 3.41 & 0.04 & 3.18 & 0.03 & 1.09 \\
    & \multicolumn{7}{c}{\textit{Treatment proportion $= 1/2$}} \\
    & 0.95 & 0.03 & 0.20 & 0.03 & 0.66 & 0.03 & 0.22 \\
    & 0.80 & 0.03 & 0.42 & 0.03 & 2.51 & 0.03 & 0.77 \\
    & 0.65 & 0.03 & 0.13 & 0.03 & 4.61 & 0.03 & 1.11 \\
    \midrule
    \multirow{8}{=}{\textit{\shortstack[l]{$20\%$ loss to\\follow-up in each arm}}} & \multicolumn{7}{c}{\textit{Treatment proportion $= 1/3$}} \\
    & 0.95 & 0.61 & 0.61 & 0.61 & 0.69 & 0.61 & 0.66 \\
    & 0.80 & 0.59 & 1.99 & 0.68 & 2.60 & 0.63 & 2.20 \\
    & 0.65 & 0.53 & 2.98 & 0.84 & 5.40 & 0.65 & 3.18 \\
     & \multicolumn{7}{c}{\textit{Treatment proportion $= 1/2$}} \\
    & 0.95 & 0.74 & 0.68 & 0.74 & 0.87 & 0.74 & 0.69 \\
    & 0.80 & 0.74 & 2.60 & 0.76 & 3.49 & 0.74 & 2.45 \\
    & 0.65 & 0.73 & 4.71 & 0.80 & 6.79 & 0.73 & 3.43 \\
    \midrule
    \multirow{8}{=}{$\rho=0.25$} & \multicolumn{7}{c}{\textit{Treatment proportion $= 1/3$}} \\
    & 0.95 & 0.07 & 0.37 & 0.07 & 0.34 & 0.07 & 0.35 \\
    & 0.80 & 0.07 & 1.14 & 0.08 & 1.16 & 0.07 & 1.05 \\
    & 0.65 & 0.06 & 1.91 & 0.10 & 2.69 & 0.07 & 1.53 \\
     & \multicolumn{7}{c}{\textit{Treatment proportion $= 1/2$}} \\
    & 0.95 & 0.09 & 0.33 & 0.09 & 0.40 & 0.09 & 0.33 \\
    & 0.80 & 0.09 & 1.31 & 0.09 & 1.68 & 0.09 & 1.14 \\
    & 0.65 & 0.09 & 2.63 & 0.10 & 3.34 & 0.09 & 1.58 \\
    \midrule
    \multirow{8}{=}{$R^2_{\psi\sim X}=\rho=0.25$} & \multicolumn{7}{c}{\textit{Treatment proportion $= 1/3$}} \\
    & 0.95 & 0.07 & 0.33 & 0.07 & 0.31 & 0.07 & 0.33 \\
    & 0.80 & 0.07 & 0.90 & 0.08 & 0.91 & 0.07 & 0.81 \\
    & 0.65 & 0.06 & 1.50 & 0.10 & 1.91 & 0.07 & 1.15 \\
     & \multicolumn{7}{c}{\textit{Treatment proportion $= 1/2$}} \\
    & 0.95 & 0.09 & 0.25 & 0.09 & 0.30 & 0.09 & 0.25 \\
    & 0.80 & 0.09 & 0.97 & 0.09 & 1.24 & 0.09 & 0.83 \\
    & 0.65 & 0.09 & 1.84 & 0.10 & 2.43 & 0.09 & 1.14 \\
    \midrule
    \multirow{8}{=}{\textit{\shortstack[l]{Aranda--Ordaz link\\($\lambda=5$)}}} & \multicolumn{7}{c}{\textit{Treatment proportion $= 1/3$}} \\
    & 0.95 & 0.07 & 0.63 & 0.07 & 0.58 & 0.07 & 0.63 \\
    & 0.80 & 0.07 & 2.02 & 0.08 & 2.37 & 0.07 & 2.02 \\
    & 0.65 & 0.08 & 3.25 & 0.16 & 5.04 & 0.07 & 2.82 \\
     & \multicolumn{7}{c}{\textit{Treatment proportion $= 1/2$}} \\
    & 0.95 & 0.09 & 0.63 & 0.09 & 0.76 & 0.09 & 0.64 \\
    & 0.80 & 0.09 & 2.40 & 0.09 & 3.38 & 0.09 & 2.21 \\
    & 0.65 & 0.11 & 2.09 & 0.13 & 3.03 & 0.09 & 2.93 \\
    \bottomrule
    \end{tabular*}
\end{table}

\begin{table}[ht]
    \centering
    \small
    \setlength{\tabcolsep}{3pt}
    \renewcommand{\arraystretch}{0.96}
    \caption{Superpopulation values of $|\epsilon|/(V-\epsilon)$ at two high levels of confounding, expressed as percentages (\%). All other factors follow the primary scenario.}
    \label{tabapp:obs_approx_high_confoundings}
    \begin{tabular*}{0.8\textwidth}{@{\extracolsep{\fill}}ccccccc@{}}
    \toprule
     & \multicolumn{3}{c}{$R^2_{\psi\sim X}=\rho=0.75$} & \multicolumn{3}{c}{$R^2_{\psi\sim X}=\rho=0.90$} \\
    \cmidrule(lr){2-4} \cmidrule(lr){5-7}
    $\phi$ & ATE & ATT & ATO & ATE & ATT & ATO \\
    \midrule
    \multicolumn{7}{c}{\textit{Treatment proportion $= 1/3$}} \\
    0.95 & 0.91 & 0.81 & 0.94 & 1.19 & 1.02 & 1.23 \\
    0.80 & 2.87 & 3.15 & 3.51 & 3.11 & 3.62 & 4.66 \\
    0.65 & 3.43 & 5.75 & 5.34 & 2.54 & 5.60 & 7.31 \\
    \midrule
    \multicolumn{7}{c}{\textit{Treatment proportion $= 1/2$}} \\
    0.95 & 1.05 & 1.28 & 1.12 & 1.34 & 1.65 & 1.48 \\
    0.80 & 3.45 & 4.62 & 4.10 & 3.88 & 5.36 & 5.48 \\
    0.65 & 4.80 & 7.58 & 5.94 & 4.02 & 7.25 & 8.23 \\
    \bottomrule
    \end{tabular*}
\end{table}

\subsection{Real-world right heart catheterization data}\label{secapp:simu-rhc}

\begin{table}[H]
    \centering
    \small
    \setlength{\tabcolsep}{3pt}
    \renewcommand{\arraystretch}{0.96}
    \caption{Estimated $\Delta/(1+\Delta)$ and $|\epsilon|/(V-\epsilon)$ in the trimmed right heart catheterization cohort, expressed as percentages (\%). Calculations follow Section \ref{secapp:delta-epsilon-data}.}
    \label{tabapp:rhc_approx}
    \begin{tabular*}{0.6\textwidth}{@{\extracolsep{\fill}}ccc@{}}
    \toprule
    Estimand & $\hat{\Delta}/(1+\hat{\Delta})$ & $|\hat{\epsilon}|/(\hat V-\hat{\epsilon})$ \\
    \midrule
    ATE & 0.04 & 2.02 \\
    ATT & 0.04 & 4.01 \\
    ATO & 0.05 & 0.51 \\
    \bottomrule
    \end{tabular*}
\end{table}

\begin{table}[H]
    \centering
    \small
    \setlength{\tabcolsep}{3pt}
    \renewcommand{\arraystretch}{0.96}
    \caption{Standardized and observed event rates (\%) of the two arms in the trimmed right heart catheterization cohort. Calculations of standardized rates follow Section \ref{secapp:delta-epsilon-data}.}
    \label{tabapp:rhc_event_rates}
    \begin{tabular*}{0.6\textwidth}{@{\extracolsep{\fill}}lcc@{}}
    \toprule
    Event rate & Treated & Control \\
    \midrule
    Standardized, ATE & 53.15 & 47.24 \\
    Standardized, ATT & 54.07 & 48.79 \\
    Standardized, ATO & 53.35 & 46.81 \\
    Observed & 54.07 & 46.18 \\
    \bottomrule
    \end{tabular*}
\end{table}

\newpage

\spacingset{1.35}

\end{document}